\documentclass{article}

\usepackage{dilab_arxiv}

\DeclareFontShape{T1}{cmr}{m}{n}{<->ssub * lmr/m/n}{}
\DeclareFontShape{T1}{cmr}{m}{it}{<->ssub * lmr/m/it}{}
\DeclareFontShape{T1}{cmr}{m}{sl}{<->ssub * lmr/m/sl}{}
\DeclareFontShape{T1}{cmr}{m}{sc}{<->ssub * lmr/m/sc}{}
\DeclareFontShape{T1}{cmr}{b}{n}{<->ssub * lmr/b/n}{}
\DeclareFontShape{T1}{cmr}{b}{it}{<->ssub * lmr/b/it}{}
\DeclareFontShape{T1}{cmr}{b}{sl}{<->ssub * lmr/b/sl}{}
\DeclareFontShape{T1}{cmr}{bx}{n}{<->ssub * lmr/bx/n}{}
\DeclareFontShape{T1}{cmr}{bx}{it}{<->ssub * lmr/bx/it}{}
\DeclareFontShape{T1}{cmr}{bx}{sl}{<->ssub * lmr/bx/sl}{}
\DeclareFontFamily{T1}{cmss}{}
\DeclareFontShape{T1}{cmss}{m}{n}{<->ssub * lmss/m/n}{}
\DeclareFontShape{T1}{cmss}{m}{it}{<->ssub * lmss/m/it}{}
\DeclareFontShape{T1}{cmss}{m}{sl}{<->ssub * lmss/m/sl}{}
\DeclareFontShape{T1}{cmss}{b}{n}{<->ssub * lmss/b/n}{}
\DeclareFontShape{T1}{cmss}{bx}{n}{<->ssub * lmss/bx/n}{}
\DeclareFontShape{T1}{cmss}{bx}{it}{<->ssub * lmss/bx/it}{}
\DeclareFontShape{T1}{cmss}{bx}{sl}{<->ssub * lmss/bx/sl}{}

\usepackage{array}
\usepackage{enumitem}      
\usepackage{algorithm}     
\usepackage{algpseudocode}
\usepackage{pifont}
\newcommand{\rev}[1]{#1}
\newenvironment{revblock}{}{}
\newcommand{\cmark}{\ding{51}}
\newcommand{\xmark}{\ding{55}}

\definecolor{ncBlue}{HTML}{0072B2}   
\definecolor{ncAmber}{HTML}{E69F00}  
\definecolor{ncGreen}{HTML}{009E73}  
\definecolor{ncGray}{HTML}{444444}

\newtheorem{theorem}{Theorem}
\newtheorem{proposition}{Proposition}
\newtheorem{lemma}{Lemma}
\newtheorem{corollary}{Corollary}
\newtheorem{assumption}{Assumption}

\theoremstyle{definition}
\newtheorem{definition}{Definition}

\theoremstyle{remark}
\newtheorem{remark}{Remark}

\newcommand{\R}{\mathbb{R}}

\newcommand{\eps}{\varepsilon}
\newcommand{\inner}[2]{\langle #1, #2 \rangle}
\newcommand{\norm}[1]{\lVert #1 \rVert}
\newcommand{\pos}[1]{\left[#1\right]_{+}}          
\newcommand{\rt}{\tilde{r}}                        
\newcommand{\xt}{\tilde{x}}                        
\newcommand{\yt}{\tilde{y}}
\newcommand{\Nt}{N}                                

\newcommand{\etainf}{\eta_{\infty}}
\newcommand{\gapf}{\mathrm{gap}}                   
\newcommand{\dist}{\mathrm{dist}}
\newcommand{\supp}{\mathrm{supp}}
\newcommand{\diag}{\mathrm{diag}}
\newcommand{\spec}{\mathrm{spec}}
\newcommand{\Fix}{\mathrm{Fix}}
\newcommand{\Atil}{\tilde{A}}                      
\newcommand{\Bmat}{B}                              
\newcommand{\Ndot}{\dot{N}}                        
\newcommand{\Ndottot}{\dot{N}_{\mathrm{tot}}}
\newcommand{\Ms}{\mathcal{M}_S}                    
\newcommand{\sigmax}{\sigma_{\max}}
\newcommand{\sigmin}{\sigma_{\min}}
\newcommand{\simplex}[1]{\Delta_{#1}}
\newcommand{\ones}{\mathbf{1}}
\newcommand{\iregprm}{\textsf{IREG-PRM}$^+$}
\newcommand{\algname}{\iregprm{}}

\title{Why Last-Iterate Scale-Invariant Regret Matching Converges Linearly?}
\runningtitle{Last-Iterate Convergence in Scale-Invariant Regret Matching}
\date{\today}

\paperlogo{\includegraphics[height=1.5cm]{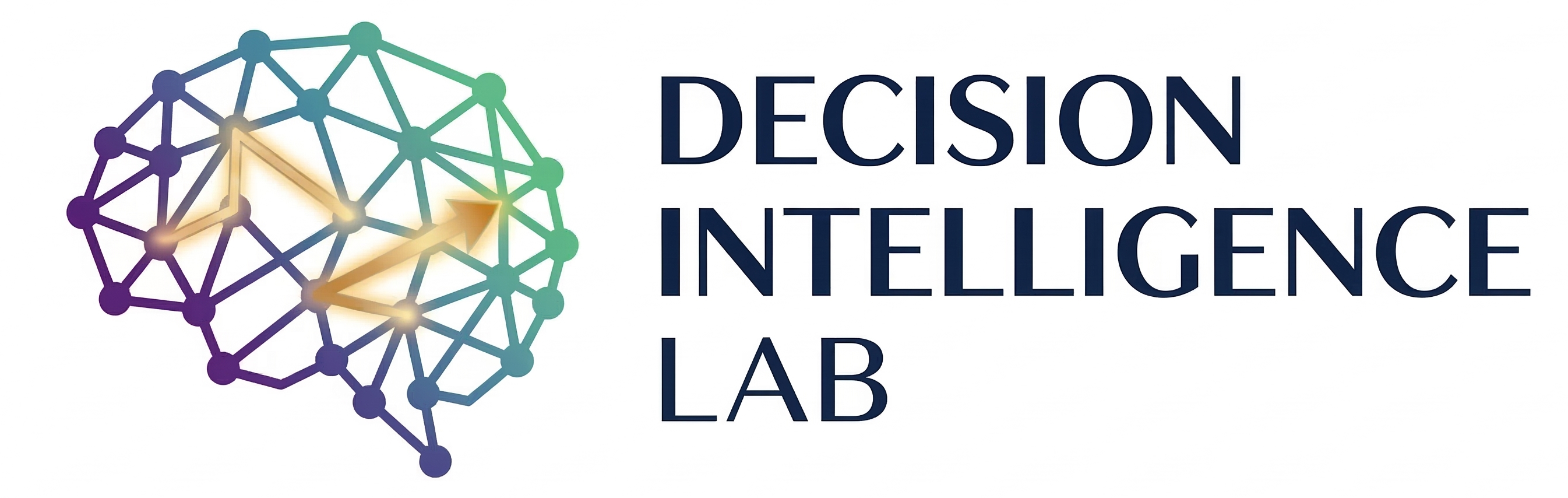}}

\author{
  Boning Li$^{1}$ and Longbo Huang$^{1\,\text{\faEnvelope}}$
  \\[0.3em]\normalfont
  $^1$Institute for Interdisciplinary Information Sciences, Tsinghua University
  \\
  \text{\faEnvelope}\ Correspondence: longbohuang@tsinghua.edu.cn
}

\begin{document}

\sloppy
\maketitle
\thispagestyle{fancy}

\begin{abstract}
Regret matching is the update rule at the core of counterfactual regret minimization. Its guarantees apply to the time-averaged strategy, while the last iterate (the current round's strategy) can fail to converge. Scale-invariant extragradient regret matching (\algname{}) normalizes the cumulative regret vector by its own norm and attains optimal regret without knowing the payoff scale. On two-player zero-sum matrix games, its last iterate converges linearly in experiments. Prior theory leaves this behavior unexplained. Analysis is difficult because the algorithm lacks a fixed step size. The step size is the inverse of the regret norm, which changes along the trajectory. Existing proofs of linear rates for regret-matching dynamics rely on restarts or modified updates. We explain the unmodified algorithm's convergence through \emph{norm saturation}. The regret norm grows to a finite limit, so the step size converges as well. We prove that saturation always occurs and provide an explicit upper bound. The last iterate then converges to a single Nash equilibrium in every two-player zero-sum matrix game, even with multiple equilibria. Near a unique equilibrium, the set of active actions freezes after one round, and the linearized dynamics have a closed form. Take the payoff matrix restricted to these active actions and centered by the game value. Whenever the limiting step size times this matrix's largest singular value is less than one, the last iterate converges linearly at a closed-form rate. Across $216$ random games, all $184$ instances with a numerically resolved limit satisfy this condition. The predicted rate matches the empirical rate with a median ratio of $1.000$. Our analysis also shows that the per-round increase in the squared regret norm scales with the squared Nash gap, which measures distance to equilibrium. In sequential games (extensive-form games), exact gap computation requires a costly best-response pass. On Kuhn and Leduc poker, the analogous norm increase tracks the Nash gap between infrequent passes. Our code is available at \url{https://github.com/lbn187/NormCert}.
\end{abstract}

\section{Introduction}\label{sec:intro}

Regret matching is the update rule inside counterfactual regret minimization, a standard method for sequential games with private information. This method solved heads-up limit hold'em and powers the strongest poker agents \citep{tammelin2014solving,bowling2015heads,brown2019solving}. Regret matching is parameter-free. Its theory rests on no-regret learning, where regret measures how much a learner's cumulative payoff falls short of the best fixed action in hindsight. Follow-the-regularized-leader and Blackwell approachability both connect to the algorithm \citep{zinkevich2003online,mcmahan2011follow,shalev2011online,abernethy2011blackwell}. Averaging its iterates yields the game value at a predictable rate \citep{hart2000simple,zinkevich2007regret,greenwald2006bounds,marden2007regret}. Full-information feedback suffices for its regret guarantee.

Every one of those guarantees applies to the \emph{average} of the strategies played so far. In a large game, maintaining this average requires storing and updating a second copy of the strategy. Analyzing the last iterate, the strategy played at the current step, is a separate problem \citep{lee2021last}. Even on a $3\times3$ matrix game, regret matching$^+$ and its predictive variant fail to converge in the last iterate \citep{cai2025lastiterate}. Existing linear last-iterate rates for regret-matching-type dynamics rely on restarts \citep{cai2025lastiterate,meng2025lastiterate} or payoff regularization \citep{pmlr-v139-perolat21a,liu2023the}. A lower bound shows that an algorithm retaining all past regrets admits no game-independent last-iterate rate \citep{cai2024fast}. Any explanation must therefore use game-dependent structure.

\algname{} normalizes the cumulative regret vector by its own norm and attains optimal regret without knowledge of the payoff scale \citep{zhang2025scale}. We consider it in the extra-gradient setup, a two-step update rule that evaluates payoffs at a lookahead point before stepping. Run as stated on zero-sum matrix games, it converges linearly in the last iterate. Its authors recorded this behavior and left it unexplained. The algorithm keeps its full regret memory, running without restarts, smoothing, or a step-size parameter.

Three features of the update block the standard arguments. Contraction arguments need a step size fixed in advance, whereas the step size here is a state variable equal to the inverse of a regret norm that the trajectory itself moves. Clipping at zero also makes the one-round map non-smooth. Finally, the fixed points form a cone with one ray per equilibrium because the regret norms are free to take any positive value. The map contracts toward a set, never toward a point.

We find the mechanism in the normalization. Scale invariance turns the norm of the cumulative regret vector into a control variable. This norm sets the step size of the round, which is the scalar governing the magnitude of each update. That step size determines how far the state can travel. If the norm increases to a finite limit, the step size converges and the dynamics settle. We call this event \emph{norm saturation}. We prove that it always occurs. Figure~\ref{fig:mechanism} illustrates these three stages on a single run.

Three steps turn saturation into rates. Saturation is equivalent to a finite second-order path length, the sum over rounds of the squared distance between the two strategies computed in a round. This equivalence yields an explicit bound on the limit of the norm. We next consider the neighborhood of a strictly complementary equilibrium, a solution where every unused pure strategy is strictly suboptimal. The clipping pattern locks after one round there. The non-smooth map becomes real-analytic on the support face, the set of states whose regret vanishes outside the equilibrium actions. On that face, the differentials of the value-centering shift cancel exactly. This cancellation leaves a closed-form Jacobian, the matrix of first-order partial derivatives of the map. In the local regime, the squared-norm increment computed by the algorithm in every round is proportional to the squared Nash gap, up to upper and lower multiplicative factors. This proportionality links an internal scalar of the algorithm directly to its progress.

\begin{figure}[!t]
  \centering
  \includegraphics[width=\textwidth]{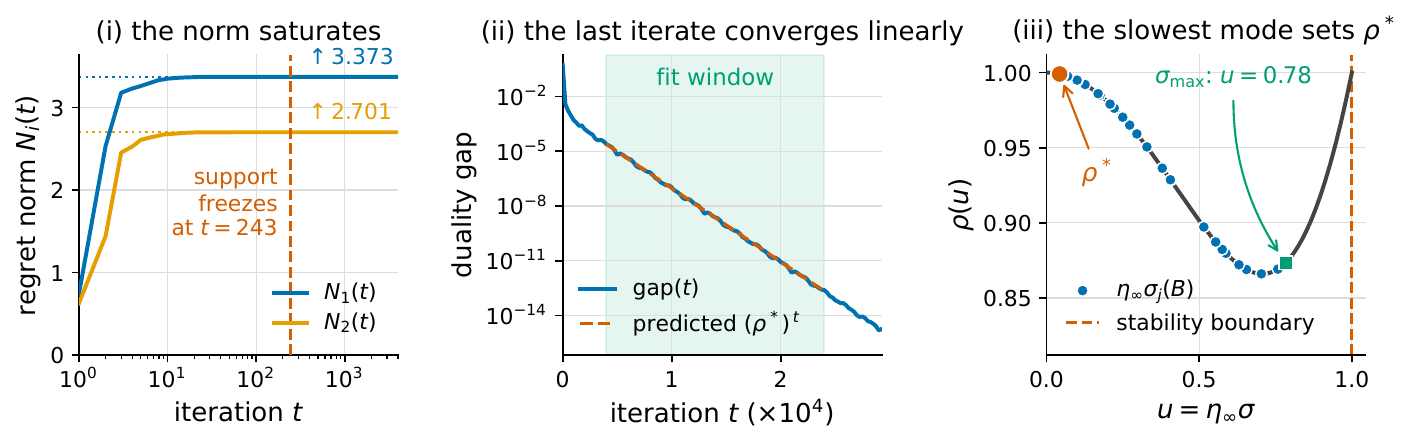}
  \caption{The mechanism on one run with $50$ actions per player.  \textbf{(i)} Both regret norms saturate.  \textbf{(ii)} The last-iterate gap then decays linearly, at the closed-form rate (dashed, computed from the converged state and not fitted).  \textbf{(iii)} The slowest mode sets the rate.}
  \label{fig:mechanism}
\end{figure}

We make four contributions:
\begin{enumerate}[leftmargin=2em, itemsep=1pt, topsep=2pt]
\item \textbf{We give the first proof that the last iterate of \algname{} converges to a single Nash equilibrium on every zero-sum matrix game.}  This holds without any uniqueness assumption.  Earlier results establish convergence to the set of equilibria \citep{zhang2025scale}.  \citet{nguyen2026convergencealternatingregretmatching} list the single-equilibrium question as unresolved for this dynamic.  In our analysis, each player's regret norm converges to a finite limit with an explicit bound (Theorem~\ref{thm:saturation}).  A monotone potential then prevents the strategies from drifting along the equilibrium set (Theorem~\ref{thm:point}).
\item \textbf{We prove the first closed-form linear last-iterate rate for a regret-matching dynamic with no restart, no smoothing and no step-size parameter.}  Near a unique equilibrium, one round identifies the set of actions in play.  Suppose the limiting step size lies below the inverse of the largest singular value of the centered payoff matrix on those actions.  The last iterate then contracts at a rate given explicitly by the singular values (Theorem~\ref{thm:rate}).  Limits violating this condition are unstable (Proposition~\ref{prop:selfstab}).
\item \textbf{We show that a scalar the algorithm already computes measures its progress.}  In the local regime, the per-round increase of the squared regret norm is proportional to the squared Nash gap.  The proportionality constant enters once without accumulating over rounds (Theorem~\ref{thm:certificate}).
\item \textbf{We confirm the theory on $216$ random matrix games and on poker.}  The closed-form spectrum matches numerical Jacobians, the matrices of first-order partial derivatives of the update map, to $2.4\times10^{-10}$.  The predicted rate matches the measured one at a median ratio of $1.000$, and every measurable limit satisfies the step-size condition.  On Kuhn and Leduc poker, the analogous scalar tracks the gap with log--log correlation $0.998$ and $0.997$.  The monitoring schedule uses $126\times$ and $64\times$ fewer best-response passes than per-iteration monitoring.
\end{enumerate}

\begin{table}[!t]
\centering
\caption{Last-iterate guarantees for regret-matching dynamics in zero-sum self-play.}
\label{tab:demarc}
\small
\setlength{\tabcolsep}{2pt}
\newcommand{\hd}[2]{\begin{tabular}[b]{@{}c@{}}#1\\#2\end{tabular}}
\begin{tabular*}{\textwidth}{@{\extracolsep{\fill}}lccccc@{}}
\toprule
dynamic & \hd{un-}{modified} & \hd{gap}{$\to0$} & \hd{proved}{linear rate} & \hd{closed-form}{rate} & \hd{certi-}{ficate} \\
\midrule
RM$^+$, PRM$^+$ \citep{tammelin2014solving,farina2021faster} & \cmark & \xmark$^{a}$ & \xmark & \xmark & \xmark \\
ExRM$^+$, SPRM$^+$ \citep{cai2025lastiterate} & \xmark & \cmark & \cmark$^{b}$ & \xmark & \xmark \\
SExRM$^+$, SPRM$^+$ \citep{meng2025lastiterate} & \xmark & \cmark & \cmark$^{b}$ & \xmark & \xmark \\
regularized FoReL \citep{pmlr-v139-perolat21a} & \xmark & \cmark & \xmark$^{c}$ & \xmark & \xmark \\
Reg-DOMD, Reg-CFR \citep{liu2023the} & \xmark & \cmark & \xmark$^{c}$ & \xmark & \xmark \\
\iregprm{} \citep{zhang2025scale} & \cmark & \cmark$^{d}$ & \xmark & \xmark & \xmark \\
\iregprm{} (this paper) & \cmark & \cmark$^{f}$ & \cmark$^{e}$ & \cmark & \cmark \\
\bottomrule
\end{tabular*}

\smallskip
{\footnotesize $^{a}$Converges to a strict NE; fails on a $3\times3$ game with a mixed NE.
$^{b}$For the restarted variants.
$^{c}$Linear to the equilibrium of the regularized game.
$^{d}$Stated as best-iterate $O(1/\sqrt{T})$; gap$\to0$, as convergence to the Nash set, follows from their results (Remark~\ref{rem:attrib}); linear rate observed, left open.
$^{e}$Local, under a unique equilibrium and $\etainf\sigmax(\Bmat)<1$ (Theorem~\ref{thm:rate}); the spectral condition is verifiable from the run given the equilibrium (Lemma~\ref{lem:epscert}).
$^{f}$To a single Nash equilibrium, without uniqueness (Theorem~\ref{thm:point}); for the alternating variants this is due to \citet{nguyen2026convergencealternatingregretmatching} (Remark~\ref{rem:single-ne}).
Unmodified: the update as published, with no added step-size parameter, smoothing, restart, or payoff regularization.  Gap$\to0$ and the proved linear rate refer to the Nash equilibrium of the original game.  Certificate: last-iterate progress read from the algorithm's own quantities.\par}
\end{table}

\section{Related Work}\label{sec:related}

\textbf{Regret matching and scale-invariant regret matching.}
Regret matching (RM, \citealp{hart2000simple}) and its tree-form extension,
counterfactual regret minimization (CFR, \citealp{zinkevich2007regret}), drive
equilibrium computation, through RM$^+$/CFR$^+$
\citep{tammelin2014solving,bowling2015heads,burch2019revisiting} and optimistic
variants \citep{farina2019stable,farina2019optimistic}.  Abstraction and pruning
refinements scale these methods to large extensive-form games
\citep{li2024rl,li2025efficient,li2026effective,li2026real,li2026abstraction,li2026gpu}.
Poker remains the driving application \citep{li2026pokerskill,li2026icm}.
The \emph{last-iterate}
was a standing gap: CFR diverges on rock--paper--scissors
\citep{lee2021last}, and RM$^+$ and its predictive variant PRM$^+$ can oscillate forever
\citep{farina2023regret}.  \citet{zhang2025scale} observed \iregprm{}'s linear
last-iterate convergence, which they left open.  Their bounded second-order path length already bounds the regret norm and
supplies the saturation equivalence of Lemma~\ref{lem:sat-equiv};
Theorem~\ref{thm:saturation} restates it with an explicit constant, and the
local theory below builds on that limit.
Concurrent parameter-free regret matching \citep{meng2026a} also uses a
nondecreasing regret norm to reach $O(1/T)$ average-iterate rates; in our
analysis the norm's \emph{limit} selects the local step size.
Analyses of regret-matching dynamics now reach potential games and constrained optimization \citep{anagnostides2025convergence}.
Two concurrent works study this same family of algorithms. In particular, \citet{sullins2026irprmpreviousutilitypredictionconstant} establish constant regret, an $O(1/T)$ average-iterate bound, an $O(T^{-1/2})$ best-iterate bound, and convergence to the set of Nash equilibria for IR-PRM$^+$ with previous-utility prediction under both simultaneous and alternating updates. Further, \citet{nguyen2026convergencealternatingregretmatching} formulate an alternating analogue of the RVU property and establish that the alternating variants Alt-IR-PRM$^+$ and Alt-HPRM$^+$ converge to a single Nash equilibrium without assuming uniqueness. As they observe, earlier set-convergence results, including ours for \iregprm{}, do not by themselves imply convergence to a single equilibrium, leaving this question unresolved. In Theorem~\ref{thm:point}, we resolve this question for the simultaneous extragradient dynamic by adapting their equilibrium-potential argument (see Remark~\ref{rem:single-ne}).

\textbf{Linear last-iterate convergence.}
Extragradient descent \citep{korpelevich1976extragradient} and its optimistic variants \citep{daskalakis2018limit,pmlr-v108-mokhtari20a} achieve linear last-iterate convergence on polytopes under saddle-point metric-subregularity \citep{wei2021linear}. They achieve the same rate in extensive-form games \citep{lee2021last}. For regret-matching-type dynamics, \citet{cai2025lastiterate} and \citet{meng2025lastiterate} prove linear last-iterate rates to Nash equilibria of the original game for smoothed variants with restarts. The regularized dynamics of \citet{pmlr-v139-perolat21a} and \citet{liu2023the} converge linearly to equilibria of the regularized game. In the original game, \citet{pmlr-v139-perolat21a} show asymptotic convergence, and \citet{liu2023the} establish a sublinear rate. The unmodified dynamics were open. Theorem~\ref{thm:rate} gives a conditional local rate for them, read off the frozen-support spectrum. This rate is game-dependent, as the lower bound of \citet{cai2024fast} requires.

The closest computable metric is the normalized duality gap of \citet{applegate2023faster}. This Lagrangian gap is computable in linear time. It upper-bounds Karush-Kuhn-Tucker (KKT) residuals and drives restarts, evaluating the current point through a dedicated gap computation. Our monitor uses the squared-norm growth that the algorithm already computes. Theorem~\ref{thm:certificate} converts this growth into a two-sided gap ratio on the last iterate. Appendix~\ref{app:discussion} compares error bounds, active sets, and computable metrics.
\section{Preliminaries}\label{sec:prelim}
\paragraph{Notation.} For a vector $v \in \R^n$, $\pos{v}$ denotes its entrywise positive part. We write $\norm{\cdot}$ and $\norm{\cdot}_1$ for the Euclidean and $\ell_1$ norms. The probability simplex is $\simplex{n}$, and $\ones$ is the all-ones vector. A nonnegative vector has support $\supp(x)$, and $\Fix(\Phi)$ denotes the fixed-point set of a map $\Phi$, meaning the set of inputs that the map sends to themselves. For a matrix $M$, $\sigmax(M)$ is its largest singular value and $\sigmin^{+}(M)$ is its smallest nonzero singular value.
\paragraph{Zero-sum matrix games.}
We consider a two-player zero-sum matrix game, a contest where one player's gain is the other player's loss, with payoff matrix $A \in \R^{n_1 \times n_2}$. Player~1 chooses a strategy $x \in \simplex{n_1}$ to maximize the expected payoff $x^\top A y$, while player~2 chooses a strategy $y \in \simplex{n_2}$ to minimize it. A strategy pair $(x^*, y^*)$ is a \emph{Nash equilibrium} (NE) when $x^\top A y^* \le (x^*)^\top A y^* \le (x^*)^\top A y$ for all $x, y$. By the minimax theorem (\citealp{von2007theory}), all NE share a common equilibrium payoff $v^*$, called the \emph{value}. The Nash (duality) gap sums both players' incentives to unilaterally deviate. Defined as $\gapf(x, y) := \max_{i \in [n_1]} (A y)_i - \min_{j \in [n_2]} (x^\top A)_j \ge 0$, this gap vanishes if and only if $(x, y)$ is an NE. For an NE $(x^*, y^*)$, the strategy supports are the sets of actions played with positive probability, written $S_1 := \supp(x^*)$ and $S_2 := \supp(y^*)$. Let $\Atil := A[S_1, S_2]$ be the support submatrix. The \emph{value-centered support matrix} $\Bmat := \Atil - v^* \ones\ones^\top$ governs the local spectrum. We define the \emph{support face} as the subspace $\Ms := \{(\rt_1, \rt_2) : \rt_1|_{S_1^c} = 0,\ \rt_2|_{S_2^c} = 0\}$. An NE is \emph{strictly complementary} if every unused action yields a strictly worse payoff against the equilibrium opponent than the game value:
\begin{equation}\label{eq:strict-comp} \delta \;:=\; \min\Big\{ \min_{i \notin S_1} \big( v^* - (A y^*)_i \big),\;\min_{j \notin S_2} \big( ((x^*)^\top A)_j - v^* \big) \Big\} \;>\; 0, \end{equation}
where the minimum over an empty set is $+\infty$, as occurs for a fully mixed player.
\paragraph{The \iregprm{} algorithm.}
Our object of study is \iregprm{} (Algorithm~\ref{alg:iregprm}), which runs one instance per player of IR-PRM$^+$~\citep[Algorithm~2]{zhang2025scale} in the \emph{extra-gradient} setup. We use its strict form. The shifted vector $r^{(t)} = \rt^{(t)} + m^{(t)} - \gamma^{(t)} \ones$ is unclipped, and the update is $\rt^{(t+1)} = \pos{r^{(t)} + g^{(t)}}$. An action's \emph{regret} is the payoff it would have earned against the opponent's strategy minus the player's actual payoff. Each \emph{extra-gradient} round first computes a tentative pre-iterate from predicted payoffs, and then updates at a main iterate using the payoffs observed there. Player~1 maintains a nonnegative cumulative regret vector $\rt_1^{(t)}$, initialized at $\mathbf{0}$ for scale invariance, alongside the \emph{pre-iterate} strategy $\xt^{(t)} = \rt_1^{(t)} / \norm{\rt_1^{(t)}}_1$. Symmetrically, player~2 maintains $\rt_2^{(t)}$ and $\yt^{(t)}$ using the payoff operator $-A^\top$. The regret norm is $\Nt_i(t):=\norm{\rt_i^{(t)}}_2$. Central to the update is the norm-preserving $\gamma$-shift:\begin{definition}[$\gamma$-shift]\label{def:gammashift}
Fix $v \in \R^n$ and a target $\tau > 0$. The mapping $\gamma \mapsto \norm{\pos{v - \gamma \ones}}_2$ is continuous, strictly decreasing on the region where it is positive, and identically zero for $\gamma\ge\max_i v_i$. Consequently, 
\begin{equation}\label{eq:gamma-shift} \big\lVert \pos{v - \gamma \ones} \big\rVert_2 \;=\; \tau \end{equation}
has a unique solution $\gamma(v, \tau)$ \citep[Appendix~B]{zhang2025scale}; invoked with $v = \rt^{(t)} + m^{(t)}$ and $\tau = \norm{\rt^{(t)}}_2$, it preserves the $\ell_2$ norm of the prediction step exactly.
\end{definition}\label{app:algorithm}Algorithm~\ref{alg:iregprm} details a single round of this strict extra-gradient update for the first player. Player~2 executes the identical procedure with the payoff operator $-A^\top$.

\begin{algorithm}[t]
\caption{\iregprm{} (strict Algorithm~2 of \citealp{zhang2025scale},
extra-gradient setup): round $t$, player~1's updates.}
\label{alg:iregprm}
\begin{algorithmic}[1]
\Require state $\rt^{(t)} \in \R^{n_1}_{\ge 0}$, default $\xt^{(1)} \gets \ones/n_1$;
  if $\rt^{(t)} = \mathbf{0}$: keep $\xt^{(t)} \gets \xt^{(t-1)}$ for $t\ge2$ (the default at $t=1$), set
  $m^{(t)} \gets \mathbf{0}$, $r^{(t)} \gets \mathbf{0}$, play
  $x^{(t)} \gets \xt^{(t)}$, and go to line~6 (no shift)
\State $\xt^{(t)} \gets \rt^{(t)} / \norm{\rt^{(t)}}_1$
  \Comment{pre-iterate; $\yt^{(t)}$ symmetric}
\State $m^{(t)} \gets A \yt^{(t)}$
  \Comment{prediction: payoff at the opponent's pre-iterate}
\State $v^{(t)} \gets \rt^{(t)} + m^{(t)}$;\quad
  $\gamma^{(t)} \gets$ solution of
  $\norm{\pos{v^{(t)} - \gamma \ones}}_2 = \norm{\rt^{(t)}}_2$
\State $r^{(t)} \gets \rt^{(t)} + m^{(t)} - \gamma^{(t)} \ones$
  \Comment{unclipped; norm-preserving}
\State $x^{(t)} \gets \pos{r^{(t)}} / \norm{\pos{r^{(t)}}}_1$
  \Comment{main iterate}
\State observe $u^{(t)} \gets A y^{(t)}$;\quad
  $\mathrm{res}^{(t)} \gets u^{(t)} - m^{(t)}$
\State $g^{(t)} \gets \mathrm{res}^{(t)}
  - \inner{\mathrm{res}^{(t)}}{x^{(t)}} \ones$
  \Comment{centered instantaneous regret}
\State $\rt^{(t+1)} \gets \pos{r^{(t)} + g^{(t)}}$
\end{algorithmic}
\end{algorithm}

\paragraph{Dynamics, limits, and non-degeneracy.}
One round of \iregprm{} defines a memoryless map $\Phi$ on states with positive norms. A fixed point of a map is an input that the map leaves unchanged. Restricting the map to the feasible region $\Ms\cap\{N_1,N_2>0\}$ of the support face yields the face map $\Phi_S:=\Phi|_{\Ms}$. By Lemma~\ref{lem:absorb}, this map sends a neighborhood of the fixed-point cone into itself. We write $z=(\rt_1,\rt_2)$ for the joint state and $N_i(z):=\norm{\rt_i}_2$ for each player's regret norm. Along a trajectory, $N_i(z^{(t)})=\Nt_i(t)$ holds. Let $D\Phi_S$ denote the Jacobian matrix of the face map. Its spectrum of eigenvalues, counted with multiplicity, is $\spec(\cdot)$. On a bounded trajectory, $\Omega$ denotes the $\omega$-limit set of all subsequential limits. Euclidean distance to a set is written $\dist(\cdot,\cdot)$. A strictly complementary equilibrium is \emph{non-degenerate} if the nullspaces satisfy $\ker\Bmat=\mathrm{span}\{y^*|_{S_2}\}$ and $\ker\Bmat^\top=\mathrm{span}\{x^*|_{S_1}\}$. Support sizes must then match, giving $|S_1|=|S_2|$. Non-degeneracy is equivalent to uniqueness for a strictly complementary equilibrium (Lemma~\ref{lem:kernel}). Assumption~\ref{asm:nondeg-strict} posits this property, and we use the two forms interchangeably.
\paragraph{Saturation.} The per-player regret norms govern the dynamics. Their convergence is the central structural event of our theory:\begin{definition}[Norm saturation; per-player inverse norm]\label{def:saturation}
Player $i$'s trajectory \emph{saturates} if $\Nt_i(t) \uparrow \Nt_{i,\infty} < \infty$. The \emph{per-player inverse norm} at round $t$ is $\eta_i(t) := 1 / \Nt_i(t)$ (defined once $\Nt_i(t) > 0$); under saturation with $\Nt_{i,\infty}>0$ it converges to $\eta_{i,\infty} := 1 / \Nt_{i,\infty} > 0$. The symbol $\etainf$ is reserved for the joint $\ell_1$-scale step size of the local theory (Section~\ref{sec:local}).
\end{definition}We define the first nonzero round as $t_i^* := \min\{t : \Nt_i(t) > 0\}$. The equality $t_i^*=\infty$ holds if and only if $\Nt_i\equiv0$. This happens precisely when the observed payoff vector is constant across the actions of player~$i$ at every round. Before round $t_i^*$, player~$i$ plays its default strategy. Lemma~\ref{lem:normmono} guarantees $\Nt_i(t) > 0$ at every subsequent round.

For $t\ge t_i^*$, this structure yields three elementary facts for player~1. First, Fact (F1) gives the orthogonality relation $\inner{g^{(t)}}{\pos{r^{(t)}}} = 0$. Fact (F2) bounds the update vector as $\norm{g^{(t)}} \le (1 + \sqrt{n_1})\, \norm{A}_{\mathrm{op}} \norm{y^{(t)} - \yt^{(t)}}$. We rely solely on the linear dependence in this bound. Building on (F1), we obtain the update rule (F3): $\Nt(t{+}1)^2 = \Nt(t)^2 + \norm{g^{(t)}}^2 - D(t)$, where $0 \le D(t) \le \norm{g^{(t)}}^2$. Combining these properties with the machinery of \citet{zhang2025scale} establishes the next result. Remark~\ref{rem:sat-constant} discusses the associated constant, and Appendix~\ref{app:local:satequiv} contains the proof.\begin{lemma}[Saturation $\Leftrightarrow$ finite second-order path length]\label{lem:sat-equiv}
If both players run \iregprm{} against each other (self-play) on $A$ from $\rt_i^{(1)} = \mathbf{0}$, then both saturate, $\Nt_i(t) \uparrow \Nt_{i,\infty} < \infty$ for $i \in \{1, 2\}$, if and only if the second-order path length is finite:
\[ \sum_{t=1}^{\infty} \Big( \norm{x^{(t)} - \xt^{(t)}}^2 + \norm{y^{(t)} - \yt^{(t)}}^2 \Big) \;<\; \infty . \]
Moreover, under saturation, for each player $i$ with $t_i^*<\infty$ (if $t_i^*=\infty$ the left side below is $0$),
\begin{equation}\label{eq:sat-explicit} \sum_{t=1}^{\infty} \norm{x_i^{(t)} - \xt_i^{(t)}}^2 \;\le\; \frac{2 n_i \, \big( \Nt_{1,\infty} + \Nt_{2,\infty} \big)} {\Nt_i(t_i^*)} \;<\; \infty , \end{equation}
and the internally available scalars $\beta_t^{i} := \inner{m_i^{(t)}}{x_i^{(t)}} - \gamma_i^{(t)}$ ($\beta_t^{i} := 0$ for $t < t_i^*$) satisfy
\begin{equation}\label{eq:beta-free} \beta_t^{i} \;\ge\; \frac{1}{2 n_i} \, \Nt_i(t)\, \norm{x_i^{(t)} - \xt_i^{(t)}}^2 \;\ge\; 0, \qquad \sum_{t=1}^{\infty} \big( \beta_t^{1} + \beta_t^{2} \big) \;\le\; \Nt_{1,\infty} + \Nt_{2,\infty} . \end{equation}
\end{lemma}
\begin{theorem}[Saturation is unconditional]\label{thm:saturation}
For every $A$ and both players, $\Nt_i(t)\uparrow\Nt_{i,\infty}<\infty$.  If $\tau:=\max(t_1^*,t_2^*)<\infty$ then, with $K_i:=(1+\sqrt{n_i})\norm{A}_{\mathrm{op}}$, $\alpha_i:=\Nt_i(\tau)/(2n_i)$, $M:=\max\{K_1^2/\alpha_2,K_2^2/\alpha_1\}$, $Q_\tau:=\Nt_1(\tau)^2+\Nt_2(\tau)^2$,
\begin{equation}\label{eq:sat-bound}
\Nt_1(t)+\Nt_2(t)\;\le\;M+\sqrt{M^2+2Q_\tau}\qquad\text{for all } t ;
\end{equation}
if $t_i^*=\infty$ then $\Nt_i\equiv0$ and the opponent's norm is eventually constant.
\end{theorem}We give a four-line proof in Appendix~\ref{app:local:satequiv}. It combines (F2), (F3), and the finite-horizon form of \eqref{eq:beta-free}. This result restates the bounded second-order path length of \citet[Theorem~5.5]{zhang2025scale} with an explicit constant. Square-summable residuals still permit slow drift, leaving convergence open. Sections~\ref{sec:global}--\ref{sec:local} supply the $\omega$-limit argument and establish the convergence rate.

\emph{In words.} Each player's regret norm and the algorithm's step size converge to finite limits. For each player, the regret norm is nondecreasing across rounds, and the payoff matrix together with the initial rounds bounds it from above. Every bounded, nondecreasing sequence converges. The step size in each round equals the reciprocal of this norm, so it converges as well.

\begin{proposition}[Fixed points $\Leftrightarrow$ Nash equilibria]\label{prop:fixedpoint}
Let $z = (\rt_1, \rt_2)$ with $\rt_1 \ne \mathbf{0} \ne \rt_2$, and let $\Phi$ be the one-round map of Algorithm~\ref{alg:iregprm}. Then $\Phi(z) = z$ if and only if $(x, y) := \big( \rt_1 / \norm{\rt_1}_1, \; \rt_2 / \norm{\rt_2}_1 \big)$ is a Nash equilibrium of $A$. In that case the shifts equal the game value: $\gamma_1 = x^\top A y = v^*$ and $\gamma_2 = -v^*$. Consequently,
\[ \Fix(\Phi) \cap \{ \rt_1 \ne \mathbf{0} \ne \rt_2 \} \;=\; \big\{ (c_1 x^*, \, c_2 y^*) \;:\; c_1, c_2 > 0, \; (x^*, y^*) \text{ is an NE of } A \big\}, \]
the positive cone over the NE directions: regret directions are pinned to equilibrium strategies while the scales $c_1, c_2$ that set the norms stay free. The proof is in Appendix~\ref{app:local:p1}.
\end{proposition}

\section{Global Convergence}
\label{sec:main}\label{sec:global}

Using the saturation property from Theorem~\ref{thm:saturation}, we show that the last-iterate gap vanishes on every matrix game (Theorem~\ref{thm:global}). Combining this limit-set characterization with an equilibrium potential yields last-iterate convergence to a single Nash equilibrium, even when the game has multiple equilibria (Theorem~\ref{thm:point}). We establish a local linear convergence rate. Near a unique strictly complementary equilibrium satisfying the scale-invariant step-size condition of Assumption~\ref{asm:stepsize}, this rate admits a closed form (Theorem~\ref{thm:rate}). Our analysis also provides a ratio certificate (Theorem~\ref{thm:certificate}). Full proofs appear in Appendices~\ref{app:local}--\ref{app:certificate}, and we sketch the arguments below.

Theorem~\ref{thm:saturation} gives $\Nt_i(t)\uparrow\Nt_{i,\infty}<\infty$. We assume a finite threshold time $t_+:=\max(t_1^*,t_2^*)<\infty$. Under this assumption, both players move at least once and the limit $\Nt_{i,\infty}>0$ is positive. For all rounds $t\ge t_+$, one round of \algname{} updates the space of positive-norm states via the memoryless map $\Phi$:
\begin{equation}\label{eq:statespace}
  \mathcal{Z} \;:=\; \bigl\{\, z : \rt_1\ge 0,\ \rt_2\ge 0,\ N_1(z)>0,\ N_2(z)>0 \,\bigr\}, \qquad\qquad \Phi:\mathcal{Z}\to\mathcal{Z}.
\end{equation}
A norm-preserving pairing identity (Lemma~\ref{lem:pairing}) relates the displacement of the prediction step to the pre-iterate movement. This pre-iterate movement is square-summable by Lemma~\ref{lem:sat-equiv}.

\begin{theorem}[Global asymptotics; no spectral condition]
\label{thm:global}
Let $t_+<\infty$, so that both norms saturate at positive limits $\Nt_{i,\infty}$ (Theorem~\ref{thm:saturation}).  Then:
\begin{enumerate}[label=(\roman*),itemsep=1pt,topsep=2pt]
\item The tail trajectory $\{z^{(t)}\}_{t\ge t_+}$ has nonempty, compact $\omega$-limit set $\Omega$ with $\Phi(\Omega)=\Omega$, $\dist(z^{(t)},\Omega)\to 0$, and
\[
\Omega\;\subseteq\;\Fix(\Phi)\,\cap\;\{\,z:\ \rt_i\ge0,\ \norm{\rt_i}_2=\Nt_{i,\infty}\,\}.
\]
Every point of $\Omega$ is a Nash ray $(c_1x^*,c_2y^*)$ (Proposition~\ref{prop:fixedpoint}), where $\gamma_1=v^*$.
\item Consequently the last-iterate Nash gap vanishes:
\[
\gapf\bigl(x^{(t)},y^{(t)}\bigr)\;=\;\max_i\,(Ay^{(t)})_i-\min_j\,({x^{(t)}}^{\!\top}\! A)_j\ \longrightarrow\ 0,
\]
and $\dist\bigl((x^{(t)},y^{(t)}),\mathrm{NE}\bigr)\to 0$, for the pre-iterates too.
\end{enumerate}
Only $t_+<\infty$ is assumed; no spectral, strict-complementarity or uniqueness condition enters.  If $t_+=\infty$ the gap is $0$ in finitely many rounds (Lemma~\ref{lem:neverpos}), so (ii) holds on every matrix game.
\end{theorem}\emph{Proof sketch} (Appendix~\ref{sec:global:proofs}). We bound the displacement at each step and show that every limit point is a fixed point. Each step moves the iterate by $O(\norm{x-\xt}+\norm{y-\yt})$. Lemma~\ref{lem:sat-equiv} implies that this sequence of displacements is $\ell^2$-summable. Continuity of the $\gamma$-shift gives $\norm{X(\bar z)-\tilde X(\bar z)}=0$ at every point in the limit set $\Omega$. Equation~\eqref{eq:pairing} lifts this equality of strategies to equality of the regret vectors \citep{la1976stability}. A compactness bootstrap argument establishes part~(ii).

\emph{In words.} Once both norms have almost stopped changing, the trajectory can return repeatedly only to states left unchanged by a single round. Each such state encodes a Nash equilibrium, where neither player can increase their payoff by changing strategy alone. The played strategies therefore converge to the set of equilibria. Recall that the gap measures the total payoff the two players could gain by deviating unilaterally. This quantity goes to zero.

\begin{remark}[Attribution]\label{rem:attrib}
As for the strategies, part~(ii) also follows from Theorem~5.5 and Lemma~D.2 of \citet{zhang2025scale} (see Appendix~\ref{app:global:zhang}). Furthermore, Theorem~\ref{thm:global} characterizes the possible limit points of the regret state $z$, while Theorem~\ref{thm:point} establishes that the state admits a limit whenever $t_+<\infty$.
\end{remark}

\begin{theorem}[Convergence to a single equilibrium; no uniqueness, no spectral condition]
\label{thm:point}
Consider \algname{} applied to any zero-sum matrix game.
\begin{enumerate}[label=(\roman*),itemsep=1pt,topsep=2pt]
\item If $t_+<\infty$, there exists a Nash equilibrium $(\bar x,\bar y)$ such that the entire state trajectory converges to a single point, namely
\[
z^{(t)}\ \longrightarrow\ \bar z:=\Bigl(\tfrac{\Nt_{1,\infty}}{\norm{\bar x}_2}\,\bar x,\; \tfrac{\Nt_{2,\infty}}{\norm{\bar y}_2}\,\bar y\Bigr), \qquad x^{(t)},\xt^{(t)}\to \bar x,\quad y^{(t)},\yt^{(t)}\to \bar y .
\]
\item If $t_+=\infty$, the strategy pair $(x^{(t)},y^{(t)})$ remains constant after finitely many rounds, and its constant value is a Nash equilibrium (Lemma~\ref{lem:neverpos}).
\end{enumerate}
In particular, in both cases the last-iterate strategy pair converges to a single Nash equilibrium. This holds without any uniqueness, strict-complementarity, or spectral condition. When the Nash equilibrium is unique, we have $(\bar x,\bar y)=(x^*,y^*)$.
\end{theorem}\emph{Proof sketch} (Appendix~\ref{sec:global:proofs}). Suppose that $t_+<\infty$. Theorem~\ref{thm:global} shows that every limit point lies at the saturation levels on a Nash ray, a half-line generated by a Nash equilibrium. A continuum of Nash equilibria could still admit multiple limit points. A potential function rules out this possibility. For every Nash equilibrium $(a,b)$, the pairing $\inner{a}{\rt_1^{(t)}}+\inner{b}{\rt_2^{(t)}}$ is nondecreasing and bounded (Lemma~\ref{lem:potential}). The realized payoff cancels between the two players. The saddle-point inequalities and the condition $\beta_t\ge0$ make the remaining contribution nonnegative, establishing monotonicity. Fix a limit point $\bar z$ with associated equilibrium direction $(\bar x,\bar y)$. The Cauchy--Schwarz deficit $\Nt_1(t)\norm{\bar x}_2-\inner{\bar x}{\rt_1^{(t)}}+\Nt_2(t)\norm{\bar y}_2-\inner{\bar y}{\rt_2^{(t)}}$ converges to a limit. Along the subsequence converging to $\bar z$, the deficit tends to zero. Its full limit is zero. Hence the iterates converge as $z^{(t)}\to\bar z$. When $t_+=\infty$, Lemma~\ref{lem:neverpos}(v) directly yields the claim.

\emph{In words.} A game can have many equilibria. Even as the strategies approach the set of equilibria, they might still wander along this set forever. Our theorem rules out such wandering and shows that the strategies settle on a single equilibrium. For each equilibrium, our proof tracks a weighted sum of the two regret vectors. This sum is bounded and monotonically nondecreasing, ruling out oscillation between distinct equilibria.

\begin{remark}[Relation to concurrent work]\label{rem:single-ne}
In concurrent work, \citet{nguyen2026convergencealternatingregretmatching} establish in Theorem~7.5 convergence to a single Nash equilibrium without assuming uniqueness for the alternating variants Alt-IR-PRM$^+$ and Alt-HPRM$^+$. As noted in their Remark~7.6, earlier set-convergence results for IR-PRM$^+$ and \iregprm{} had left convergence to a single equilibrium unresolved when the equilibrium is not unique. Theorem~\ref{thm:point} resolves this question for the simultaneous extragradient dynamic. To do so, we adapt both their equilibrium pairing and their Cauchy--Schwarz deficit argument. Under simultaneous extragradient updates, the realized payoffs cancel exactly. In contrast, their proof for Alt-IR-PRM$^+$ relies on a utility-improvement bound for the column player. Finally, \citet{cai2025lastiterate} establish convergence to a single Nash equilibrium under simultaneous updates for ExRM$^+$ and SPRM$^+$, both of which employ a payoff-dependent step size.
\end{remark}

\section{Local Linear Rate}
\label{sec:local}

Saturation forces convergence, and the observed rate is \emph{linear}. Establishing this rate requires two hypotheses beyond saturation. Near a strictly complementary equilibrium, a single step of \algname{} places the state on the support face. The dynamics on this face are real-analytic. The eigenvalues of the local Jacobian then govern the rate of convergence.\begin{assumption}[Unique, strictly complementary equilibrium]
\label{asm:unique-strict}
The game $A\in\R^{n_1\times n_2}$ has a unique Nash equilibrium $(x^*,y^*)$, with value $v^*$, and the equilibrium is \emph{strictly complementary}: with $S_1:=\supp(x^*)$, $S_2:=\supp(y^*)$,
\begin{equation}
\label{eq:margin}
\delta \;:=\; \min\Bigl\{\, \min_{i\notin S_1}\bigl(v^*-(Ay^*)_i\bigr),\ \min_{j\notin S_2}\bigl((x^{*\top}A)_j-v^*\bigr) \Bigr\} \;>\;0 .
\end{equation}
\end{assumption}Our proofs use uniqueness only through a non-degeneracy condition on the kernel of the support matrix. By Lemma~\ref{lem:kernel}, this condition is equivalent to uniqueness for a strictly complementary equilibrium.\begin{proposition}[Assumption~\ref{asm:unique-strict} is uniqueness]\label{prop:unique-strict}
A zero-sum matrix game satisfies Assumption~\ref{asm:unique-strict} if and only if its Nash equilibrium is unique.
\end{proposition}Every matrix game admits a strictly complementary equilibrium pair \citep{goldman1956theory}. A unique equilibrium coincides with this pair. We then obtain $\delta>0$ (Appendix~\ref{app:local:unique-strict}). Equilibrium uniqueness holds outside a Lebesgue-null set of payoff matrices \citep{bohnenblust1950solutions}. Appendix~\ref{app:add-asm} works out two explicit $3\times3$ games, one fully mixed and one with a $2\times2$ support.

We define the support submatrix $\Atil:=A[S_1,S_2]$ and the \emph{value-centered support matrix} $\Bmat:=\Atil-v^*\ones\ones^\top$. Restricting the equilibrium strategies to their supports gives the face coordinates $\xi:=x^*|_{S_1}$ and $\psi:=y^*|_{S_2}$. Proposition~\ref{prop:fixedpoint} characterizes the fixed-point cone as $E_+:=\{(a\xi,b\psi):a,b>0\}$. For scalars $a,b>0$, a fixed point $(a x^*,b y^*)$ carries the effective step size $\eta(a,b):=1/\sqrt{ab}$. Evaluating this quantity at the reference fixed point $z^*=(c_1x^*,c_2y^*)$ yields $\etainf:=\eta(c_1,c_2)$.\begin{assumption}[Spectral step-size condition]
\label{asm:stepsize}
At the reference fixed point, we have $\etainf\,\sigmax(\Bmat) < 1$. In particular, when we apply this assumption at the limit of a trajectory with $t_+<\infty$ under Assumption~\ref{asm:unique-strict}, the reference point coincides with the limit $\bar z$ of Theorem~\ref{thm:point}, with $c_1=\Nt_{1,\infty}/\norm{x^*}_2$ and $c_2=\Nt_{2,\infty}/\norm{y^*}_2$.
\end{assumption}

\emph{In words.} Assumption~\ref{asm:unique-strict} requires a unique equilibrium in the game. Any action outside its support yields a payoff strictly below the equilibrium value. Under Assumption~\ref{asm:stepsize}, the limiting step size stays small relative to the largest singular value of the matrix $\Bmat$, which measures its maximum stretching factor. This bound mirrors the requirement for the fixed-step extra-gradient method. That algorithm keeps the step size below the inverse of the operator norm.

\begin{lemma}[Monotone certificate for Assumption~\ref{asm:stepsize}]\label{lem:epscert}
Under Assumption~\ref{asm:unique-strict} and $t_+<\infty$, for $t\ge t_+$ let
\[
\eps(t)\;:=\;\sigmax(\Bmat)\,\frac{\sqrt{\norm{x^*}_2\norm{y^*}_2}}{\sqrt{\Nt_1(t)\,\Nt_2(t)}} .
\]
Then $\eps$ is nonincreasing and $\lim_{t\to\infty}\eps(t)=\etainf\sigmax(\Bmat)$.  Hence Assumption~\ref{asm:stepsize} holds if and only if $\eps(t)<1$ for some $t\ge t_+$, and then for every later $t$: a recorded crossing certifies the assumption, and under the assumption a crossing is recorded once the record extends far enough.
\end{lemma}By Lemma~\ref{lem:normmono}, the norms are nondecreasing, so the quantity $\eps$ is nonincreasing. The limiting ratios $c_1=\Nt_{1,\infty}/\norm{x^*}_2$ and $c_2=\Nt_{2,\infty}/\norm{y^*}_2$ determine its limit (Appendix~\ref{app:local:epscert}). Evaluating this certificate requires the largest singular value $\sigmax(\Bmat)$ and the equilibrium norms. The certificate verifies the condition only once the equilibrium is known. We test this behavior on the $175$ testbed instances with a unique equilibrium. The bound $\eps(t)<1$ holds by iteration $t=250$ on $168$ instances and on all $175$ instances by the end of the run.

\emph{In words.} The quantity $\eps(t)$ substitutes the current norms for the limiting norms. The norms never decrease, so $\eps(t)$ only decreases. A single reading below one guarantees that the condition holds for the rest of the run.

We characterize the local dynamics on the support face $\Ms$ via the Jacobian at a fixed point. For nonnegative states, this face is defined by the support condition $\supp(\rt_i)\subseteq S_i$. One step absorbs onto this face with exact support freezing (Lemma~\ref{lem:absorb}). At a fixed point $z_f=(a\xi,b\psi)$, the face Jacobian takes the closed form
\[
J \;=\; I+L+L^2, \qquad L \;=\; \begin{pmatrix} 0 & \Bmat/b \\ -\Bmat^\top/a & 0 \end{pmatrix},
\]
where the matrix $L$ is the one-step \emph{extra-gradient} operator on the matrix $\Bmat$ at the effective step size $\eta(a,b)$ selected by the saturation level (Lemma~\ref{lem:jacobian}). The centering identity of Lemma~\ref{lem:centering} ensures that the square of $L$ has real nonpositive eigenvalues. This pins the one-step spectrum to a single closed-form threshold.

We express this threshold through the singular values of the matrix. These are the square roots of the eigenvalues of its Gram matrix. Let the multiset $\Sigma^+(\Bmat)$ collect the nonzero singular values of $\Bmat$. Under Assumption~\ref{asm:nondeg-strict}, there are exactly $m-1$ such values, where $m:=|S_1|=|S_2|$ is the support size. We fix Assumption~\ref{asm:unique-strict}. We also set aside the pure saddle case, where each player chooses a single deterministic action, corresponding to $m=1$. When this happens, $\Bmat=0$, and the feasible part $\Ms\cap\mathcal Z$ of the support face coincides with the fixed-point cone $E_+$. Lemma~\ref{lem:absorb} then gives $\Phi(\mathcal U_0\cap\mathcal Z)\subseteq\Fix(\Phi)$ for a neighborhood $\mathcal U_0$ of the fixed point. The iteration converges in one round. Theorem~\ref{thm:rate}(ii)--(iv) and Section~\ref{sec:certificate} therefore assume $m\ge2$.\begin{theorem}[Closed-form local linear rate]
\label{thm:rate}
Let Assumption~\ref{asm:unique-strict} hold with $m\ge2$ and let $z^*=(c_1x^*,c_2y^*)$, $c_1,c_2>0$, be a reference fixed point.
\begin{enumerate}[label=(\roman*),itemsep=1pt,topsep=2pt]
\item \textbf{(Spectrum.)}  For every fixed point $z_f=(a\xi,b\psi)$ with $\eta=\eta(a,b)$,
\begin{equation}
\label{eq:spec}
\spec\bigl(D\Phi_S(z_f)\bigr) \;=\; \{1,\,1\}\ \cup\ \bigl\{\, 1-\eta^2\sigma^2 \pm i\,\eta\sigma \;:\; \sigma\in\Sigma^+(\Bmat) \bigr\}
\end{equation}
as a multiset; the eigenvalue $1$ has algebraic multiplicity exactly two, is semisimple, with eigenspace the tangent space of the fixed-point cone.
\item \textbf{(Stability threshold.)}  The non-unit spectral radius is
\begin{equation}
\label{eq:rhostar}
\rho(a,b) \;=\; \max_{\sigma\in\Sigma^+(\Bmat)} \sqrt{\,1-\eta^2\sigma^2\bigl(1-\eta^2\sigma^2\bigr)\,},
\end{equation}
and $\rho(a,b)<1$ iff $\eta\,\sigmax(\Bmat)<1$, while $\eta\,\sigmax(\Bmat)>1$ makes some mode unstable.
\item \textbf{(Local linear last-iterate convergence.)}  Under Assumption~\ref{asm:stepsize}, let $\rho^*:=\rho(c_1,c_2)<1$.  For every $\eps\in(0,1-\rho^*)$ there is an open neighborhood $\mathcal{U}'\ni z^*$ and $C_\eps<\infty$ such that every trajectory with $z^{(t_0)}\in\mathcal{U}'$ converges to a fixed point $z_\infty=(a_\infty\xi,b_\infty\psi)\in E_+$ with
\begin{equation}
\label{eq:linrate}
\norm{z^{(t)}-z_\infty} \;\le\; C_\eps\,(\rho^*+\eps)^{\,t-t_0} \qquad\text{for all } t\ge t_0 .
\end{equation}
\item \textbf{(Gap at the same rate.)}  With the same $\mathcal{U}'$,
\begin{equation}
\label{eq:gaprate}
\gapf\bigl(x^{(t)},y^{(t)}\bigr) \;\le\; L_{\gapf}\,C_\eps\,(\rho^*+\eps)^{\,t-t_0} \qquad\text{for all } t\ge t_0,
\end{equation}
where $L_{\gapf}<\infty$ depends only on $\norm{A}$, $c_1$, $c_2$, and the dimensions.
\end{enumerate}
\end{theorem}\emph{Proof sketch} (Appendix~\ref{app:local:rate}). Our proof is elementary. We combine linear algebra with an induction and a uniform second-order Taylor expansion. A unit eigenvalue of multiplicity two remains in the spectrum throughout. Standard arguments through the spectral radius, the largest eigenvalue modulus, fail here. The space splitting $E\oplus W$ block-diagonalizes the Jacobian. The drift along the subspace $E$ is $O(\norm{w}^2)$. A telescoping sum establishes that the limit exists.

\emph{In words.} Near the limit, a single round acts as multiplication by a fixed matrix $J$. Two directions of $J$ leave the error unchanged. Rescaling the two regret vectors along these directions preserves the strategies. Under our step-size condition, every other direction shrinks each round by a factor of at most $\rho^*<1$. The distance to the limit and the gap decay geometrically at rate $(\rho^*)^t$, up to an arbitrarily small loss in the base.

For a fixed step size $\eta$, the largest singular value $\sigmax(\Bmat)$ alone determines whether all transverse modes, the directions transverse to the fixed-point cone, contract. The bound $\etainf\sigmax(\Bmat)<1$ suffices for the linear envelope. A single trajectory can still converge when this condition fails, as Theorem~\ref{thm:point} shows. The matrix family $A_\lambda$ of Section~\ref{sec:capture} provides another example: it admits a trajectory converging to a limit that violates this bound.

Which extreme nonzero singular value dictates the convergence rate follows from Corollary~\ref{cor:extreme}. A small smallest positive singular value $\sigmin^{+}(\Bmat)$ slows convergence while preserving stability. This rate is \emph{spectral}. The parameter $\delta$ enters the leading constants and leaves the asymptotic factor $\rho^*$ unchanged. Parts~(iii)--(iv) establish an \emph{envelope} of $O((\rho^*+\eps)^t)$ for every tolerance $\eps>0$. Across the measurable runs in Section~\ref{sec:experiments}, the fitted slopes match this bound.\label{rem:rate-semantics}

\label{sec:capture}

\begin{corollary}[Global entry $+$ local capture $\Rightarrow$ eventual linear convergence]
\label{cor:capture}
Assume the hypotheses of Theorem~\ref{thm:global}, and suppose $\Omega$ contains \emph{at least one} point $z^\dagger$ to which Theorem~\ref{thm:rate} applies.  Write $z^\dagger=(a_1x^*,a_2y^*)$ for a strictly complementary equilibrium $(x^*,y^*)$ (Assumption~\ref{asm:unique-strict}), so that its effective step size $\eta(z^\dagger):=1/\sqrt{a_1a_2}$ satisfies $\eta(z^\dagger)\,\sigmax(\Bmat)<1$.  Then the entire trajectory converges to a single fixed point $z_\infty\in\Fix(\Phi)$, and there exist $T^*\ge t_+$, constants $C,C'<\infty$ and $\bar\rho\in(\rho^*,1)$, with $\rho^*$ the closed-form rate of Theorem~\ref{thm:rate}, such that for all $t\ge T^*$,
\[
\norm{z^{(t)}-z_\infty}\;\le\;C\,\bar\rho^{\,t-T^*}, \qquad \gapf\bigl(x^{(t)},y^{(t)}\bigr)\;\le\;C'\,\bar\rho^{\,t-T^*} .
\]
Under Assumptions~\ref{asm:unique-strict} and~\ref{asm:stepsize} with $m\ge2$ the hypothesis on $\Omega$ is automatic.  For the pure saddle $m=1$, Lemma~\ref{lem:absorb} gives capture in finitely many rounds, so the conclusion holds for any $\bar\rho$.
\end{corollary}\emph{Proof sketch} (Appendix~\ref{sec:global:proofs}). The trajectory visits every neighborhood of $z^\dagger$ infinitely often, including the capture neighborhood $\mathcal{U}'$ from Theorem~\ref{thm:rate}. At some finite time $T^*$, the trajectory enters this neighborhood and is captured.

\emph{In words.} We combine our global and local guarantees to establish the full convergence rate. The global theorem brings the trajectory close to its limit. Once the iterates enter a neighborhood of that limit, the local theorem takes over. Suppose the limit satisfies the step-size condition. Convergence is then linear from some finite round onward.

Any ray satisfying $\etainf\,\sigmax(\Bmat)>1$ is linearly unstable by Theorem~\ref{thm:rate}(ii). Nearby trajectories are generically pushed away. By Lemma~\ref{lem:normmono}, their norms continue to grow. We expect this growth to persist until the trajectory reaches a ray satisfying the stability condition. Section~\ref{sec:experiments} tests this hypothesis.

\begin{table}[!t]
\centering\small\setlength{\tabcolsep}{2pt}
\caption{The family $A_\lambda$: the symmetric orbit, the zero-initialized run, and perturbed runs.}
\label{tab:alambda}
\begin{tabular*}{\textwidth}{@{\extracolsep{\fill}}rccccccc@{}}
\toprule
$\lambda$ & $\eta_{\rm sym}\sigmax$ & $>1$ & escape step & realized $\etainf\sigmax$ & $\mu_{\rm meas}/\mu_{\rm pred}$ & slope ratio & final $<1$ \\
\midrule
1.0 & 0.7325 & no & never & 0.7325 & -- & -- & -- \\
2.0 & 0.9766 & no & never & 0.9766 & -- & -- & -- \\
2.1 & 1.0254 & yes & 605 & 0.9736 & -- & -- & -- \\
2.5 & 1.2208 & yes & 50 & 0.7533 & 0.9996 & 0.992 & 60/60 \\
3.0 & 1.4649 & yes & 23 & 0.7455 & 0.9997 & 0.963 & 60/60 \\
5.0 & 2.4415 & yes & 10 & 0.7082 & 0.9801 & 1.114 & 60/60 \\
20.0 & 9.7661 & yes & 8 & 0.6944 & 1.0296 & 0.989 & 60/60 \\
100.0 & 48.8306 & yes & 6 & 0.7270 & -- & -- & -- \\
\bottomrule
\end{tabular*}

\smallskip
{\footnotesize Columns 2--3: trajectory held on $\rt_2=\rt_3$ (exact arithmetic).  Columns 4--5: unconstrained zero-initialized float64 run, $T=2\times10^5$; the escape step is the first $t$ with $|\rt_2-\rt_3|>10^{-8}$ (rounding seeds the unstable mode), and the realized $\etainf\sigmax(\Bmat)$ for $\lambda\ge2.5$ is one float64 realization.  Columns 6--8: $10$ seeds each, perturbation $\eps\cdot U[0,1)^3$ applied after $2\times10^4$ saturated rounds on the ray; $\mu_{\rm pred}=\sqrt{1-q+q^2}$ with $q=(\eta_{\rm sym}2\lambda)^2$; the $t_{\rm esc}$ slope is $dt_{\rm esc}/d\log_{10}(1/\eps)$ against the prediction $\ln10/\ln\mu_{\rm pred}$.\par}
\end{table}

An individual orbit can converge without satisfying the stability condition. Consider the matrix family $A_\lambda=uu^\top+\lambda vv^\top$ with scalar parameter $\lambda>0$, where $u=(1,-1,-1)^\top$ and $v=(0,1,-1)^\top$. Starting from zero, the orbit remains in the invariant subspace where each player's second and third regret coordinates are equal. Payoff vectors on this subspace do not depend on the parameter, making both the orbit and its limiting scales independent of it. Yet the largest singular value satisfies $\sigmax(\Bmat)=\max(3,2\lambda)$ and grows without bound. For all sufficiently large parameter values, the limit of this orbit violates the stability condition and realizes the exceptional case of Proposition~\ref{prop:selfstab}. Numerically, the orbit converges to a ray with $\eta_{\rm sym}\sigmax(\Bmat)\approx0.2442\max(3,2\lambda)$. This product exceeds one whenever $\lambda>1/(2\eta_{\rm sym})\approx2.048$. The limiting ray is a saddle. In our experiments, any perturbation with a component along the unstable mode escapes the subspace at the rate predicted by the linearization. Across all $240$ perturbed runs, the new limit satisfies the stability condition (Appendix~\ref{app:add-exp}).\begin{proposition}[Instability of step-size-violating rays]
\label{prop:selfstab}
Let $z^\dagger=(c_1x^*,c_2y^*)\in\Fix(\Phi)$ correspond to a strictly complementary, nondegenerate equilibrium (Assumption~\ref{asm:nondeg-strict}), and suppose $\eta(z^\dagger)\,\sigmax(\Bmat)>1$.  Then:
\begin{enumerate}[label=(\alph*),itemsep=1pt,topsep=2pt]
\item (Full space) There is a neighborhood $B_\dagger\ni z^\dagger$ and a $C^1$ local center-stable manifold $W^{cs}_{loc}(z^\dagger)$ of dimension $d-d_u$, where $d:=n_1+n_2$ and $d_u=2\,\#\{k:\eta(z^\dagger)\sigma_k(\Bmat)>1\}\ge2$, such that every orbit that remains in $B_\dagger$ forever is contained in $W^{cs}_{loc}(z^\dagger)$.
\item ($\Ms$-local, measure zero) Within the support subspace $\Ms$, forward-invariant near $z^\dagger$ by Lemma~\ref{lem:absorb}, the local basin of attraction of $z^\dagger$ has Lebesgue measure zero in $\Ms$.
\end{enumerate}
\end{proposition}\emph{Proof sketch}. Our proof combines three ingredients. Lemma~\ref{lem:analytic} provides an analytic extension of $\Phi$ to an ambient neighborhood of the ray. Next, we apply a theorem on the center-stable manifold, the set of nearby states whose orbits stay near the fixed point, for $C^1$ maps \citep{hirsch1977invariant} in its saddle-avoidance form \citep{lee2019first}. We also invoke the nonsingularity of the face Jacobian $J=I+L+L^2$ from Lemma~\ref{lem:jacobian}. The eigenvalues of $J$ satisfy $|\mu|^2=1-q+q^2\ge3/4$ for all $q=\etainf^2\sigma^2\ge0$. The face map is therefore a local diffeomorphism near the ray. Preimages of null sets along orbit segments that remain in that neighborhood are themselves null.

\emph{In words.} Suppose a limit point violates the step-size condition. It repels almost every nearby trajectory. Consider the face of points sharing the same active actions. Within this face, the initial conditions that remain nearby form a set of measure zero.

\begin{remark}[Self-stabilization]
\label{rem:mechanism}
We establish that rays satisfying $\etainf\sigmax(\Bmat)>1$ are unstable and that their trapped local basin is Lebesgue-null in $\Ms$ (Proposition~\ref{prop:selfstab}). Once the bound holds strictly, capture follows (Corollary~\ref{cor:capture}).  Between those endpoints sits the link that the norms keep growing while the bound is violated; we measure it on all $184$ instances with a resolvable limit, where $\etainf\sigmax(\Bmat)$ has median $0.790$, maximum $0.9985$, and zero violations, with several instances just inside the boundary (Section~\ref{sec:experiments}).  The center-stable set lives in the full state space, while the measurement follows the single zero-initialized orbit (Remark~\ref{rem:init}).
\end{remark}

\section{The Ratio Law}
\label{sec:certificate}

Theorem~\ref{thm:global} shows that the gap converges to zero. Here we prove a \emph{quantitative} counterpart. We establish a \emph{ratio} law under which the norm increments track the gap at the level of \emph{rates}. This law absorbs the equilibrium-dependent constants into a single condition number. In Lemma~\ref{lem:epscert}, the monotone certificate $\eps(t)$ verifies the spectral condition once the equilibrium is known, and by itself gives no information about regime entry or a target gap.

In a matrix game, the exact gap of the main iterate equals $\max_i (Ay^{(t)})_i+\max_j(-A^\top x^{(t)})_j$. Line~6 of Algorithm~\ref{alg:iregprm} already observes both payoff vectors, making direct gap monitoring immediate. For matrix games, the ratio law characterizes an internal quantity of the algorithm. Extensive-form games require a separate best-response pass to compute the gap, and the ratio law motivates an online monitor for that setting (Section~\ref{sec:experiments}). We track the squared-norm increment, defined for each player and in total by
\[
\Ndot_i(t)\;:=\;\Nt_i(t{+}1)^2-\Nt_i(t)^2\;\ge\;0, \qquad \Ndottot(t)\;:=\;\Ndot_1(t)+\Ndot_2(t).
\]
Lemma~\ref{lem:normmono} guarantees nonnegativity. Two scalar subtractions compute this quantity online.

We impose Assumptions~\ref{asm:unique-strict} and~\ref{asm:stepsize} throughout in the local regime of Theorem~\ref{thm:rate}. Lemma~\ref{lem:absorb} freezes the set of active actions once the state is near the fixed-point cone. From then on, the state decomposes as $z^{(t)}=z_f^{(t)}+w^{(t)}$, where $z_f^{(t)}$ lies in the fixed-point cone. The residual norm measures the distance to the fixed-point set: $\dist(z^{(t)},\Fix)=\norm{w^{(t)}}$. We track on-support energy because taking differences of the full squared norm $\Nt^2$ loses representability. Summing this energy across both players completes our measurement, as formalized in Lemmas~\ref{lem:ndot-identity}, \ref{lem:gap-sharp}, and~\ref{lem:ndot-theta} along with Remarks~\ref{rem:ndot-measure} and~\ref{rem:single-player} in Appendix~\ref{app:certificate}.\begin{theorem}[Ratio certificate]
\label{thm:certificate}
Assume Assumptions~\ref{asm:unique-strict} and~\ref{asm:stepsize}, and let the trajectory be in the frozen local regime of Theorem~\ref{thm:rate} with $\norm{w^{(t)}}\le\min(\delta_2,\delta_3)$ for all $t\ge t_{\mathrm{reg}}$, the regime-entry time, where $\delta_2$ and $\delta_3$ are the radii of Lemmas~\ref{lem:gap-sharp} and~\ref{lem:ndot-theta}.  Then for all $t_2\ge t_1\ge t_{\mathrm{reg}}$ with $\Ndottot(t_1)>0$,
\[
\kappa^{-1}\,\sqrt{\frac{\Ndottot(t_2)}{\Ndottot(t_1)}} \;\le\; \frac{\gapf(t_2)}{\gapf(t_1)} \;\le\; \kappa\,\sqrt{\frac{\Ndottot(t_2)}{\Ndottot(t_1)}}, \qquad \kappa:=\frac{L_{\gapf}}{c_-}\cdot\sqrt{\frac{n_+}{n_-}}\,,
\]
where $L_{\gapf},c_-,n_\pm$ are the constants of Lemmas~\ref{lem:gap-sharp} and~\ref{lem:ndot-theta}; $\kappa$ is the product of a gap-sharpness condition number $L_{\gapf}/c_-$ and a spectral condition number $\sqrt{n_+/n_-}$, independent of $t_1,t_2$ and of the trajectory amplitude; no factor $1/(1-\rho^*)$ enters.
\end{theorem}The hypothesis $\Ndottot(t_1)>0$ excludes one degenerate case. By Lemma~\ref{lem:ndot-theta}, $\Ndottot(t_1)=0$ forces $w^{(t_1)}=0$, placing the state on the fixed-point cone where $\gapf(t_1)=0$. For nondegenerate trajectories, we evaluate the gap at a single time $t_1$ to define the estimate $\widehat{\gapf}(t_2):=\gapf(t_1)\sqrt{\Ndottot(t_2)/\Ndottot(t_1)}$. The true gap then satisfies $\kappa^{-1}\widehat{\gapf}(t_2)\le\gapf(t_2)\le\kappa\,\widehat{\gapf}(t_2)$ at all subsequent times.

\emph{In words.} Near the solution, the gap and the square root of the norm increment shrink together. Suppose the norm increment drops by a factor of $100$ between two rounds. The gap then drops by a factor of $10$, up to the constant factor $\kappa$.

The constant $\kappa$ remains unobservable during execution. It depends on unknown regime thresholds and properties of the target equilibrium: the set of actions played with positive probability, the smallest nonzero probabilities on that set, and the extreme nonzero singular values of the matrix $\Bmat$. The certificate derives its practical value from the slope-$2$ law of Corollary~\ref{cor:theta}. We evaluate this law without unknown equilibrium constants, using only the increment $\Ndottot$ and the gap measured over the fitted window.\begin{corollary}[$\theta=2$ prediction]
\label{cor:theta}
Under the hypotheses of Theorem~\ref{thm:certificate}, for all $t\ge t_{\mathrm{reg}}$ with $\gapf(t)>0$ (equivalently $\Ndottot(t)>0$, Lemma~\ref{lem:ndot-theta}),
\[
\log\Ndottot(t)\;=\;2\,\log\gapf(t)\;+\;C(t), \qquad C(t)\in\Bigl[\log\tfrac{n_-}{L_{\gapf}^{2}},\;\log\tfrac{n_+}{c_-^{2}}\Bigr],
\]
an interval of fixed width $2\log\kappa$ independent of $t$.  Consequently, over any measurement window with $\mathrm{sd}(\log\gapf)>0$, the least-squares slope $\theta$ of $\log\Ndottot$ against $\log\gapf$ satisfies $|\theta-2|\le\log\kappa/\mathrm{sd}(\log\gapf)$, where $\mathrm{sd}(\log\gapf)$ is the standard deviation of the recorded gap logarithms.  As $\mathrm{sd}(\log\gapf)$ grows, the fitted slope is forced to $2$; on a record schedule with a bounded number of records per decade, $\mathrm{sd}(\log\gapf)$ grows with the number of decades spanned (Appendix~\ref{app:certificate}).
\end{corollary}Under the corollary's hypotheses, a fitted slope differing from $2$ by more than $\log\kappa/\mathrm{sd}(\log\gapf)$ would refute the local theory. We evaluate this prediction on our $216$-instance test suite in Section~\ref{sec:experiments}. The fitted slope has median $1.986$. Across slope-measurable instances, $96.1\%$ lie between $1.8$ and $2.2$.
\section{Experiments}\label{sec:experiments}

We test five falsifiable predictions from our theory. Exact support freezing follows from Lemma~\ref{lem:absorb}, while Theorem~\ref{thm:rate} bounds the last-iterate convergence rate by the frozen-face spectral radius. That radius admits a closed form in the matrix $\Bmat$ and the saturation levels. Corollary~\ref{cor:theta} gives the exponent $\theta=2$. A ratio certificate with a non-accumulating constant comes from Theorem~\ref{thm:certificate}, and the heuristic in Remark~\ref{rem:mechanism} yields the inequality $\etainf\,\sigmax(\Bmat)\le1$ at realized limits. We evaluate all five statements on $216$ random matrix games and monitor extensive-form runs with the analogous update energy (Table~\ref{tab:main}).

\begin{table}[!t]
  \centering
  \caption{Headline outcomes on the $216$-instance testbed.}
  \label{tab:main}
  \small
  \begin{tabular*}{\textwidth}{@{\extracolsep{\fill}}llrr@{}}
    \toprule
    Prediction & Source & Instances & Reading \\
    \midrule
    Exact-zero pattern fixed from $T/2$ on & Lem.~\ref{lem:absorb} & 216 & $99.5\%$ \\
    Fitted rate over finite-difference radius & Thm.~\ref{thm:rate} & 155 & median $1.000$ (p10--p90 $0.984$--$1.016$) \\
    Closed form equals finite difference & Lem.~\ref{lem:jacobian} & 32 & $2.4\times10^{-10}$ \\
    Increment exponent $\theta=2$, interval proxy & Cor.~\ref{cor:theta} & 129 & $1.986$; $96.1\%$ \\
    Increment exponent $\theta=2$, matched time & Cor.~\ref{cor:theta} & 24 & $2.002$ ($1.981$--$2.086$) \\
    $\etainf\sigmax(\Bmat)\le1$ at limits\textsuperscript{$\dagger$} & Rem.~\ref{rem:mechanism} & 184 & max $0.9985$ \\
    Ratio-law deviation $\kappa_{\mathrm{emp}}$, interval proxy & Thm.~\ref{thm:certificate} & 124 & median $1.63$ \\
    Progress monitoring, Kuhn/Leduc & empirical & 2 & $1.80$, $3.22$ \\
    \bottomrule
  \end{tabular*}

\smallskip
{\footnotesize Ratios and percentage points are summarized across instances; p10--p90 is the 10th to 90th percentile range.  \textsuperscript{$\dagger$}Heuristic of Remark~\ref{rem:mechanism}, empirically consistent with Proposition~\ref{prop:selfstab}.\par}
\end{table}

\begin{figure}[!t]
  \centering
  \includegraphics[width=0.85\textwidth]{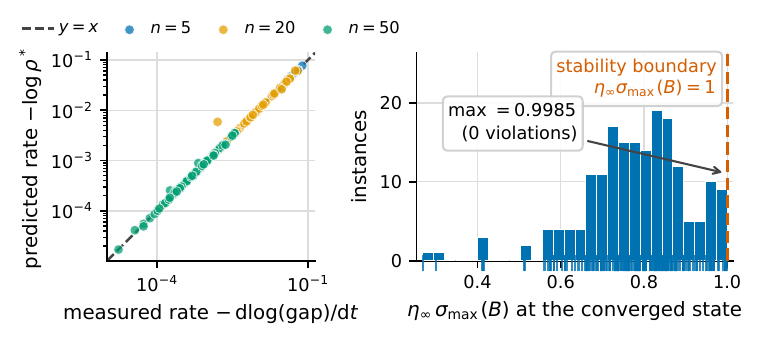}
  \caption{Left: predicted rate against the measured gap slope on the $155$ instances with a fitted slope.  Right: $\etainf\,\sigmax(\Bmat)$ on the $184$ measurable instances, with the stability boundary dashed.}
  \label{fig:spectral}
  \label{fig:etasigma}
\end{figure}

We construct a benchmark of synthetic game instances across several dimensions, and we provide the full protocol in Appendix~\ref{app:experiments}. We vary the dimension $n\in\{5,20,50\}$. We test three payoff distributions: Gaussian, uniform on $[-1,1]$, and sparse $\pm1$ with density $0.3$. We apply two matrix generators: raw payoff matrices and antisymmetrized matrices $(G-G^\top)/\sqrt2$. Using $12$ random seeds for each configuration yields $3\times3\times2\times12=216$ instances. For each instance, we use linear programming to compute the equilibrium, its supports, the game value $v^*$, and the payoff margin. We run \algname{} for $T=2\times10^5$ iterations, increasing the horizon to $10^6$ iterations when $n=50$. Our implementation uses the \emph{unclipped} prediction step $r^{(t)}=\rt^{(t)}+m^{(t)}-\gamma\ones$. We omit clipping because it introduces small positive off-support ``dust'' that prevents exact support freezing.

\textbf{1. Exact support freezing and the local rate.}
The exact-zero pattern of the iterate $\rt$ freezes before the halfway iteration $T/2$ on $99.5\%$ of instances.
This early stabilization is consistent with Lemma~\ref{lem:absorb}.
Once the gap falls sufficiently, we store a state snapshot on the instances with a fitted slope.
At each snapshot, we compute a finite-difference Jacobian, an approximation to the matrix of partial derivatives obtained by perturbing each coordinate slightly.
This linearized operator reproduces the measured semilog slope at a median ratio of $1.000$ and a correlation of $0.9982$ (Figure~\ref{fig:spectral}, left).
The empirical decay rates match the envelope of Theorem~\ref{thm:rate}(iii).

\textbf{2. The closed-form spectral radius matches finite differences.}
Lemmas~\ref{lem:jacobian} and~\ref{lem:centering} express the eigenvalues of the one-step update in closed form using the saturation levels and the matrix $\Bmat = \Atil - v^*\ones\ones^\top$.
To verify this analytical characterization, we reconstruct $\Bmat$ from the \emph{converged} support and saturation levels.
We then compare the resulting largest non-unit eigenvalue modulus $\rho_2$ against its finite-difference estimate, obtained by perturbing each coordinate of the update map.
Across a $32$-instance sample, the two values agree to a median relative error of $2.4\times10^{-10}$ (Appendix~\ref{app:experiments}).
We also evaluate a degenerate instance whose Jacobian has nine eigenvalues of modulus one.
This game violates the hypotheses of Theorem~\ref{thm:rate}.
Even here, the closed-form formula matches the empirical value.

\textbf{3. The certificate exponent is $\theta = 2$.}
We measure this exponent on a candidate frozen window. Across this window, the exact-zero pattern of the vector $\rt$ stays fixed.
On each interval between recorded times, we compute the accumulated per-unit-time on-support energy, which Lemma~\ref{lem:ndot-identity} links to the increment $\Ndottot$ as an interval proxy.
Comparing this proxy against the geometric-mean gap over the interval under a representability guard (Appendix~\ref{app:experiments}) yields a median slope of $1.986$. Of the slope-measurable instances, $96.1\%$ fall between $1.8$ and $2.2$ (Figure~\ref{fig:theta}).

\begin{figure}[!t]
  \centering
  \includegraphics[width=0.8\textwidth]{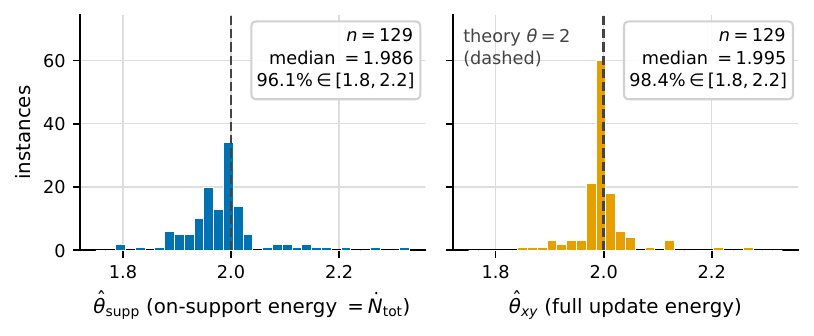}
  \caption{Fitted exponent $\theta$ on the $129$ usable instances for two estimators.  The dashed line is the predicted value $\theta=2$.}
  \label{fig:theta}
\end{figure}

\textbf{4. Self-stabilization, the ratio certificate, and the point estimate.}
The realized dynamics satisfy the spectral condition of Theorem~\ref{thm:rate} on all $184$ measurable instances. The remaining instances reach an exactly zero gap in finitely many rounds. The product $\etainf\,\sigmax(\Bmat) \le 1$ holds on every measurable instance, in line with the heuristic of Remark~\ref{rem:mechanism}. An interval proxy for the ratio law of Theorem~\ref{thm:certificate} yields an empirical deviation $\kappa_{\mathrm{emp}}$ with median $1.63$, bounded across all tested instances by a first-order benchmark constant of median $1.7\times10^{8}$. Because this worst-case bound is loose by eight orders of magnitude, the empirical deviation provides the informative diagnostic.

A matched-time test eliminates the interval proxy. On a stratified subset of $30$ games, we record the clipping masks, the instantaneous on-support energy, and the exact gap at identical rounds. We keep the $24$ games with at least $20$ frozen records above the gap floor. Across multi-decade windows, the instantaneous slope has median $2.002$. The ratio $\gapf/\sqrt{\Ndottot}$ varies over a window by a median factor of $5.5$. A single outlier delays support freezing until late in the run and reaches a variation factor of $483$. Read forward as a point estimate of the gap, the guarded certificate achieves a worst deviation of median $1.50$ and stays below $2$ on $91.5\%$ of instances.

\textbf{5. Progress monitoring in extensive-form games (empirical).}
We empirically track convergence across information sets, the indistinguishable decision states of a player, using the accumulated update energy $G(t)=\sum_{s\le t}\sum_j\norm{g_j^{(s)}}^2$. This metric yields one scalar per iteration without an extra pass. Monitoring its per-unit-time increment at recorded times captures the full correction energy. No theorem covers this quantity. We compare this increment against the NashConv gap, the sum of unilateral best-response gains across players. On Kuhn and Leduc poker, the log--log correlation between the increment and the NashConv gap reaches $0.998$ and $0.997$. The increment of the aggregate squared regret norm correlates more weakly on Kuhn at $0.92$ (Appendix~\ref{app:exp-protocols}). Computing the NashConv gap requires a best-response pass, a game-tree traversal that finds each player's optimal reply to the current strategies. Our sparse record schedule uses $126\times$ and $64\times$ fewer best-response passes than an every-iteration schedule.

\textbf{6. The spectral condition on games derived from extensive form.}
The spectral condition and our closed-form radius hold on games with non-unique equilibria, extending beyond Assumption~\ref{asm:unique-strict}.
We evaluate both on Kuhn poker, Goofspiel-4 (the instance of \citealp{cai2025lastiterate}), Blotto-5, Blotto-10, and rock--paper--scissors with a duplicated action.
We represent the sequential benchmarks in their reduced normal forms, which are payoff matrices over complete contingency plans with strategically identical plans merged.
At the realized limit, the spectral quantity $\etainf\sigmax(\Bmat)$ stays strictly below one across all five games, ranging from $0.36$ to $0.83$.
We compare our closed-form radius against the spectral radius of the finite-difference Jacobian, a numerical derivative of the update map computed by perturbing each coordinate.
The two match to within $2.5\times10^{-5}$.

We also evaluate IR-PCFR$^+$ directly on extensive-form games.
Here we form the finite-difference Jacobian of one iteration over the regret vectors at all information sets, which are the decision points that a player cannot distinguish given their observations.
To isolate the contracting modes, we report the largest retained modulus, defined as the largest eigenvalue modulus not exceeding $1-10^{-5}$.
Across the nine measurable games, this quantity ranges from $0.8688$ on Kuhn to $0.999990$ on five-rank Leduc, both at the rate-fit snapshot and at $0.5T$.
This statistic lies strictly below one by construction.
Testing stability requires the unfiltered spectrum.
The full-spectrum maximum at the late probe equals $1.00000$ on three games and reaches $1.00057$ on Liar's dice, where the finite-difference probe fails to separate a unit cluster from an unstable mode.
When no unit cluster interferes, the numerical spectrum captures the empirical convergence rate.
On six-card Kuhn, the late radius of $0.960291$ reproduces the fitted factor of $0.960142$ (Appendix~\ref{app:add-efg}).

\begin{table}[!t]
\centering\small\setlength{\tabcolsep}{2pt}
\caption{\algname{} on reduced normal forms of extensive-form games.}
\label{tab:rnf}
\begin{tabular*}{\textwidth}{@{\extracolsep{\fill}}lrrcccc@{}}
\toprule
game & size & support & $\etainf\sigmax(\Bmat)$ & $\rho$ closed form / FD & rel.\ gap & slope ratio \\
\midrule
Kuhn poker & $27\times64$ & $8\times4$ & 0.3925 & 0.94990 / 0.94993 & $2\times10^{-5}$ & 1.0044 \\
Goofspiel-4 & $7808\times7808$ & $160\times160$ & 0.3592 & 0.98971 / 0.98972 & $7\times10^{-6}$ & 1.0003 \\
Blotto-5 & $21\times21$ & $9\times9$ & 0.5076 & 0.89930 / 0.89930 & $4\times10^{-10}$ & -- \\
Blotto-10 & $66\times66$ & $48\times48$ & 0.8318 & 0.99774 / 0.99775 & $7\times10^{-6}$ & 1.0034 \\
RPS, duplicated action & $4\times3$ & $4\times3$ & 0.8113 & 0.88035 / 0.88035 & $8\times10^{-8}$ & 1.0274 \\
\bottomrule
\end{tabular*}

\smallskip
{\footnotesize ``rel.\ gap'' is the relative difference between the closed-form non-unit radius and the finite-difference value; ``slope ratio'' is measured over predicted semilog gap slope (games with $R^2>0.9$).  All equilibria are non-unique.\par}
\end{table}

\textbf{7. Larger games, random starts, and the family $A_\lambda$.}
Scaling the action space up to sizes $n=100$ and $200$ leaves our theoretical condition intact.
Every tested instance satisfies $\etainf\sigmax(\Bmat)<1$, and the exponent $\theta$ maintains a median of $2.065$.
Random starting points give the same picture.
Across $1962$ random initializations, the final state satisfies $\etainf\sigmax(\Bmat)<1$ on every run, and about two thirds of the runs reach a final gap of at most $10^{-8}$.
Sequential decision structure preserves this behavior as well.
On additional extensive-form games, the gap decreases on every game, and the fitted exponent is close to two where it is identifiable.
Perturbed runs on the parametric family $A_\lambda$ satisfy the condition at every realized limit (Appendix~\ref{app:add-alambda}).
\section{Conclusion}\label{sec:conclusion}

Regret matching has driven equilibrium computation in large games for two decades, and for as long its last iterate has resisted analysis. We provide the first complete account of why a regret-matching dynamic, executed exactly as published, converges at a linear rate in the last iterate. One mechanism carries the analysis. The regret norm of \iregprm{} acts as an inverse step size, and it saturates in every zero-sum matrix game. From saturation we prove last-iterate convergence to a single Nash equilibrium with no assumptions on the game, which resolves a question left open in concurrent work. Near a unique equilibrium we obtain the first closed-form linear convergence rate for a regret-matching dynamic without restarts, smoothing, or step-size parameters. The rate is an explicit function of the singular values of the game. The theory predicts the measured convergence rate with a median ratio of $1.000$ across $216$ games, and its step-size condition holds at every measurable limit. We also show that a scalar quantity already computed by the algorithm decays proportionally to the square of the Nash gap. On Kuhn and Leduc poker the analogous quantity tracks the gap between best-response passes scheduled two orders of magnitude apart. A last iterate with a proved rate is what lets a solver drop the averaged strategy, the second copy of the model that large-scale solvers store and update. Two questions remain open. Where does the norm saturate, and what is the rate when the equilibrium is not unique?

\clearpage
\label{sec:refstart}
\bibliographystyle{dilab_ref}
\bibliography{references}

\appendix
\makeappendixtoc

\section{Proofs for the Local Theory}
\label{app:local}

Throughout this appendix, one round of \algname{} (Algorithm~2 of
\citealp{zhang2025scale}, strict version: the shifted regret vector is kept
unclipped and clipping is applied once, in the update) from a state
$z=(\rt_1,\rt_2)$ with $\rt_1,\rt_2\neq 0$ reads, for player~1,
\begin{align}
&\xt=\rt_1/\norm{\rt_1}_1, \qquad \yt=\rt_2/\norm{\rt_2}_1, \qquad
  m_1 = A\yt, \notag\\
&\gamma_1 \ \text{the unique solution of}\
  \norm{\pos{\rt_1+m_1-\gamma_1\ones}}_2=\norm{\rt_1}_2=:\Nt_1, \notag\\
&r_1:=\rt_1+m_1-\gamma_1\ones \ \ (\text{kept unclipped}), \qquad
  p_1:=\pos{r_1}, \qquad x:=p_1/\norm{p_1}_1, \notag\\
&\mathrm{res}_1:=A(y-\yt), \qquad
  g_1:=\mathrm{res}_1-\inner{\mathrm{res}_1}{x}\ones, \qquad
  \rt_1' = \pos{r_1+g_1},
\label{eq:round-recap}
\end{align}
and identically for player~2 with $A$ replaced by $-A^\top$ (so
$m_2=-A^\top\xt$, $\mathrm{res}_2=-A^\top(x-\xt)$, and $y:=p_2/\norm{p_2}_1$).
Recall that
$\varphi(\gamma):=\norm{\pos{v-\gamma\ones}}_2$ is continuous and
nonincreasing in $\gamma$, strictly decreasing wherever it is positive, with
$\varphi(\gamma)\to\infty$ as $\gamma\to-\infty$ and $\varphi(\gamma)\to0$ as
$\gamma\to+\infty$; since $\Nt_1>0$, the defining equation for $\gamma_1$ has
a unique solution, and by construction the clipped shifted vector is
norm-preserving: $\norm{p_1}_2=\Nt_1$.  Two elementary facts are used
repeatedly: $\inner{\pos{u}}{u}=\norm{\pos{u}}_2^2$ for every $u$, and the
orthogonality
\begin{equation}
\label{eq:orth}
\inner{g_1}{p_1}
=\inner{\mathrm{res}_1}{p_1}-\inner{\mathrm{res}_1}{x}\,\norm{p_1}_1
=0,
\end{equation}
which holds by the definition $x=p_1/\norm{p_1}_1$ and
$\inner{\ones}{p_1}=\norm{p_1}_1$ (as $p_1\ge0$).

\subsection{Fixed points: an algebraic proof of Proposition~\ref{prop:fixedpoint}}
\label{app:local:p1}

We first isolate two statewise lemmas.  Both concern a single round at a
single state; no property of the trajectory is used anywhere in this
subsection.

\begin{lemma}[$\gamma$-shift energy identity]
\label{lem:energy-identity}
Let $\rt\in\R^n_{\ge0}\setminus\{0\}$ and $m\in\R^n$, put
$\Nt:=\norm{\rt}_2$, let $\gamma$ be the unique solution of
$\norm{\pos{\rt+m-\gamma\ones}}_2=\Nt$, and set
$r:=\rt+m-\gamma\ones$, $p:=\pos{r}$, $x:=p/\norm{p}_1$, and
$\xt:=\rt/\norm{\rt}_1$.  Then
\begin{equation}
\label{eq:energy-identity}
\inner{m}{x}-\gamma \;=\; \frac{\norm{p-\rt}_2^{\,2}}{2\,\norm{p}_1}\;.
\end{equation}
In particular $\inner{m}{x}\ge\gamma$, with equality if and only if
$p=\rt$; and
\begin{equation}
\label{eq:energy-lb}
\inner{m}{x}-\gamma \;\ge\; \frac{\Nt}{2n}\,\norm{x-\xt}_2^{\,2}\;.
\end{equation}
\end{lemma}

\begin{proof}
Since $p=\pos{r}$ we have $\inner{p}{r}=\norm{p}_2^2=\Nt^2$.  Substituting
$r=\rt+m-\gamma\ones$ and $\inner{p}{\ones}=\norm{p}_1$,
\[
\inner{p}{\rt} + \inner{p}{m} - \gamma\norm{p}_1 \;=\; \Nt^2 .
\]
Because $\norm{p}_2=\norm{\rt}_2=\Nt$,
\[
\Nt^2-\inner{p}{\rt}
=\tfrac12\bigl(\norm{p}_2^2+\norm{\rt}_2^2-2\inner{p}{\rt}\bigr)
=\tfrac12\norm{p-\rt}_2^2 ,
\]
and dividing by $\norm{p}_1>0$ (note $\norm{p}_1\ge\norm{p}_2=\Nt>0$) gives
\eqref{eq:energy-identity}.  Nonnegativity and the equality case are
immediate.  For \eqref{eq:energy-lb}, put $c:=\norm{p}_1$.  Since
$p=c\,x$, $\rt=\norm{\rt}_1\xt$ and
$\ones^\top(p-\rt)=c-\norm{\rt}_1$,
\[
c\,(x-\xt)
\;=\;(p-\rt)-\xt\,\ones^\top(p-\rt)
\;=\;\bigl(I-\xt\ones^\top\bigr)(p-\rt).
\]
We claim $\norm{I-\xt\ones^\top}_{\mathrm{op}}\le\sqrt n$.  For $n=1$ the
matrix is zero.  For $n\ge2$ write $a:=\xt$, $b:=\ones$, so that
$b^\top a=1$, and let
$G:=(I-ab^\top)^\top(I-ab^\top)=I-ba^\top-ab^\top+\norm{a}_2^2\,bb^\top$.
The matrix $G$ is the identity on $\{a,b\}^\perp$ and maps
$\mathrm{span}\{a,b\}$ into itself.  If $a$ and $b$ are parallel, then
$ab^\top$ is an orthogonal projection and the norm is at most $1$.
Otherwise $\mathrm{span}\{a,b\}$ has dimension two, $Ga=0$ because
$(I-ab^\top)a=0$, and the trace of $G$ on this plane is
$\operatorname{tr}G-(n-2)=2-2\,b^\top a+\norm{a}_2^2\norm{b}_2^2
=\norm{a}_2^2\norm{b}_2^2$, so the second eigenvalue there is
$\norm{a}_2^2\norm{b}_2^2$.  Hence
$\norm{I-\xt\ones^\top}_{\mathrm{op}}^2\le\max\{1,\,n\norm{\xt}_2^2\}\le n$,
using $\norm{\xt}_2\le\norm{\xt}_1=1$.  Therefore
$c^2\norm{x-\xt}_2^2\le n\,\norm{p-\rt}_2^2
=2n\,c\,\bigl(\inner{m}{x}-\gamma\bigr)$ by \eqref{eq:energy-identity},
and dividing by $2nc$ and using $c=\norm{p}_1\ge\Nt$ gives
\eqref{eq:energy-lb}.
\end{proof}

\begin{lemma}[Clipping rigidity]
\label{lem:clip-rigidity}
Let $r,g\in\R^n$ with $\inner{g}{\pos{r}}=0$.  Then
$\norm{\pos{r+g}}_2\ge\norm{\pos{r}}_2$, and if equality holds then
\[
\pos{r+g}\;=\;\pos{r},
\qquad\text{and}\qquad
g_i=0 \ \text{ for every } i \text{ with } r_i>0 .
\]
\end{lemma}

\begin{proof}
By the orthogonality hypothesis,
$\norm{\pos{r}+g}_2^2=\norm{\pos{r}}_2^2+\norm{g}_2^2$.  Define
\[
D:=\norm{\pos{r}+g}_2^2-\norm{\pos{r+g}}_2^2
=\sum_i d_i,
\qquad
d_i:=\bigl(\pos{r_i}+g_i\bigr)^2-\pos{r_i+g_i}^2,
\]
so that $\norm{\pos{r+g}}_2^2=\norm{\pos{r}}_2^2+\norm{g}_2^2-D$.  We claim
$0\le d_i\le g_i^2$ for every $i$, and we characterize the equality cases of
the upper bound.
\begin{itemize}
\item \emph{Case $r_i\ge0$, $r_i+g_i\ge0$:}\ then
$d_i=(r_i+g_i)^2-(r_i+g_i)^2=0$, and $g_i^2-d_i=g_i^2\ge0$, with equality iff
$g_i=0$.
\item \emph{Case $r_i\ge0$, $r_i+g_i<0$:}\ then $d_i=(r_i+g_i)^2\ge0$ and
$g_i^2-d_i=-r_i(r_i+2g_i)$.  Here $g_i<-r_i\le0$, hence
$r_i+2g_i=(r_i+g_i)+g_i<0$, so $g_i^2-d_i=r_i\,\bigl(-(r_i+2g_i)\bigr)\ge0$,
with equality iff $r_i=0$.
\item \emph{Case $r_i<0$:}\ then $\pos{r_i}=0$, so
$d_i=g_i^2-\pos{r_i+g_i}^2$ and $g_i^2-d_i=\pos{r_i+g_i}^2\ge0$, with
equality iff $r_i+g_i\le0$.  Also $d_i\ge0$: if $r_i+g_i\le0$ this is clear,
and if $r_i+g_i>0$ then $0<r_i+g_i<g_i$, so $\pos{r_i+g_i}^2<g_i^2$.
\end{itemize}
Thus $0\le D\le\norm{g}_2^2$, which proves the inequality.  If
$\norm{\pos{r+g}}_2=\norm{\pos{r}}_2$, then $D=\norm{g}_2^2$, i.e.\
$g_i^2-d_i=0$ for every $i$.  By the case analysis this forces: $g_i=0$
whenever $r_i\ge0$ and $r_i+g_i\ge0$; $r_i=0$ whenever $r_i\ge0$ and
$r_i+g_i<0$; and $r_i+g_i\le0$ whenever $r_i<0$.  In the first case
$\pos{r_i+g_i}=r_i=\pos{r_i}$; in the second, $\pos{r_i+g_i}=0=r_i=\pos{r_i}$;
in the third, $\pos{r_i+g_i}=0=\pos{r_i}$.  Hence $\pos{r+g}=\pos{r}$
componentwise.  Finally, if $r_i>0$ the second case is excluded (it forces
$r_i=0$), so the first case applies and $g_i=0$.
\end{proof}

\begin{proof}[Proof of Proposition~\ref{prop:fixedpoint}]
Let $z=(\rt_1,\rt_2)$ with $\rt_1,\rt_2\neq0$ and let
$(\xt,\yt)=(\rt_1/\norm{\rt_1}_1,\ \rt_2/\norm{\rt_2}_1)$.  We prove:
\emph{(a)} if $(\xt,\yt)$ is a Nash equilibrium then $\Phi(z)=z$, with
$\gamma_1=v^*$ and $\gamma_2=-v^*$;
\emph{(b)} conversely, if $\Phi(z)=z$ then $(\xt,\yt)$ is a Nash equilibrium,
$\gamma_1=v^*$, $\gamma_2=-v^*$, and each $\rt_i$ is a positive multiple of
the corresponding equilibrium strategy (in particular
$\supp(\rt_1)=\supp(\xt)$ and $\supp(\rt_2)=\supp(\yt)$).
Together, (a) and (b) give
$\Fix(\Phi)\cap\{\rt_1,\rt_2\neq0\}
=\{(c_1x^\dagger,c_2y^\dagger):\ (x^\dagger,y^\dagger)\ \text{a Nash
equilibrium},\ c_1,c_2>0\}$,
which is the two-parameter cone $E_+$ when the equilibrium is unique.

We use the standard complementary-slackness facts for zero-sum games: if
$(x^\dagger,y^\dagger)$ is a Nash equilibrium with value
$v^*=x^{\dagger\top}Ay^\dagger$, then, because $x^\dagger$ best-responds
to $y^\dagger$, $\max_i(Ay^\dagger)_i=v^*$ and every $i\in\supp(x^\dagger)$
attains this maximum; symmetrically $\min_j(x^{\dagger\top}A)_j=v^*$ and every
$j\in\supp(y^\dagger)$ attains this minimum.

\medskip\noindent
\emph{(a)}\ \ Suppose $(\xt,\yt)$ is a Nash equilibrium; write
$c_1:=\norm{\rt_1}_1$, $c_2:=\norm{\rt_2}_1$, so $\rt_1=c_1\xt$,
$\rt_2=c_2\yt$.  Then $m_1=A\yt$, and $\gamma=v^*$ solves the $\gamma$-shift
equation: for $i\in\supp(\xt)$ we have $(A\yt)_i=v^*$, so
$(\rt_1+m_1-v^*\ones)_i=\rt_{1,i}>0$; for $i\notin\supp(\xt)$ we have
$\rt_{1,i}=0$ and $(A\yt)_i\le v^*$, so
$(\rt_1+m_1-v^*\ones)_i\le0$.  Hence
$\pos{\rt_1+m_1-v^*\ones}=\rt_1$, whose $\ell_2$ norm is $\Nt_1$; by
uniqueness of the solution, $\gamma_1=v^*$ and $p_1=\rt_1$, so
$x=\xt$.  Symmetrically (with $m_2=-A^\top\xt$ and the minimum conditions),
$\gamma_2=-v^*$, $p_2=\rt_2$, and $y=\yt$.  Then
$\mathrm{res}_1=A(y-\yt)=0$, hence $g_1=0$ and
$\rt_1'=\pos{r_1+0}=\pos{r_1}=p_1=\rt_1$; likewise $\rt_2'=\rt_2$.  So
$\Phi(z)=z$.

\medskip\noindent
\emph{(b)}\ \ Suppose $\Phi(z)=z$, i.e.\ $\rt_1'=\rt_1$ and
$\rt_2'=\rt_2$.

\emph{Step 1 (rigidity: the update returns the clipped shifted vector).}
Since $\rt_1'=\rt_1$, we have
$\norm{\pos{r_1+g_1}}_2=\norm{\rt_1}_2=\Nt_1=\norm{p_1}_2$.  By
\eqref{eq:orth}, Lemma~\ref{lem:clip-rigidity} applies (with $r=r_1$,
$g=g_1$, noting $\pos{r_1}=p_1$) and its equality case yields
$\pos{r_1+g_1}=p_1$.  Hence $\rt_1=\rt_1'=p_1$.  Symmetrically
$\rt_2=p_2$.

\emph{Step 2 (energy identity: pre-iterate equals main iterate).}
By Lemma~\ref{lem:energy-identity} applied to player~1,
$p_1=\rt_1$ forces $\inner{m_1}{x}-\gamma_1=0$, and
$x=p_1/\norm{p_1}_1=\rt_1/\norm{\rt_1}_1=\xt$.  Symmetrically $y=\yt$.

\emph{Step 3 (residuals vanish).}
From $y=\yt$: $\mathrm{res}_1=A(y-\yt)=0$, hence $g_1=0$; symmetrically
$g_2=0$.

\emph{Step 4 (componentwise reading of the fixed-point equation).}
By Step 1, $\rt_1=p_1=\pos{\rt_1+m_1-\gamma_1\ones}$.  Componentwise: if
$\rt_{1,i}>0$ then $\rt_{1,i}+m_{1,i}-\gamma_1=\rt_{1,i}$, i.e.\
$m_{1,i}=\gamma_1$; if $\rt_{1,i}=0$ then
$\rt_{1,i}+m_{1,i}-\gamma_1\le0$, i.e.\ $m_{1,i}\le\gamma_1$.  Since
$y=\yt$, $m_1=A\yt=Ay$; therefore, for every strategy $x'\in\simplex{n_1}$,
$x'^{\top}A\yt=\inner{m_1}{x'}\le\gamma_1$, with equality at $x'=\xt$
(which is supported on $\{i:\rt_{1,i}>0\}$): $\xt$ is a best response to
$\yt$ and $\gamma_1=\xt^\top A\yt$.  Symmetrically, for player~2,
$(-A^\top\xt)_j=\gamma_2$ on $\supp(\rt_2)$ and $\le\gamma_2$ elsewhere, so
$\yt\in\arg\max_{y'}\,(-\xt^\top Ay')=\arg\min_{y'}\,\xt^\top Ay'$ and
$\gamma_2=-\xt^\top A\yt$.  Hence $(\xt,\yt)$ is a Nash equilibrium; since
all equilibria of a zero-sum game share the value,
$\gamma_1=\xt^\top A\yt=v^*$ and $\gamma_2=-v^*$.  Finally
$\rt_1=\norm{\rt_1}_1\,\xt$ and $\rt_2=\norm{\rt_2}_1\,\yt$ by definition of
the pre-iterates, which is the claimed cone structure, and
$\supp(\rt_i)$ equals the support of the corresponding equilibrium strategy.
\end{proof}

\begin{remark}[A trajectory-free proof]
\label{rem:p1-algebraic}
The proof above is purely algebraic: it inspects a single round at a single
state and never invokes summability along trajectories, whereas the
telescoping machinery of \citet{zhang2025scale} yields
$\inner{m}{x}-\gamma\to0$ only along a trajectory.  The exact identity
\eqref{eq:energy-identity} is of independent interest.
The bound \eqref{eq:energy-lb} is exactly the per-round scalar lower bound
\eqref{eq:beta-free} on which the saturation analysis rests, and it is used
verbatim in the proof of Lemma~\ref{lem:sat-equiv} below.
Its equality case pinpoints the fixed
points without any limiting argument.
\end{remark}

\subsection{Saturation equivalence: proof of Lemma~\ref{lem:sat-equiv}}
\label{app:local:satequiv}

\begin{remark}[Tighter constant]\label{rem:sat-constant}
The constant $1/(2n_i)$ in \eqref{eq:beta-free} and
\eqref{eq:sat-explicit} equals the one that the element-wise inequality of
\citet[Theorem~3.2]{zhang2025scale} yields.  Here it follows from the exact
identity \eqref{eq:energy-identity} and the norm of the rank-one projection
$\xt\ones^\top$ (Lemma~\ref{lem:energy-identity}), so the saturation
analysis stays self-contained.
\end{remark}

The saturation analysis uses one elementary fact.
\begin{lemma}[Norm monotonicity]\label{lem:normmono}
Let $r \in \R^n$ and $g \in \R^n$ satisfy $\inner{g}{\pos{r}} = 0$. Then $\norm{\pos{r + g}}_2 \ge \norm{\pos{r}}_2$. Consequently, along any trajectory of \iregprm{}, $\Nt(t{+}1) \ge \Nt(t)$.
\end{lemma}
\begin{proof}
The inequality is the first claim of Lemma~\ref{lem:clip-rigidity}.  Along one
round of \iregprm{}, $\rt_i^{(t+1)}=\pos{r_i^{(t)}+g_i^{(t)}}$ with
$\inner{g_i^{(t)}}{\pos{r_i^{(t)}}}=0$ (F1), and the $\gamma$-shift is defined
so that $\norm{\pos{r_i^{(t)}}}_2=\norm{\rt_i^{(t)}}_2$
(Definition~\ref{def:gammashift}); the inequality therefore reads
$\Nt_i(t{+}1)\ge\Nt_i(t)$.  \rev{If $\Nt_i(t)=0$ the claim is the trivial $\Nt_i(t{+}1)\ge0$; the shift argument is needed only when $\Nt_i(t)>0$.}
\end{proof}

The proof below is self-contained: it uses only the energy identity
(Lemma~\ref{lem:energy-identity}), the structural facts (F2)--(F3), norm
monotonicity (Lemma~\ref{lem:normmono}), and the elementary nonnegativity
of the summed external regrets in a zero-sum game (cf.\
\citealp[Fact~2.2]{zhang2025scale}), which we reprove in one line.

\begin{proof}[Proof of Lemma~\ref{lem:sat-equiv}]
We write the argument for player~1; player~2 is symmetric (replace $A$ by
$-A^\top$).  Set $u^{(t)}:=Ay^{(t)}$, the payoff vector at the opponent's
main iterate, and
$\beta_t^{1}=\inner{m_1^{(t)}}{x^{(t)}}-\gamma_1^{(t)}$.  For the trivial
rounds $t<t_1^*$ (where $\rt_1^{(t)}=0$) the algorithm sets
$m_1^{(t)}=0$, applies no shift ($\gamma_1^{(t)}=0$), and plays
$x^{(t)}=\xt^{(t)}$; hence $\beta_t^{1}=0$ and
$\norm{x^{(t)}-\xt^{(t)}}=0$ on those rounds, consistently with the
convention in the statement.  \rev{Explicitly (Algorithm~\ref{alg:iregprm}): when $\rt_1^{(t)}=\mathbf 0$ the pre-iterate $\xt^{(t)}$ is carried over from the previous round, starting from the uniform strategy, $m_1^{(t)}=\mathbf 0$, $r_1^{(t)}=\mathbf 0$, $x^{(t)}=\xt^{(t)}$, and the residual is $\mathrm{res}_1^{(t)}=u^{(t)}-m_1^{(t)}=u^{(t)}$, so $g_1^{(t)}=u^{(t)}-\inner{u^{(t)}}{x^{(t)}}\ones$ and $\rt_1^{(t+1)}=\pos{g_1^{(t)}}$; this is the branch used in Step~2.}

\emph{Step 1 (per-round lower bound; first part of \eqref{eq:beta-free}).}
For $t\ge t_1^*$, apply Lemma~\ref{lem:energy-identity} at the state
$\rt=\rt_1^{(t)}$ with $m=m_1^{(t)}$ and $n=n_1$: by
\eqref{eq:energy-identity} and \eqref{eq:energy-lb},
\[
\beta_t^{1}
\;=\;\frac{\norm{p_1^{(t)}-\rt_1^{(t)}}_2^{\,2}}{2\norm{p_1^{(t)}}_1}
\;\ge\;\frac{\Nt_1(t)}{2n_1}\,\norm{x^{(t)}-\xt^{(t)}}_2^{\,2}
\;\ge\;0 .
\]

\emph{Step 2 (entrywise telescoping).}
Since $\pos{v}\ge v$ entrywise, the update in \eqref{eq:round-recap}
satisfies, entrywise,
\[
\rt_1^{(t+1)}=\pos{r_1^{(t)}+g_1^{(t)}}
\;\ge\;\rt_1^{(t)}+m_1^{(t)}+g_1^{(t)}-\gamma_1^{(t)}\ones .
\]
Substituting $m_1^{(t)}+g_1^{(t)}
=u^{(t)}-\inner{\mathrm{res}_1^{(t)}}{x^{(t)}}\ones$ and
$\gamma_1^{(t)}+\inner{\mathrm{res}_1^{(t)}}{x^{(t)}}
=\inner{u^{(t)}}{x^{(t)}}-\beta_t^{1}$ (both identities are immediate from
$\mathrm{res}_1^{(t)}=u^{(t)}-m_1^{(t)}$ and the definition of
$\beta_t^{1}$),
\[
\rt_1^{(t+1)}
\;\ge\;\rt_1^{(t)}
+\bigl(u^{(t)}-\inner{u^{(t)}}{x^{(t)}}\ones\bigr)
+\beta_t^{1}\,\ones ;
\]
the same inequality holds on the trivial rounds (there
$\rt_1^{(t+1)}=\pos{g_1^{(t)}}\ge u^{(t)}-\inner{u^{(t)}}{x^{(t)}}\ones$
and $\beta_t^{1}=0$).  Telescoping from $\rt_1^{(1)}=\mathbf{0}$ and
evaluating at the coordinate $a$ attaining the external regret
$\mathrm{Reg}_1^{(T)}:=\max_{a'}\sum_{t\le T}
\bigl(u^{(t)}_{a'}-\inner{u^{(t)}}{x^{(t)}}\bigr)$, and using
$\rt_1^{(T+1)}[a]\le\norm{\rt_1^{(T+1)}}_2=\Nt_1(T{+}1)$,
\begin{equation}
\label{eq:sat-telescope}
\Nt_1(T{+}1)\;\ge\;\mathrm{Reg}_1^{(T)}+\sum_{t\le T}\beta_t^{1} .
\end{equation}

\emph{Step 3 (summing the players; second part of \eqref{eq:beta-free}).}
With $u_2^{(t)}=-A^\top x^{(t)}$, the two external regrets add up to
\[
\mathrm{Reg}_1^{(T)}+\mathrm{Reg}_2^{(T)}
=\max_{x'}\sum_{t\le T}x'^{\top}Ay^{(t)}
-\min_{y'}\sum_{t\le T}x^{(t)\top}Ay'
\;\ge\;Tv^*-Tv^*=0,
\]
where any equilibrium $(x^\circ,y^\circ)$ can be used in the two
optimizations ($x^{\circ\top}Ay\ge v^*$ for every $y$, and
$x^\top Ay^\circ\le v^*$ for every $x$); this is Fact~2.2 of
\citet{zhang2025scale}.  Adding \eqref{eq:sat-telescope} for the two
players and rearranging, for every $T$,
\[
\sum_{t\le T}\bigl(\beta_t^{1}+\beta_t^{2}\bigr)
\;\le\;\Nt_1(T{+}1)+\Nt_2(T{+}1)
-\bigl(\mathrm{Reg}_1^{(T)}+\mathrm{Reg}_2^{(T)}\bigr)
\;\le\;\Nt_{1,\infty}+\Nt_{2,\infty}
\]
under two-sided saturation; since every term is nonnegative (Step~1), the
series converges, proving \eqref{eq:beta-free}.

\emph{Step 4 (second-order path length; \eqref{eq:sat-explicit} and
$\Rightarrow$).}
For $t\ge t_1^*$ with $t_1^*<\infty$, Lemma~\ref{lem:normmono} gives
$\Nt_1(t)\ge\Nt_1(t_1^*)>0$, so Step~1 yields
$\norm{x^{(t)}-\xt^{(t)}}^2\le 2n_1\,\beta_t^{1}/\Nt_1(t_1^*)$; rounds
before $t_1^*$ contribute zero.  Summing and using
$\sum_t\beta_t^{1}\le\Nt_{1,\infty}+\Nt_{2,\infty}$ (Step~3, with
$\beta_t^{2}\ge0$) gives \eqref{eq:sat-explicit}, and in particular the
finiteness of the second-order path length under saturation.

\emph{Step 5 ($\Leftarrow$).}
If $t_1^*=\infty$ then $\Nt_1\equiv0$ saturates trivially.  Otherwise, for
$t\ge t_1^*$ the prediction is $m_1^{(t)}=A\yt^{(t)}$ and (F2) gives
$\norm{g_1^{(t)}}\le(1+\sqrt{n_1})\norm{A}_{\mathrm{op}}
\norm{y^{(t)}-\yt^{(t)}}$ (on the finitely many rounds $t<t_1^*$ the
prediction is $0$ and (F2) is not used), so
$\sum_{t\ge t_1^*}\norm{g_1^{(t)}}^2<\infty$ when the second-order path
length is finite; by (F3), valid for $t\ge t_1^*$, and $D(t)\ge0$,
$\Nt_1(T{+}1)^2=\Nt_1(t_1^*)^2+\sum_{t_1^*\le t\le T}\bigl(\norm{g_1^{(t)}}^2-D(t)\bigr)
\le\Nt_1(t_1^*)^2+\sum_{t\ge t_1^*}\norm{g_1^{(t)}}^2<\infty$, and the nondecreasing
(Lemma~\ref{lem:normmono}) bounded sequence $\Nt_1(t)$ saturates; the same
holds for player~2.
\end{proof}

\begin{proof}[Proof of Theorem~\ref{thm:saturation}]
Write $e_1(t):=\norm{x^{(t)}-\xt^{(t)}}$, $e_2(t):=\norm{y^{(t)}-\yt^{(t)}}$,
$S(t):=\Nt_1(t)+\Nt_2(t)$ and $Q(t):=\Nt_1(t)^2+\Nt_2(t)^2\ge S(t)^2/2$.
Three of the inequalities above hold for every finite horizon, without
any saturation hypothesis.
\emph{(i)} For $t\ge\tau$, Step~1 and Lemma~\ref{lem:normmono}
($\Nt_i(t)\ge\Nt_i(\tau)$) give $\beta_t^{i}\ge\alpha_i\,e_i(t)^2$.
\emph{(ii)} The first inequality of Step~3 holds for every $T$:
$\sum_{t\le T}(\beta_t^{1}+\beta_t^{2})\le S(T{+}1)$, since
$\mathrm{Reg}_1^{(T)}+\mathrm{Reg}_2^{(T)}\ge0$.
\emph{(iii)} For $t\ge\tau$ both predictions are $m_i^{(t)}=A\yt^{(t)}$,
$-A^\top\xt^{(t)}$, so (F3) with $D(t)\ge0$ and (F2) give
$\Nt_1(t{+}1)^2-\Nt_1(t)^2\le\norm{g_1^{(t)}}^2\le K_1^2\,e_2(t)^2$ and
symmetrically $\Nt_2(t{+}1)^2-\Nt_2(t)^2\le K_2^2\,e_1(t)^2$.
Summing (iii) from $\tau$ to $T$ and inserting (i) then (ii),
\[
\begin{aligned}
Q(T{+}1)&\;\le\;Q_\tau+K_1^2\sum_{t=\tau}^{T}e_2(t)^2+K_2^2\sum_{t=\tau}^{T}e_1(t)^2\\
&\;\le\;Q_\tau+\frac{K_1^2}{\alpha_2}\sum_{t\le T}\beta_t^{2}+\frac{K_2^2}{\alpha_1}\sum_{t\le T}\beta_t^{1}
\;\le\;Q_\tau+M\,S(T{+}1).
\end{aligned}
\]
With $Q(T{+}1)\ge S(T{+}1)^2/2$ this is the quadratic inequality
$S^2-2MS-2Q_\tau\le0$ in $S=S(T{+}1)$, whence
$S(T{+}1)\le M+\sqrt{M^2+2Q_\tau}$ for every $T\ge\tau$; for $t\le\tau$
the bound holds by monotonicity.  Each $\Nt_i$ is nondecreasing
(Lemma~\ref{lem:normmono}) and bounded, hence convergent.
If $t_1^*=\infty$, then $\rt_1^{(t)}=\mathbf 0$ and $x^{(t)}=\xt^{(t)}$
for all $t$; hence $\mathrm{res}_2^{(t)}=-A^\top(x^{(t)}-\xt^{(t)})=0$
and $g_2^{(t)}=0$ for $t\ge t_2^*$, so by (F3) $\Nt_2$ is constant from
$t_2^*$ on (and $\Nt_2\equiv0$ if also $t_2^*=\infty$).
\end{proof}

\begin{remark}[The bound is loose but structural]\label{rem:sat-bound-loose}
On the testbed of Section~\ref{sec:experiments} the right-hand side of
\eqref{eq:sat-bound} exceeds the realized $\Nt_{1,\infty}+\Nt_{2,\infty}$ by
a factor between $7\times10^{2}$ and $9\times10^{5}$ (it grows like
$n^{2}\norm{A}^{2}/\Nt_i(\tau)$), and none of the three finite-horizon
inequalities (i)--(iii) is violated in the first $5000$ rounds on any of
the $216$ instances.  Its role is qualitative: it
removes saturation from the hypothesis list of every result in this paper.
\end{remark}
\subsection{Proofs}
\label{sec:local:proofs}

This subsection proves the three results of this section: the absorption
lemma, the closed-form Jacobian together with the value-centering identities,
and the spectral geometry that produces the closed form \eqref{eq:spec} of
Theorem~\ref{thm:rate}.  The two remaining steps of the proof of
Theorem~\ref{thm:rate}, a uniform second-order Taylor expansion and the
bootstrap that converts the spectral picture into the rate
\eqref{eq:linrate}, are standard in technique and are carried out in
Appendix~\ref{app:local:rate}; the algebraic proof of
Proposition~\ref{prop:fixedpoint} and the saturation equivalence of
Lemma~\ref{lem:sat-equiv} are in Appendix~\ref{app:local}.

\subsubsection{One-step absorption: proof of Lemma~\ref{lem:absorb}}
\label{app:local:absorb}

\begin{lemma}[One-step absorption and local analyticity]
\label{lem:absorb}
Let $(x^*,y^*)$ be a strictly complementary equilibrium with margin
$\delta>0$ as in \eqref{eq:margin}, and let
$z_f=(c_1x^*,c_2y^*)$, $c_1,c_2>0$, be any point of the corresponding ray
(a fixed point of $\Phi$ by Proposition~\ref{prop:fixedpoint}).  There is
an open neighborhood $\mathcal{U}_0\subset\R^{n_1+n_2}$ of $z_f$ such that:
\begin{enumerate}[label=(\roman*),itemsep=1pt,topsep=2pt]
\item $\Phi(\mathcal{U}_0\cap\mathcal Z)\subseteq\Ms\cap\mathcal Z$; more precisely, for every
$z\in\mathcal{U}_0\cap\mathcal Z$ the off-support coordinates of $\Phi(z)$ are
\emph{exactly} zero;
\item on $\mathcal{U}_0\cap\Ms\cap\mathcal Z$ the clipping pattern of every $\pos{\cdot}$
operation in one round is constant, and $\Phi_S:=\Phi|_{\Ms}$ is
real-analytic there (so $\Ms$ is forward-invariant near $z_f$; it need
not be invariant globally).
\end{enumerate}
\end{lemma}

\begin{proof}[Proof of Lemma~\ref{lem:absorb}]
Let $z_f=(c_1x^*,c_2y^*)$ with $c_1,c_2>0$.  By
Proposition~\ref{prop:fixedpoint}, $\Phi(z_f)=z_f$ with $\gamma_1=v^*$,
$\gamma_2=-v^*$, and at $z_f$ we have $x=\xt=x^*$, $y=\yt=y^*$,
$p_1=\rt_1$, $p_2=\rt_2$, and $g_1=g_2=0$.

\emph{Step 1 (continuity of the round at $z_f$).}
On a neighborhood of $z_f$ the pre-iterate normalizations are continuous
(denominators $\norm{\rt_i}_1\ge c_i/2$), hence so are $m_1$ and $m_2$.  The
multiplier $\gamma_1$ depends continuously on its inputs: the map
$\varphi(\gamma;v):=\norm{\pos{v-\gamma\ones}}_2$ is jointly continuous and,
at the solution (where $\varphi=\Nt_1>0$), strictly decreasing in $\gamma$;
a standard monotone-crossing argument then gives continuity of the solution
in $(v,\Nt_1)$.  Consequently $p_1$, $x$, $\mathrm{res}_1$, and $g_1$ are
continuous near $z_f$ (the denominator $\norm{p_1}_1\ge\norm{p_1}_2=\Nt_1$
is bounded below), and $g_1\to0$, $\gamma_1\to v^*$ as $z\to z_f$; the same
holds for player~2.

\emph{Step 2 (off-support coordinates are clipped exactly).}
Fix $i\notin S_1$.  As $z\to z_f$,
\[
r_{1,i}=\rt_{1,i}+m_{1,i}-\gamma_1
\;\longrightarrow\;
0+(Ay^*)_i-v^*\;\le\;-\delta ,
\]
by strict complementarity \eqref{eq:margin}.  Choose the neighborhood
$\mathcal{U}_0$ small enough that, for all $z\in\mathcal{U}_0$ and all
$i\notin S_1$, $r_{1,i}\le-\delta/2$ and $\norm{g_1}_\infty\le\delta/4$.
Then $(r_1+g_1)_i\le-\delta/4<0$, so $\rt'_{1,i}=\pos{(r_1+g_1)_i}=0$
\emph{exactly} (and also $p_{1,i}=0$).  The same holds for player~2, using
the second minimum in \eqref{eq:margin}.  Hence
\rev{$\Phi(\mathcal{U}_0\cap\mathcal Z)\subseteq\Ms$}, proving~(i).

\emph{Step 3 (support coordinates are never clipped).}
Fix $i\in S_1$.  As $z\to z_f$, $r_{1,i}\to c_1x^*_i>0$.  Shrinking
$\mathcal{U}_0$, we may assume $r_{1,i}\ge\kappa:=\tfrac{c_1}{2}
\min_{i\in S_1}x^*_i>0$ and $\norm{g_1}_\infty\le\kappa/2$ for all
$z\in\mathcal{U}_0$, so $(r_1+g_1)_i\ge\kappa/2>0$: neither the clipped
shifted vector nor the update clips any support coordinate.  Symmetrically
for player~2.

\emph{Step 4 (frozen pattern and analyticity on $\mathcal{U}_0\cap\Ms$).}
Set $\iota:=\min\{\delta/8,\kappa/4\}$ and shrink $\mathcal{U}_0$ once more
so that $|\gamma_1-v^*|\le\iota$ on $\mathcal{U}_0$ (possible by Step~1).
Then, for $z\in\mathcal{U}_0\cap\Ms$ and every $\gamma$ with
$|\gamma-v^*|\le\iota$ (an interval that contains the solution $\gamma_1$,
and $|\gamma-\gamma_1|\le2\iota$ throughout), the vector
$\rt_1+m_1-\gamma\ones$ is strictly positive exactly on $S_1$ (entries
$\ge\kappa-2\iota\ge\kappa/2$, by Step~3) and strictly negative off $S_1$
(entries $\le-\delta/2+2\iota\le-\delta/4$, by Step~2); hence on
$\mathcal{U}_0\cap\Ms$ the defining equation for $\gamma_1$ reads
\begin{equation}
\label{eq:gamma-face}
\sum_{i\in S_1}\bigl(\rt_{1,i}+m_{1,i}-\gamma_1\bigr)^2
\;=\;\sum_{i\in S_1}\rt_{1,i}^{\,2}\,,
\end{equation}
a quadratic equation in $\gamma_1$ whose left-hand side has
$\gamma$-derivative $-2\sum_{i\in S_1}r_{1,i}\le-2\kappa<0$ at the solution.
By the implicit function theorem, $\gamma_1$ is a real-analytic function of
$(\rt_1|_{S_1},m_1|_{S_1})$, and $m_1|_{S_1}=\Atil\,\yt|_{S_2}$ is
real-analytic in $\rt_2|_{S_2}$.  All remaining operations of the round
(normalizations with denominators bounded below, matrix products, the update
with the clipping pattern frozen by Steps~2--3) are real-analytic on
$\mathcal{U}_0\cap\Ms$.  Hence $\Phi_S=\Phi|_{\Ms}$ is a composition of
real-analytic maps on $\mathcal{U}_0\cap\Ms$, proving~(ii).
\end{proof}

\subsubsection{The Jacobian on the frozen face: proofs of
Lemma~\ref{lem:jacobian} and Lemma~\ref{lem:centering}}
\label{app:local:jacobian}

Throughout this part Assumption~\ref{asm:nondeg-strict} is in force at
the designated equilibrium $(x^*,y^*)$; the computation uses only the
kernel condition, not uniqueness (which it implies,
Lemma~\ref{lem:kernel}).  We work in face coordinates
$z=(\rt_1,\rt_2)\in\R^{S_1}\times\R^{S_2}$ on
$\mathcal{U}_0\cap\Ms$, where $\mathcal{U}_0$ is the neighborhood of a fixed
point $z_f=(a\xi,b\psi)$, $a,b>0$, provided by Lemma~\ref{lem:absorb}.  With
the clipping pattern frozen, one round of the face dynamics reads
\begin{equation}
\label{eq:face-round}
\rt_1' \;=\; r_1+g_1 \;=\; \rt_1+\Atil\,y
-\bigl(\gamma_1+\inner{\mathrm{res}_1}{x}\bigr)\ones,
\qquad
\rt_2' \;=\; \rt_2-\Atil^\top x
-\bigl(\gamma_2+\inner{\mathrm{res}_2}{y}\bigr)\ones,
\end{equation}
where $r_1=\rt_1+m_1-\gamma_1\ones>0$, $x=r_1/\norm{r_1}_1$,
$m_1=\Atil\yt$, $\mathrm{res}_1=\Atil(y-\yt)$ (so that
$m_1+\mathrm{res}_1=\Atil y$), and symmetrically
$r_2=\rt_2+m_2-\gamma_2\ones>0$, $y=r_2/\norm{r_2}_1$, $m_2=-\Atil^\top\xt$,
$\mathrm{res}_2=-\Atil^\top(x-\xt)$.  At the fixed point $z_f$:
$\gamma_1=v^*$, $\gamma_2=-v^*$, $r_1=\rt_1=a\xi$, $r_2=\rt_2=b\psi$,
$x=\xt=\xi$, $y=\yt=\psi$, and $\mathrm{res}_1=\mathrm{res}_2=0$
(Proposition~\ref{prop:fixedpoint} restricted to the face).  Recall also the
differential of the normalization $u\mapsto u/\inner{\ones}{u}$ (for $u>0$):
at $u=a\xi$ it is $du\mapsto P_\xi\,du/a$, and at $u=b\psi$ it is
$du\mapsto P_\psi\,du/b$.

\begin{lemma}[Vanishing of the multiplier differentials]
\label{lem:dgamma}
At any fixed point $z_f=(a\xi,b\psi)$, $a,b>0$, the differentials of
$\gamma_1$ and $\gamma_2$ along $\Ms$ satisfy
$d\gamma_1=\inner{\xi}{dm_1}$ and $d\gamma_2=\inner{\psi}{dm_2}$, and both
vanish: $d\gamma_1=d\gamma_2=0$.
\end{lemma}

\begin{proof}
Differentiating the frozen-pattern multiplier equation \eqref{eq:gamma-face}
gives
\[
\inner{r_1}{\,d\rt_1+dm_1-d\gamma_1\ones\,}
\;=\;\inner{\rt_1}{d\rt_1}.
\]
At $z_f$ we have $r_1=\rt_1=a\xi$, so the terms in $d\rt_1$ cancel
\emph{exactly} (this is the cancellation between the gradient
$\nabla_v\gamma=x$ of the shift in its first argument and the derivative of
the norm constraint), leaving
$a\inner{\xi}{dm_1}=d\gamma_1\cdot a\inner{\xi}{\ones}=a\,d\gamma_1$, i.e.\
$d\gamma_1=\inner{\xi}{dm_1}$.  Now $dm_1=\Atil\,d\yt$ with
$d\yt=P_\psi\,d\rt_2/b$, hence
\[
d\gamma_1
=\tfrac1b\,\xi^\top\Atil P_\psi\,d\rt_2
=\tfrac1b\,v^*\,\ones^\top P_\psi\,d\rt_2
=0,
\]
using the equilibrium indifference identity $\xi^\top\Atil=v^*\ones^\top$
(Lemma~\ref{lem:centering}) and
$\ones^\top P_\psi=\ones^\top-(\ones^\top\psi)\ones^\top=0$.  Symmetrically,
$d\gamma_2=\inner{\psi}{dm_2}=-\tfrac1a\,\psi^\top\Atil^\top P_\xi\,d\rt_1
=-\tfrac1a\,v^*\,\ones^\top P_\xi\,d\rt_1=0$.
\end{proof}

\begin{lemma}[Vanishing of the centering differentials]
\label{lem:dcenter}
At any fixed point $z_f=(a\xi,b\psi)$, $a,b>0$,
$d\inner{\mathrm{res}_1}{x}=0$ and $d\inner{\mathrm{res}_2}{y}=0$.
\end{lemma}

\begin{proof}
On the face, $\inner{\mathrm{res}_1}{x}=x^\top\Atil\,(y-\yt)$.  At $z_f$ we
have $y-\yt=0$ and $x=\xi$, so by the product rule
\[
d\inner{\mathrm{res}_1}{x}
=(dx)^\top\Atil\,(y-\yt)\big|_{z_f}
+\xi^\top\Atil\,(dy-d\yt)
=0+v^*\,\ones^\top(dy-d\yt)
=0,
\]
since $y$ and $\yt$ are simplex-valued analytic maps, whence
$\ones^\top dy=\ones^\top d\yt=0$.  For player~2,
$\inner{\mathrm{res}_2}{y}=-y^\top\Atil^\top(x-\xt)$, and the same argument
with $\Atil\psi=v^*\ones$ gives
$d\inner{\mathrm{res}_2}{y}=-\psi^\top\Atil^\top(dx-d\xt)
=-v^*\ones^\top(dx-d\xt)=0$.
\end{proof}

\begin{lemma}[Closed-form Jacobian; $J=I+L+L^2$]
\label{lem:jacobian}
Let Assumption~\ref{asm:nondeg-strict} hold at $(x^*,y^*)$.  At any fixed
point $z_f=(a\xi,b\psi)$ with $a,b>0$,
\begin{equation}
\label{eq:Ldef}
J \;:=\; D\Phi_S(z_f) \;=\; I + L + L^2,
\qquad
L \;=\; \begin{pmatrix} 0 & U/b \\ -V/a & 0 \end{pmatrix},
\qquad
U:=\Atil P_\psi,\quad V:=\Atil^\top P_\xi,
\end{equation}
and consequently $L^2=-\diag(UV,\,VU)/(ab)$.  The derivation rests on two
exact cancellations at $z_f$: the differentials of the shift multipliers
vanish, $d\gamma_1=d\gamma_2=0$, and so do the differentials of the
centering terms.
\end{lemma}

\begin{proof}[Proof of Lemma~\ref{lem:jacobian}]
Differentiate \eqref{eq:face-round} at $z_f$.  By Lemma~\ref{lem:dgamma} and
Lemma~\ref{lem:dcenter} the differentials of the two scalar terms vanish, so
\[
d\rt_1'=d\rt_1+\Atil\,dy,
\qquad
d\rt_2'=d\rt_2-\Atil^\top dx .
\]
For $dy$: at $z_f$, $r_2=b\psi$, so $dy=P_\psi\,dr_2/b$ with
$dr_2=d\rt_2+dm_2-d\gamma_2\ones=d\rt_2-\tfrac1a\Atil^\top P_\xi\,d\rt_1$
(using $dm_2=-\Atil^\top d\xt$, $d\xt=P_\xi\,d\rt_1/a$, and
$d\gamma_2=0$).  Hence, with $U=\Atil P_\psi$ and $V=\Atil^\top P_\xi$,
\[
d\rt_1'
=d\rt_1+\tfrac1b\,U\,d\rt_2-\tfrac1{ab}\,UV\,d\rt_1 .
\]
For $dx$: $dx=P_\xi\,dr_1/a$ with
$dr_1=d\rt_1+dm_1-d\gamma_1\ones=d\rt_1+\tfrac1b\,\Atil P_\psi\,d\rt_2
=d\rt_1+\tfrac1b\,U\,d\rt_2$, so
\[
d\rt_2'
=d\rt_2-\tfrac1a\,V\,d\rt_1-\tfrac1{ab}\,VU\,d\rt_2 .
\]
These two displays are exactly the block rows of
$(I+L+L^2)\,(d\rt_1,d\rt_2)$ with $L$ as in \eqref{eq:Ldef}, since
$L^2=-\diag(UV,VU)/(ab)$.
\end{proof}

\begin{lemma}[Value-centering identities; $U=\Bmat$, $V=\Bmat^\top$]
\label{lem:centering}
At the equilibrium, the indifference conditions $\Atil\psi=v^*\ones$ and
$\Atil^\top\xi=v^*\ones$ hold, and therefore
\begin{equation}
\label{eq:centering}
U \;=\; \Atil P_\psi \;=\; \Atil - v^*\ones\ones^\top \;=\; \Bmat,
\qquad
V \;=\; \Atil^\top P_\xi \;=\; \Bmat^\top .
\end{equation}
Moreover $\Bmat\psi=0$ and $\Bmat^\top\xi=0$; under
Assumption~\ref{asm:nondeg-strict} these span the kernels.
\end{lemma}

\begin{proof}[Proof of Lemma~\ref{lem:centering}]
Since every $i\in S_1=\supp(x^*)$ is a best response to $y^*$,
$(Ay^*)_i=v^*$ for $i\in S_1$, i.e.\ $\Atil\psi=v^*\ones$; symmetrically
$\Atil^\top\xi=v^*\ones$.  Hence
\[
U=\Atil P_\psi=\Atil-(\Atil\psi)\ones^\top=\Atil-v^*\ones\ones^\top=\Bmat,
\qquad
V=\Atil^\top P_\xi=\Atil^\top-(\Atil^\top\xi)\ones^\top=\Bmat^\top .
\]
Moreover $\Bmat\psi=\Atil\psi-v^*(\ones^\top\psi)\ones=v^*\ones-v^*\ones=0$
and, symmetrically, $\Bmat^\top\xi=0$.  That these vectors \emph{span} the
kernels is exactly the non-degeneracy requirement of
Assumption~\ref{asm:nondeg-strict}; under
Assumption~\ref{asm:unique-strict} it holds automatically by
Lemma~\ref{lem:kernel} below.  Finally,
$UV/(ab)=\Bmat\Bmat^\top/(ab)$ is a Gram matrix up to the positive scalar
$1/(ab)$, hence symmetric positive semidefinite with real nonnegative
spectrum; the identity $\Bmat^\top\xi=0$ places the equilibrium direction
$\xi$ in its kernel, and likewise $\Bmat\psi=0$ places $\psi$ in the kernel
of $VU/(ab)=\Bmat^\top\Bmat/(ab)$.
\end{proof}

\subsubsection{Spectral geometry: Step A of the proof of
Theorem~\ref{thm:rate}}
\label{sec:local:stepA}

\paragraph{Standing notation.}
Assumption~\ref{asm:unique-strict} is in force.  By Lemma~\ref{lem:kernel}
below, Assumption~\ref{asm:nondeg-strict} then holds at $(x^*,y^*)$, so
Lemma~\ref{lem:absorb}, Lemma~\ref{lem:jacobian} and
Lemma~\ref{lem:centering} apply, and $|S_1|=|S_2|=:s$; the
face has dimension $n_S:=2s$.  Let
$\Bmat=\sum_{k=1}^{s-1}\sigma_k u_kv_k^\top$ be a reduced singular value
decomposition with $\sigma_1\ge\cdots\ge\sigma_{s-1}>0$; write
$\sigmax=\sigma_1$ and $\sigmin^{+}=\sigma_{s-1}$.  Set $e_1=(\xi,0)$,
$e_2=(0,\psi)$, $E=\mathrm{span}\{e_1,e_2\}$,
$E_+=\{ae_1+be_2:a,b>0\}$, and $z^*=c_1e_1+c_2e_2$ with $c_1,c_2>0$ the
reference fixed point; $\eta=\eta(a,b)=1/\sqrt{ab}$,
$\etainf=\eta(c_1,c_2)$, and
$\rho(a,b)=\max_k\sqrt{1-\eta^2\sigma_k^2(1-\eta^2\sigma_k^2)}$,
$\rho^*=\rho(c_1,c_2)$.  Closed Euclidean balls in face coordinates are
denoted $\mathbb{B}(z,r)$.  By Proposition~\ref{prop:fixedpoint} and
uniqueness of the equilibrium, $\Fix(\Phi_S)$ coincides with $E_+$ in a
neighborhood of any of its points, and $\Phi_S$ fixes $E_+$ pointwise.

The proofs use uniqueness only through a kernel condition on the support matrix, stated separately:
\begin{assumption}[Strictly complementary, non-degenerate equilibrium ray]
\label{asm:nondeg-strict}
The designated Nash equilibrium $(x^*,y^*)$ is \emph{strictly complementary} ($\delta>0$) and \emph{non-degenerate}: $\ker\Bmat=\mathrm{span}\{y^*|_{S_2}\}$ and $\ker\Bmat^\top=\mathrm{span}\{x^*|_{S_1}\}$, equivalently $\mathrm{rank}\,\Bmat=|S_1|-1=|S_2|-1$; in particular $|S_1|=|S_2|$.
\end{assumption}
\begin{remark}[One hypothesis, two forms]
\label{rem:no-a2}
By Lemma~\ref{lem:kernel} the two assumptions are equivalent, so a game with several equilibria has no strictly complementary non-degenerate one.
\end{remark}

\begin{lemma}[Uniqueness $\Leftrightarrow$ non-degeneracy:
Assumption~\ref{asm:unique-strict} $\Leftrightarrow$
Assumption~\ref{asm:nondeg-strict}]
\label{lem:kernel}
Let $(x^*,y^*)$ be a strictly complementary equilibrium.  Then it is
non-degenerate, $\ker\Bmat=\mathrm{span}\{\psi\}$ and
$\ker\Bmat^\top=\mathrm{span}\{\xi\}$, if and only if it is that game's
unique equilibrium.  In particular, under either assumption
$\mathrm{rank}\,\Bmat=|S_1|-1=|S_2|-1$, so $|S_1|=|S_2|$; and in a game
with several equilibria no equilibrium is both strictly complementary and
non-degenerate.
\end{lemma}

\begin{proof}
\emph{Uniqueness $\Rightarrow$ non-degeneracy.}\ \
Let $d\in\ker\Bmat$, i.e.\ $\Atil d=v^*(\ones^\top d)\ones$.

\emph{Case 1: $\ones^\top d=0$ and $d\notin\mathrm{span}\{\psi\}$ (so
$d\neq0$).}\ \ Then $\Atil d=0$.  For $t\in\R$, let $y_t$ denote the vector
$\psi+td$ on $S_2$, extended by zero off $S_2$.  Since $\psi>0$ on $S_2$ and
$\ones^\top(\psi+td)=1$, for $|t|$ small enough $y_t$ is a strategy with
$\supp(y_t)\subseteq S_2$.  We claim $(x^*,y_t)$ is a Nash equilibrium:
(i)~on $S_1$, $\Atil(\psi+td)=v^*\ones$, while for $i\notin S_1$,
$(Ay_t)_i=(Ay^*)_i+t\,(A[i,S_2]\,d)\le v^*-\delta+O(|t|)<v^*$ for $|t|$
small; hence $\max_i(Ay_t)_i=v^*$, attained on $S_1\supseteq\supp(x^*)$, so
$x^*$ is a best response to $y_t$.  (ii)~$(x^{*\top}A)_j=v^*$ for all
$j\in S_2$ (indifference) and $\ge v^*+\delta$ off $S_2$, so
$\min_j(x^{*\top}A)_j=v^*$ and every strategy supported in $S_2$, in
particular $y_t$, is a best response to $x^*$ for the minimizing player.
For $t\neq0$, $y_t\neq y^*$, contradicting uniqueness of the equilibrium.

\emph{Case 2: $\ones^\top d\neq0$.}\ \ Normalize so that $\ones^\top d=1$.
Then $d':=d-\psi$ satisfies $\Atil d'=v^*\ones-v^*\ones=0$ and
$\ones^\top d'=0$, so $d'\in\ker\Bmat$; moreover $d'\notin
\mathrm{span}\{\psi\}$ unless $d'=0$ (any nonzero multiple of $\psi$ has
nonzero coordinate sum).  If $d'\neq0$, Case~1 applies to $d'$ and gives a
contradiction; hence $d'=0$ and $d=\psi$.

Therefore $\ker\Bmat=\mathrm{span}\{\psi\}$; the statement for
$\Bmat^\top$ and $\xi$ is symmetric.  The rank--nullity theorem gives
$\mathrm{rank}\,\Bmat=|S_2|-1$ and $\mathrm{rank}\,\Bmat^\top=|S_1|-1$;
equality of the two ranks forces $|S_1|=|S_2|$.

\emph{Non-degeneracy $\Rightarrow$ uniqueness.}\ \
Let $(x^\dagger,y^\dagger)$ be any equilibrium.  Since $x^*$ guarantees
$v^*$ ($x^{*\top}Ay\ge v^*$ for all $y$) and $y^\dagger$ guarantees $v^*$
($x^\top Ay^\dagger\le v^*$ for all $x$), $x^{*\top}Ay^\dagger=v^*$, so
$(x^*,y^\dagger)$ is an equilibrium.  First, $y^\dagger$ is a best response
to $x^*$: \rev{write $s_j:=(x^{*\top}A)_j-v^*$; then $s_j\ge0$ for every $j$ and, by \eqref{eq:margin}, $s_j>0$ for $j\notin S_2$, while $\sum_j y^\dagger_j s_j=0$ with every product nonnegative, so $y^\dagger_j s_j=0$ for all $j$ and hence $y^\dagger_j=0$ for $j\notin S_2$, i.e.\ $\supp(y^\dagger)\subseteq S_2$.}  Second, $x^*$ is a best response to
$y^\dagger$ with $\max_i(Ay^\dagger)_i=v^*$, so every
$i\in\supp(x^*)=S_1$ attains the maximum: $\Atil\,y^\dagger|_{S_2}=v^*\ones$.
Hence $d:=y^\dagger|_{S_2}-\psi$ satisfies $\Atil d=0$ and $\ones^\top d=0$,
so $\Bmat d=\Atil d-v^*(\ones^\top d)\ones=0$ and $d\in\ker\Bmat=
\mathrm{span}\{\psi\}$; as $\ones^\top\psi=1\ne0=\ones^\top d$, $d=0$ and
$y^\dagger=y^*$.  The symmetric argument gives $x^\dagger=x^*$.

Both directions are consistent with the
classical non-degeneracy theory of matrix games
\citep{bohnenblust1950solutions}; the proofs above are self-contained.
\end{proof}

\begin{lemma}[Orthogonal splitting]
\label{lem:splitting}
Define
\[
W_k:=\mathrm{span}\bigl\{(u_k,0),\,(0,v_k)\bigr\}\quad(k=1,\dots,s-1),
\qquad
W:=\bigoplus_{k=1}^{s-1}W_k .
\]
Then $\R^{n_S}=E\oplus W$ with $E\perp W$ (Euclidean orthogonality); the
corresponding orthogonal projections $P_E,P_W$ satisfy
$\norm{P_E}=\norm{P_W}=1$.  For $w=(w_1,w_2)\in W$ one has
$w_1\in\mathrm{range}(\Bmat)=\mathrm{span}\{u_k\}$,
$w_2\in\mathrm{range}(\Bmat^\top)=\mathrm{span}\{v_k\}$, and
$\norm{w}^2=\norm{w_1}^2+\norm{w_2}^2$.
\end{lemma}

\begin{proof}
By Lemma~\ref{lem:centering}, $\Bmat^\top\xi=0$, i.e.\
$\xi\perp\mathrm{range}(\Bmat)=\mathrm{span}\{u_k\}$; hence
$e_1=(\xi,0)\perp(u_k,0)$ for all $k$, and $e_1\perp(0,v_k)$ trivially.  The
statement for $e_2$ is symmetric (using $\Bmat\psi=0$).  The $u_k$ are
orthonormal, as are the $v_k$, so the $W_k$ are mutually orthogonal and give
orthonormal coordinates on $W$.  Dimensions:
$\dim E+\dim W=2+2(s-1)=2s=n_S$ (using Lemma~\ref{lem:kernel}), so the
orthogonal sum exhausts the face.
\end{proof}

\begin{lemma}[Uniform block diagonalization and closed-form modes]
\label{lem:blockdiag}
For every fixed point $z_f=(a\xi,b\psi)\in E_+$, the Jacobian
$J(z_f)=I+L(a,b)+L(a,b)^2$ of Lemma~\ref{lem:jacobian} (with $U=\Bmat$,
$V=\Bmat^\top$ by Lemma~\ref{lem:centering}) satisfies:
\begin{enumerate}[label=(\arabic*)]
\item $J(z_f)|_E=I_E$; in fact $J(z_f)e_1=e_1$ and $J(z_f)e_2=e_2$ exactly.
\item Each $W_k$ is $J(z_f)$-invariant, and in the basis
$\{(u_k,0),(0,v_k)\}$,
\[
J(z_f)|_{W_k}=J_k(a,b):=I_2+L_k+L_k^2,
\qquad
L_k=\begin{pmatrix}0&\sigma_k/b\\ -\sigma_k/a&0\end{pmatrix}.
\]
\item Consequently $P_EJ(z_f)P_W=0$ and $P_WJ(z_f)P_E=0$ for \emph{all}
$a,b>0$; write $\mathcal{A}(a,b):=J(z_f)|_W=\bigoplus_kJ_k(a,b)$.
\item $\spec\bigl(J_k(a,b)\bigr)=\{\mu_k^{\pm}\}$ with
$\mu_k^{\pm}=1-\eta^2\sigma_k^2\pm i\,\eta\sigma_k$ and
$|\mu_k^{\pm}|^2=1-\eta^2\sigma_k^2(1-\eta^2\sigma_k^2)$, where
$\eta=1/\sqrt{ab}$.  Since $\eta\sigma_k\neq0$, the pair is non-real and
conjugate, so each $J_k$ is diagonalizable over $\mathbb{C}$ and
$\mathcal{A}(a,b)$ is semisimple.  Moreover
$\rho(a,b)=\max_k|\mu_k^{\pm}(a,b)|$.
\item The map $(a,b)\mapsto\mathcal{A}(a,b)$ is real-analytic on
$\{a,b>0\}$, and on the rectangle
$R_\square:=[c_1/2,\,3c_1/2]\times[c_2/2,\,3c_2/2]$ it is Lipschitz:
\[
\begin{aligned}
\norm{\mathcal{A}(a,b)-\mathcal{A}(a',b')}
&\le C_J\bigl(|a-a'|+|b-b'|\bigr),\\
C_J:={}&\sigmax\Bigl(\tfrac{4}{c_1^2}+\tfrac{4}{c_2^2}\Bigr)
\Bigl(1+2\sigmax\bigl(\tfrac{2}{c_1}+\tfrac{2}{c_2}\bigr)\Bigr).
\end{aligned}
\]
\end{enumerate}
\end{lemma}

\begin{proof}
(1)--(2): direct verification.  Using $\Bmat v_k=\sigma_ku_k$ and
$\Bmat^\top u_k=\sigma_kv_k$,
\begin{gather*}
L(a,b)\,(u_k,0)=(0,\,-\Bmat^\top u_k/a)=(0,\,-\sigma_kv_k/a)\in W_k,
\\
L(a,b)\,(0,v_k)=(\Bmat v_k/b,\,0)=(\sigma_ku_k/b,\,0)\in W_k,
\end{gather*}
so $W_k$ is $L$-invariant with matrix $L_k$ in the given basis, hence
$J$-invariant with matrix $I_2+L_k+L_k^2$.  Similarly
$L(a,b)\,e_1=(0,\,-\Bmat^\top\xi/a)=0$ (Lemma~\ref{lem:centering}), so
$Je_1=e_1$; and $Le_2=(\Bmat\psi/b,\,0)=0$, so $Je_2=e_2$.

(3): immediate from (1)--(2) and the orthogonality of the splitting
(Lemma~\ref{lem:splitting}): $J$ maps $E$ into $E$ and $W$ into $W$.

(4): the eigenvalues of $L_k$ solve $\lambda^2=-\sigma_k^2/(ab)$, i.e.\
$\lambda=\pm i\eta\sigma_k$; hence the eigenvalues of $J_k$ are
$\mu=1+\lambda+\lambda^2=1-\eta^2\sigma_k^2\pm i\eta\sigma_k$, and
\[
|\mu_k^{\pm}|^2
=(1-\eta^2\sigma_k^2)^2+\eta^2\sigma_k^2
=1-\eta^2\sigma_k^2\,(1-\eta^2\sigma_k^2).
\]
The pair is non-real and distinct because $\eta\sigma_k>0$, so $J_k$ is
diagonalizable over $\mathbb{C}$; the direct sum $\mathcal{A}(a,b)$ is
therefore semisimple, and its spectral radius is $\max_k|\mu_k^\pm|$, which
equals $\rho(a,b)$ by definition.

(5): $L(a,b)$ depends affinely on $(1/a,1/b)$, so it is real-analytic on
$\{a,b>0\}$; hence so is $\mathcal{A}$.  For the Lipschitz bound, note that
for the block-antidiagonal $L$, $\norm{L(a,b)}=
\max\{\sigmax/b,\ \sigmax/a\}\le\sigmax(2/c_1+2/c_2)$ on $R_\square$, and
\[
\norm{L(a,b)-L(a',b')}
\le\sigmax\max\Bigl\{\bigl|\tfrac1a-\tfrac1{a'}\bigr|,\
\bigl|\tfrac1b-\tfrac1{b'}\bigr|\Bigr\}
\le\sigmax\Bigl(\tfrac{4}{c_1^2}+\tfrac{4}{c_2^2}\Bigr)
\bigl(|a-a'|+|b-b'|\bigr),
\]
using $a,a'\ge c_1/2$ and $b,b'\ge c_2/2$.  Finally
$\norm{L^2-L'^2}\le(\norm{L}+\norm{L'})\norm{L-L'}$, so
$\norm{\mathcal{A}-\mathcal{A}'}\le
\bigl(1+2\sigmax(\tfrac2{c_1}+\tfrac2{c_2})\bigr)\norm{L-L'}$, which gives
$C_J$.
\end{proof}

\begin{remark}
Items (1)--(3) hold \emph{simultaneously at every fixed point}: the mode
planes $W_k$ are fixed once and for all by the singular vectors of $\Bmat$,
and only the scalars $\sigma_k/a$, $\sigma_k/b$ vary along the cone.  By
items (1) and (4), the unit eigenvalue of $J(z_f)$ is exactly double and
semisimple with eigenspace $E$ (a unit eigenvalue inside some $W_k$ would
require $\lambda+\lambda^2=0$, i.e.\ $\lambda\in\{0,-1\}$, impossible for
$\lambda=\pm i\eta\sigma_k$).  Thus the usual non-degeneracy hypothesis of
local analyses is a theorem here, as announced in Remark~\ref{rem:no-a2}.
\end{remark}

\begin{lemma}[Adapted norm on $W$]
\label{lem:adapted}
Let $\norm{w}_{*}^2:=\norm{w_1}^2/c_1+\norm{w_2}^2/c_2$ for $w=(w_1,w_2)\in W$,
and set $m_*:=1/\sqrt{\max(c_1,c_2)}$, $M_*:=1/\sqrt{\min(c_1,c_2)}$,
$\kappa_*:=M_*/m_*=\sqrt{\max(c_1,c_2)/\min(c_1,c_2)}$.  Then:
\begin{enumerate}[label=(\arabic*)]
\item $\norm{\mathcal{A}(c_1,c_2)}_{*}=\rho^*$ (the operator norm equals the
spectral radius exactly);
\item $m_*\norm{w}\le\norm{w}_{*}\le M_*\norm{w}$ for all $w\in W$;
\item for all $(a,b)\in R_\square$,
$\norm{\mathcal{A}(a,b)}_{*}\le\rho^*
+\kappa_*C_J\bigl(|a-c_1|+|b-c_2|\bigr)$.
\end{enumerate}
\end{lemma}

\begin{proof}
Let $S:=\mathrm{diag}(\sqrt{c_1},\sqrt{c_2})$.  In the basis
$\{(u_k,0),(0,v_k)\}$ of $W_k$ (Lemma~\ref{lem:blockdiag}(2)),
\[
S^{-1}L_k(c_1,c_2)S
=\begin{pmatrix}0&\sigma_k\sqrt{c_2}/(c_2\sqrt{c_1})\\
-\sigma_k\sqrt{c_1}/(c_1\sqrt{c_2})&0\end{pmatrix}
=\etainf\sigma_k\begin{pmatrix}0&1\\-1&0\end{pmatrix},
\]
a skew-symmetric matrix, so
$S^{-1}J_k(c_1,c_2)S=(1-\etainf^2\sigma_k^2)I_2+\etainf\sigma_k
\begin{pmatrix}0&1\\-1&0\end{pmatrix}=|\mu_k|\,R(\varphi_k)$
with $R(\varphi)=\begin{pmatrix}\cos\varphi&-\sin\varphi\\
\sin\varphi&\cos\varphi\end{pmatrix}$ and $\varphi_k=\arg\mu_k^{-}$; the same
$S$ serves every block.  For $w\in W$ with components $w_k\in W_k$ (the
$W_k$ are mutually orthogonal), $\sum_k\norm{S^{-1}w_k}^2
=\norm{w_1}^2/c_1+\norm{w_2}^2/c_2=\norm{w}_*^2$.  Then
\[
\norm{\mathcal{A}(c_1,c_2)\,w}_{*}^2
=\sum_k\bigl\lVert\,|\mu_k|R(\varphi_k)\,S^{-1}w_k\bigr\rVert^2
=\sum_k|\mu_k|^2\norm{S^{-1}w_k}^2
\le(\rho^*)^2\norm{w}_{*}^2 ,
\]
with equality when $w$ is supported on a maximizing block; this proves~(1).
(2) is immediate from $\norm{S^{-1}}=M_*$ and $\norm{S}=1/m_*$.
For~(3), for any operator $T$ on
$W$ one has $\norm{T}_{*}\le\kappa_*\norm{T}$ (combine the two bounds
of~(2)), hence
\[
\norm{\mathcal{A}(a,b)}_{*}
\le\norm{\mathcal{A}(c_1,c_2)}_{*}
+\kappa_*\norm{\mathcal{A}(a,b)-\mathcal{A}(c_1,c_2)}
\le\rho^*+\kappa_*C_J\bigl(|a-c_1|+|b-c_2|\bigr),
\]
by Lemma~\ref{lem:blockdiag}(5).
\end{proof}

\begin{corollary}[Extreme singular values determine the rate]
\label{cor:extreme}
For $\sigma\in\Sigma^+(\Bmat)$ let $q_\sigma:=\etainf^2\sigma^2$, so the corresponding mode has modulus $|\mu_\sigma|=\sqrt{1-q_\sigma(1-q_\sigma)}$.  Then:
\begin{enumerate}[label=(\roman*),itemsep=1pt,topsep=2pt]
\item $\rho^*$ is attained at an \emph{extreme} nonzero singular value: $\rho^*=\max\{|\mu_{\sigmax(\Bmat)}|,\ |\mu_{\sigmin^{+}(\Bmat)}|\}$, where $\sigmin^{+}(\Bmat)$ denotes the smallest nonzero singular value.
\item Per mode, contraction is best exactly at $q_\sigma=\tfrac12$, where $|\mu_\sigma|=\sqrt{3}/2$, and on the contracting range $q_\sigma\in(0,1)$ it degrades monotonically as $q_\sigma\to0$ or $q_\sigma\to1$.
\item Under Assumption~\ref{asm:stepsize} (every mode contracts), $\rho^*\to1$ as $\etainf\,\sigmax(\Bmat)\uparrow1$ or as $\etainf\,\sigmin^{+}(\Bmat)\downarrow0$; without the assumption an unstable mode can keep $\rho^*$ above $1$ while $\etainf\sigmin^{+}(\Bmat)\downarrow0$.
\end{enumerate}
\end{corollary}

\subsection{Proof of Theorem~\ref{thm:rate}: Steps B and C}
\label{app:local:rate}

The standing notation of this proof, and Step~A (the spectral geometry of the
Jacobian), are in Section~\ref{sec:local:stepA}; the two steps below use that
notation without change.  Throughout, $m=|S_1|=|S_2|\ge2$, so that
$W\neq\{0\}$ and $\Sigma^+(\Bmat)\neq\emptyset$; in the pure-saddle case
$m=1$ one has $\Ms=E$, every point of $\Ms\cap\mathcal Z$ is a fixed point
(Proposition~\ref{prop:fixedpoint}), and Lemma~\ref{lem:absorb}(i) already
gives \rev{$\Phi(\mathcal U_0\cap\mathcal Z)\subseteq\Fix(\Phi)$}, which is the finite-step
statement recorded before Theorem~\ref{thm:rate}.

\subsubsection*{Step B: uniform Taylor expansion and bootstrap}

Let $\mathcal{U}$ denote the domain of analyticity of $\Phi_S$ around $z^*$
provided by Lemma~\ref{lem:absorb} (in face coordinates).  Fix $r_0>0$ such
that $\mathbb{B}(z^*,3r_0)\subset\mathcal{U}$, the \emph{ambient} closed
ball of radius $3r_0$ around $z^*$ is contained in the absorption
neighborhood $\mathcal{U}_0$ of Lemma~\ref{lem:absorb}, and
\begin{equation}
\label{eq:r0pos}
r_0\;\le\;\tfrac16\min\bigl\{c_1\norm{\xi}_2,\ c_2\norm{\psi}_2\bigr\}.
\end{equation}
Since $e_1\perp e_2$, a point $z_f=ae_1+be_2\in E$ satisfies
$|a-c_1|\le\norm{z_f-z^*}/\norm{\xi}_2$ and
$|b-c_2|\le\norm{z_f-z^*}/\norm{\psi}_2$; define accordingly
\begin{equation}
\label{eq:CE}
C_E:=\frac{1}{\norm{\xi}_2}+\frac{1}{\norm{\psi}_2},
\end{equation}
so that $|a-c_1|+|b-c_2|\le C_E\norm{z_f-z^*}$.  In particular, by
\eqref{eq:r0pos}, every $z_f\in E\cap\mathbb{B}(z^*,3r_0)$ has coordinates
$(a,b)\in[c_1/2,\,3c_1/2]\times[c_2/2,\,3c_2/2]\subset R_\square$; such
points lie in the open cone $E_+$ and are fixed points of $\Phi_S$.  Finally
set
\begin{equation}
\label{eq:Kdef}
K:=\tfrac12\sup_{z\in\mathbb{B}(z^*,3r_0)}\norm{D^2\Phi_S(z)}\;<\;\infty,
\end{equation}
finite because $\Phi_S$ is analytic on a neighborhood of the compact ball.

\begin{lemma}[Uniform second-order Taylor expansion]
\label{lem:taylor}
For every $z_f\in E_+\cap\mathbb{B}(z^*,r_0)$ and every $w\in W$ with
$\norm{w}\le r_0$,
\[
\Phi_S(z_f+w)\;=\;z_f+J(z_f)\,w+R(z_f,w),
\qquad
\norm{R(z_f,w)}\le K\norm{w}^2 .
\]
\end{lemma}

\begin{proof}
The segment $\{z_f+sw\}_{s\in[0,1]}$ lies in $\mathbb{B}(z^*,2r_0)$, on which
$\Phi_S$ is analytic.  Since $z_f\in E\cap\mathbb{B}(z^*,3r_0)\subset E_+$ and $\Phi_S$ fixes
$E_+$ pointwise (Proposition~\ref{prop:fixedpoint}), $\Phi_S(z_f)=z_f$; and $D\Phi_S(z_f)=J(z_f)$ by
Lemma~\ref{lem:jacobian}.  Taylor's theorem with integral remainder and
\eqref{eq:Kdef} give the claim.
\end{proof}

Fix $\eps\in(0,1-\rho^*)$ and set $\theta:=\rho^*+\eps<1$.  Define
\begin{equation}
\label{eq:Mconst}
r_1:=\min\Bigl\{r_0,\ \frac{\eps}{2\,\kappa_*C_JC_E}\Bigr\},
\qquad
\delta_1:=\min\Bigl\{m_*r_0,\ \frac{\eps\,m_*^2}{2M_*K},\
m_*\sqrt{\frac{r_1(1-\theta^2)}{2K}}\Bigr\}.
\end{equation}

\begin{lemma}[Bootstrap]
\label{lem:bootstrap}
Let $z^0\in\Ms$ satisfy $\norm{P_Ez^0-z^*}\le r_1/2$ and
$\norm{P_Wz^0}_{*}\le\delta_1$.  Define $z^{t+1}:=\Phi_S(z^t)$,
$z_f^t:=P_Ez^t$, and $w^t:=P_Wz^t$.  Then for all $t\ge0$:
\begin{enumerate}[label=(M\arabic*)]
\item $\norm{w^t}_{*}\le\theta^t\,\norm{w^0}_{*}$;
\item $\norm{z_f^{t+1}-z_f^t}\le K\norm{w^t}^2
\le (K/m_*^2)\,\theta^{2t}\norm{w^0}_{*}^2$;
\item $\norm{z_f^t-z^*}\le r_1$; in particular
$z_f^t\in E_+\cap\mathbb{B}(z^*,r_0)$ with coordinates
$(a_t,b_t)\in R_\square$.
\end{enumerate}
\end{lemma}

\begin{proof}
Induction on $t$.  For $t=0$, (M1) is trivial and (M3) holds by hypothesis.
Assume (M1) and (M3) hold for all $s\le t$; we prove (M2) at $t$ and (M1),
(M3) at $t+1$.

\emph{(i) The expansion applies.}\ \ By (M1) and
Lemma~\ref{lem:adapted}(2),
$\norm{w^t}\le\norm{w^t}_{*}/m_*\le\delta_1/m_*\le r_0$, and by (M3),
$z_f^t\in E_+\cap\mathbb{B}(z^*,r_1)\subseteq\mathbb{B}(z^*,r_0)$.
Lemma~\ref{lem:taylor} gives
\[
z^{t+1}=z_f^t+J(z_f^t)\,w^t+R^t,
\qquad
\norm{R^t}\le K\norm{w^t}^2 .
\]

\emph{(ii) $W$-component.}\ \ By Lemma~\ref{lem:blockdiag}(3), applied at the
\emph{current} fixed point $z_f^t$ (block diagonality is exact for every
$(a_t,b_t)$), $J(z_f^t)w^t\in W$ and
$P_WJ(z_f^t)w^t=\mathcal{A}(a_t,b_t)\,w^t$; hence
$w^{t+1}=\mathcal{A}(a_t,b_t)\,w^t+P_WR^t$.  Lemma~\ref{lem:adapted}(3) bounds
$\norm{\mathcal{A}(a_t,b_t)}_{*}$ by $\rho^*+\kappa_*C_J(|a_t-c_1|+|b_t-c_2|)$;
by \eqref{eq:CE} the bracket is at most $C_E\norm{z_f^t-z^*}$, which (M3)
bounds by $C_Er_1$; and $r_1$ was chosen in \eqref{eq:Mconst} so that
$\kappa_*C_JC_Er_1\le\eps/2$.  Hence
\[
\norm{\mathcal{A}(a_t,b_t)}_{*}
\le\rho^*+\kappa_*C_J\,C_E\,r_1
\le\rho^*+\tfrac{\eps}{2}.
\]
Using $\norm{P_W}=1$, $\norm{P_WR^t}_{*}\le M_*\norm{R^t}\le
M_*K\norm{w^t}^2\le (M_*K/m_*^2)\,\norm{w^t}_{*}^2
\le (M_*K\delta_1/m_*^2)\,\norm{w^t}_{*}$, so
\[
\norm{w^{t+1}}_{*}
\le\Bigl(\rho^*+\tfrac{\eps}{2}+\tfrac{M_*K}{m_*^2}\,\delta_1\Bigr)
\norm{w^t}_{*}
\le\theta\,\norm{w^t}_{*},
\]
where the last step uses $\delta_1\le\eps m_*^2/(2M_*K)$ from
\eqref{eq:Mconst}.  This proves (M1) at $t+1$.

\emph{(iii) $E$-component.}\ \ By Lemma~\ref{lem:blockdiag}(3),
$P_EJ(z_f^t)w^t=0$, hence $z_f^{t+1}=z_f^t+P_ER^t$ and
$\norm{z_f^{t+1}-z_f^t}\le K\norm{w^t}^2$; combining with (M1) and
Lemma~\ref{lem:adapted}(2) gives (M2) at $t$.

\emph{(iv) (M3) at $t+1$.}\ \ Summing (M2),
\[
\norm{z_f^{t+1}-z^*}
\le\norm{z_f^0-z^*}+\sum_{s=0}^{t}K\norm{w^s}^2
\le\frac{r_1}{2}+\frac{K\,\delta_1^2}{m_*^2\,(1-\theta^2)}
\le\frac{r_1}{2}+\frac{r_1}{2}=r_1,
\]
where the last inequality uses
$\delta_1\le m_*\sqrt{r_1(1-\theta^2)/(2K)}$ from \eqref{eq:Mconst}.  The
coordinate statement follows from \eqref{eq:r0pos} and \eqref{eq:CE} as
before.
\end{proof}

\subsubsection*{Step C: proof of the theorem}

\begin{proof}[Proof of Theorem~\ref{thm:rate}]
\emph{Parts (i) and (ii).}\ \ By Lemma~\ref{lem:jacobian} and
Lemma~\ref{lem:centering}, $J(z_f)=I+L+L^2$ with the blocks of $L$ given by
$\Bmat/b$ and $-\Bmat^\top/a$.  By Lemma~\ref{lem:splitting} and
Lemma~\ref{lem:blockdiag}, the face splits orthogonally as $E\oplus W$ with
$J(z_f)|_E=I_E$ and $J(z_f)|_{W_k}=J_k(a,b)$, whose eigenvalues are
$1-\eta^2\sigma_k^2\pm i\eta\sigma_k$.  Counting dimensions gives
$2+2(s-1)=n_S$, so the multiset union in \eqref{eq:spec} exhausts the
spectrum.  No block contributes a unit eigenvalue (that would require
$\lambda+\lambda^2=0$ with $\lambda=\pm i\eta\sigma_k$, impossible), so the
eigenvalue $1$ has algebraic multiplicity exactly two; it is semisimple
because $J=I_E\oplus\bigl(\bigoplus_kJ_k\bigr)$ with each $J_k$
diagonalizable over $\mathbb{C}$, and its eigenspace is $E$.  Since
$\Fix(\Phi_S)$ coincides with the open cone $E_+\subset E$ near $z_f$
(Proposition~\ref{prop:fixedpoint} and uniqueness), $E$ is the tangent space
of the fixed-point cone.  For (ii), write $q_k:=\eta^2\sigma_k^2>0$; then
$|\mu_k^\pm|^2=1-q_k(1-q_k)<1$ iff $q_k\in(0,1)$, so
$\rho(a,b)<1$ iff $q_{\max}=\eta^2\sigmax^2<1$ iff
$\eta\,\sigmax(\Bmat)<1$; and if $\eta\sigmax(\Bmat)>1$ then
$|\mu_1^\pm|^2=1+q_{\max}(q_{\max}-1)>1$.

\emph{Part (iii).}\ \ Assume Assumption~\ref{asm:stepsize}, fix
$\eps\in(0,1-\rho^*)$, and let $\theta$, $r_1$, $\delta_1$ be as in
\eqref{eq:Mconst}.

\emph{Step 1 (capture).}\ \ Let
\[
\mathcal{V}:=\bigl\{z\in\Ms:\ \norm{P_Ez-z^*}<r_1/2,\
\norm{P_Wz}_{*}<\delta_1\bigr\},
\]
a relatively open subset of $\Ms$ containing $z^*$.  \begin{revblock}The map $\Phi$ is defined on $\mathcal Z$ and is continuous on
$\mathcal Z$ in the relative topology (Lemma~\ref{lem:phi-cont}), with
$\Phi(z^*)=z^*$ and $\Phi(\mathcal{U}_0\cap\mathcal Z)\subseteq\Ms$ for the
ambient open absorption neighborhood $\mathcal U_0$ of
Lemma~\ref{lem:absorb} (shrunk if necessary); hence
\[
\mathcal O:=\{z\in\mathcal{U}_0\cap\mathcal Z:\ \Phi(z)\in\mathcal V\}
\cap\mathbb{B}(z^*,r_0)^{\circ}
\]
is relatively open in $\mathcal Z$ and contains $z^*$ (here
$\mathbb{B}(z^*,r_0)^{\circ}$ is the open ball in the ambient space).
Choose an ambient open set $G$ with $\mathcal O=G\cap\mathcal Z$ and set
$\mathcal{U}':=G\cap\mathcal{U}_0\cap\mathbb{B}(z^*,r_0)^{\circ}$, an
ambient open neighborhood of $z^*$ whose feasible points are exactly
$\mathcal O$; trajectory states are feasible, so the statements below
concern $\mathcal O$.\end{revblock}  If
$z^{(t_0)}\in\mathcal{U}'$, then $z^{(t_0+1)}\in\mathcal{V}\subset\Ms$ and
$z^{(t)}=\Phi_S^{\,t-t_0-1}\bigl(z^{(t_0+1)}\bigr)$ for $t\ge t_0+1$.  Apply
Lemma~\ref{lem:bootstrap} with $z^0:=z^{(t_0+1)}$ (its hypotheses are the
defining inequalities of $\mathcal{V}$).

\emph{Step 2 (geometric decay of the normal component).}\ \ By (M1) and
Lemma~\ref{lem:adapted}(2), $\norm{w^t}\le(\delta_1/m_*)\,\theta^t$.

\emph{Step 3 (convergence of the tangential component).}\ \ By (M2), for
all $p\ge1$,
\[
\norm{z_f^{t+p}-z_f^t}
\le\sum_{s\ge t}K\norm{w^s}^2
\le\frac{K\,\delta_1^2}{m_*^2(1-\theta^2)}\,\theta^{2t}
\;\xrightarrow[t\to\infty]{}\;0,
\]
so $(z_f^t)$ is Cauchy and converges to some
$z_\infty\in E\cap\mathbb{B}(z^*,r_1)$.  By \eqref{eq:r0pos} and
\eqref{eq:CE}, its coordinates satisfy $a_\infty\ge c_1/2>0$ and
$b_\infty\ge c_2/2>0$, so $z_\infty\in E_+$ and $z_\infty$ is a fixed point
(Proposition~\ref{prop:fixedpoint}).

\emph{Step 4 (rate synthesis).}\ \ For the bootstrap time $t$,
\[
\begin{aligned}
\norm{z^t-z_\infty}
&\le\norm{w^t}+\norm{z_f^t-z_\infty}
\le\frac{\delta_1}{m_*}\,\theta^t
+\frac{K\delta_1^2}{m_*^2(1-\theta^2)}\,\theta^{2t}
\le C'\,\theta^t,\\
C':={}&\frac{\delta_1}{m_*}
\Bigl(1+\frac{K\delta_1}{m_*(1-\theta^2)}\Bigr).
\end{aligned}
\]
Returning to the original time index ($t\mapsto t-t_0-1$), for $t\ge t_0+1$
we get $\norm{z^{(t)}-z_\infty}\le(C'/\theta)\,\theta^{\,t-t_0}$; for
$t=t_0$, $\norm{z^{(t_0)}-z_\infty}\le
\norm{z^{(t_0)}-z^*}+\norm{z^*-z_\infty}\le r_0+r_1\le2r_0$.  Hence
\eqref{eq:linrate} holds with $C_\eps:=\max\{C'/\theta,\ 2r_0\}$.

\emph{Part (iv).}\ \ For a state $z$ near $z^*$, let $(x(z),y(z))$ denote
the main-iterate strategies produced by the round at $z$ (zero-embedded when
$z\in\Ms$), and recall
$\gapf(x,y)=\max_i(Ay)_i-\min_j(x^\top A)_j\ge0$.  On the ambient ball
$\mathbb{B}(z^*,2r_0)$ the map $z\mapsto(x(z),y(z))$ is Lipschitz: the
pre-iterate normalizations have denominators bounded below,
$\norm{\rt_1}_1\ge\norm{\rt_1}_2\ge c_1\norm{\xi}_2-2r_0
\ge\tfrac23\,c_1\norm{\xi}_2$ by \eqref{eq:r0pos} (and symmetrically for
$\rt_2$); the multiplier $\gamma_1$ is $\sqrt2$-Lipschitz in
$(v,\Nt_1)$, each partial derivative having norm at most one: wherever the active set $\{i:v_i>\gamma_1\}$ is locally
constant, implicit differentiation of $\varphi(\gamma_1;v)^2=\Nt_1^2$ gives
$\nabla_v\gamma_1=p_1/\norm{p_1}_1$ and
$\partial\gamma_1/\partial\Nt_1=-\Nt_1/\norm{p_1}_1$, both of norm at most
one (as $\norm{p_1}_1\ge\norm{p_1}_2=\Nt_1$), and $\gamma_1$ is continuous
across active-set changes (Step~1 of the proof of
Lemma~\ref{lem:absorb}); clipping is
$1$-Lipschitz; and the main-iterate normalizations have denominators
$\norm{p_i}_1\ge\norm{p_i}_2=\Nt_i$ bounded below.  Composing with the
Lipschitz functionals $y\mapsto\max_i(Ay)_i$ and $x\mapsto\min_j(x^\top A)_j$
shows that $z\mapsto\gapf(x(z),y(z))$ is $L_{\gapf}$-Lipschitz on
$\mathbb{B}(z^*,2r_0)$, with $L_{\gapf}$ depending only on $\norm{A}$,
$c_1$, $c_2$, and the dimensions.  Every fixed point $z_f\in E_+$ produces
the strategies $(x^*,y^*)$ (Proposition~\ref{prop:fixedpoint}), so
$\gapf(x(z_\infty),y(z_\infty))=0$, and therefore, for all $t\ge t_0$,
\[
\begin{aligned}
\gapf\bigl(x^{(t)},y^{(t)}\bigr)
&=\gapf\bigl(x(z^{(t)}),y(z^{(t)})\bigr)
-\gapf\bigl(x(z_\infty),y(z_\infty)\bigr)\\
&\le L_{\gapf}\,\norm{z^{(t)}-z_\infty}
\le L_{\gapf}\,C_\eps\,(\rho^*+\eps)^{\,t-t_0}.
\end{aligned}
\qedhere
\]
\end{proof}

\paragraph{Dependence of constants.}
For reference: $r_0$ depends on the absorption neighborhood of
Lemma~\ref{lem:absorb} (hence on $\delta$, $\norm{A}$, $c_1$, $c_2$, the
dimensions, and the smallest on-support probabilities $\min_{i\in S_1}x^*_i$,
$\min_{j\in S_2}y^*_j$, which enter Step~3 of its proof); $K$ on
the second derivatives of $\Phi_S$ over $\mathbb{B}(z^*,3r_0)$; $C_J$ and
$C_E$ are explicit in $(\sigmax,c_1,c_2,\norm{\xi}_2,\norm{\psi}_2)$; $m_*$,
$M_*$, $\kappa_*$ are the closed forms of Lemma~\ref{lem:adapted}, functions
of $(c_1,c_2)$ alone;
$r_1$, $\delta_1$, $C_\eps$ depend on all of the above and on $\eps$.

\begin{revblock}
\subsection{Proof of Proposition~\ref{prop:unique-strict}}\label{app:local:unique-strict}
If Assumption~\ref{asm:unique-strict} holds the equilibrium is unique by hypothesis.  Conversely let $(x^*,y^*)$ be the unique equilibrium.  The equilibria of $A$ are the optimal solutions of the primal--dual pair $\max\{v: A^\top x\ge v\ones,\ \ones^\top x=1,\ x\ge0\}$ and $\min\{w: Ay\le w\ones,\ \ones^\top y=1,\ y\ge0\}$, whose common optimal value is $v^*$.  By the strict complementarity theorem of \citet{goldman1956theory} this pair has an optimal solution $(x^\circ,y^\circ)$ with $x^\circ_i>0$ exactly when $(Ay^\circ)_i=v^*$ and $y^\circ_j>0$ exactly when $(x^{\circ\top}A)_j=v^*$.  Uniqueness gives $(x^\circ,y^\circ)=(x^*,y^*)$, so every $i\notin S_1$ has $(Ay^*)_i<v^*$ and every $j\notin S_2$ has $(x^{*\top}A)_j>v^*$, i.e.\ $\delta>0$ in \eqref{eq:margin}.  \qed

\subsection{Proof of Lemma~\ref{lem:epscert}}\label{app:local:epscert}
For $t\ge t_+$ both norms are positive and nondecreasing (Lemma~\ref{lem:normmono}), so $\Nt_1(t)\Nt_2(t)$ is nondecreasing and $\eps(t)$ is nonincreasing.  By Theorem~\ref{thm:saturation} $\Nt_i(t)\uparrow\Nt_{i,\infty}$, and with $c_i=\Nt_{i,\infty}/\norm{x^*}_2$ (resp.\ $\norm{y^*}_2$) and $\etainf=1/\sqrt{c_1c_2}$,
\[
\lim_{t\to\infty}\eps(t)=\sigmax(\Bmat)\frac{\sqrt{\norm{x^*}_2\norm{y^*}_2}}{\sqrt{\Nt_{1,\infty}\Nt_{2,\infty}}}=\frac{\sigmax(\Bmat)}{\sqrt{c_1c_2}}=\etainf\sigmax(\Bmat).
\]
Monotonicity gives $\eps(t)\ge\etainf\sigmax(\Bmat)$ for every $t\ge t_+$, so $\eps(t)<1$ at one time implies Assumption~\ref{asm:stepsize} and $\eps(t')<1$ for all $t'\ge t$; conversely $\etainf\sigmax(\Bmat)<1$ implies $\eps(t)<1$ for all large $t$, so any record that extends far enough contains a crossing.  A finite record without a crossing decides nothing.  If $t_+=\infty$ one norm is identically zero, $\eps$ is undefined, and Lemma~\ref{lem:neverpos} applies instead.  \qed
\end{revblock}

\subsection{Open constants}
\label{app:local:open}

\begin{remark}[Open: quantitative absorption radius]
\label{rem:gap2}
Lemma~\ref{lem:absorb} is proved by a continuity argument, so the radius of
the absorption neighborhood $\mathcal{U}_0$ (and through it $r_0$ and, via
\eqref{eq:Mconst}, the constants $r_1$, $\delta_1$, $C_\eps$) is
\emph{existential}.  We have not derived an explicit lower bound on it in
terms of $(\delta,\norm{A},c_1,c_2,\min_{i\in S_1}x^*_i,\min_{j\in S_2}y^*_j)$ and the dimensions; a fully explicit-constants version
of the local theory would require quantifying this radius.
\end{remark}

\begin{remark}[Open: second-order constants]
\label{rem:gap4}
The curvature bound $K$ of \eqref{eq:Kdef} is obtained by compactness and is
not given in closed form; the rate
$\rho^*$ of \eqref{eq:rhostar}, the constants $C_J$, $C_E$, and the
adapted-norm constants $m_*,M_*,\kappa_*$ of Lemma~\ref{lem:adapted} are
fully explicit.  Thus in Theorem~\ref{thm:rate} the \emph{rate} is closed-form
while the multiplicative \emph{constant} $C_\eps$ is not, through $K$ and the
absorption radius alone; closing this
asymmetry is open.
\end{remark}

%
%

\section{Proofs for the Global Theory}
\label{app:global}

Throughout this appendix, $A\in\R^{n_1\times n_2}$ is the payoff matrix
(row player maximizes $x^\top Ay$), $d:=n_1+n_2$, and $\norm{A}$ denotes
the $\ell_2\to\ell_2$ operator norm. The state is
$z=(\rt_1,\rt_2)\in\R^{n_1}_{\ge0}\times\R^{n_2}_{\ge0}$ with
$N_i(z):=\norm{\rt_i}_2$, and the state space is $\mathcal Z$ as
in~\eqref{eq:statespace}. All statements refer to
Algorithm~\ref{alg:iregprm}, the strict version ($r=\rt+m-\gamma\mathbf 1$
without clipping in the prediction step; update $\pos{r+g}$).

\begin{remark}[Robustness to the clipped-prediction variant]
\label{rem:variant}
A legitimate variant clips already in the prediction step
($r_{\mathrm{clip}}=\pos{\rt+m-\gamma\mathbf 1}$, update
$\pos{r_{\mathrm{clip}}+g}$). \rev{The global arguments of this appendix transfer to that variant once the standing facts (E1)--(E5) are re-established for it, and the same computations give them.  Write $p:=\pos{\rt+m-\gamma\ones}=r_{\mathrm{clip}}$ with $\gamma$ chosen so that $\norm p_2=\norm{\rt}_2$.  The energy identity (Lemma~\ref{lem:energy-identity}) is stated for this $p$, so $\beta=\inner{m}{x}-\gamma\ge0$ with the same lower bound; componentwise $p\ge\rt+m-\gamma\ones$, hence $\rt^+=\pos{p+g}\ge p+g\ge\rt+u-\inner{u}{x}\ones+\beta\ones$, which is the inequality telescoped in Step~2 of the proof of Lemma~\ref{lem:sat-equiv} (it uses $g=(u-m)-\inner{u-m}{x}\ones$), so the $\beta$-budget (E4), saturation (E2, Theorem~\ref{thm:saturation}) and (E3) follow verbatim; with $g=(u-m)-\inner{u-m}{x}\ones$ the correction of Algorithm~\ref{alg:iregprm}, $\inner{g}{p}=0$ because $g$ is centered against $x=p/\norm p_1$, so Lemma~\ref{lem:clip-rigidity} gives $\norm{\rt^+}_2^2\in[\norm p_2^2,\norm p_2^2+\norm g_2^2]$, i.e.\ (E1) and (F3); (F2) is unchanged; and the fixed-point characterization (Proposition~\ref{prop:fixedpoint}) uses only the energy identity, clipping rigidity and minimax, with $p$ in the role of $\pos{r}$.  The remaining global proofs use (a) norm preservation and nonnegativity of $p$, (b) the $1$-Lipschitz property of $\pos{\cdot}$, and (c) that the update outputs $p$ when $g=0$.} \rev{The local theory is a different matter: clipping in the prediction step removes the negative buffer $r_i\le-\delta/2$ on which the absorption argument rests (schematically, $p=(1,0)$ with $g=(0,\varepsilon)$ gives $\pos{p+g}=(1,\varepsilon)$, whereas the strict predictor $r=(1,-\delta)$ keeps the coordinate at zero for $\varepsilon<\delta$), so exact support freezing (Lemma~\ref{lem:absorb}), the analyticity arguments of Appendix~\ref{app:global-selfstab}, and everything built on them are asserted for the strict version only; no local theorem is claimed for the clipped variant.}
\end{remark}

\subsection{The norm-preserving pairing lemma}\label{app:pairing}
\begin{lemma}[Norm-preserving pairing]
\label{lem:pairing}
Let $a,b\in\R^n_{\ge0}\setminus\{0\}$ with $\norm{a}_2=\norm{b}_2=N$, and let $x:=a/\norm{a}_1$, $\xt:=b/\norm{b}_1$.  Then
\begin{equation}
a-b \;=\; N\Bigl(\tfrac{x}{\norm{x}_2}-\tfrac{\xt}{\norm{\xt}_2}\Bigr), \qquad\text{hence}\qquad \norm{a-b}_2 \;\le\; 2\sqrt{n}\,N\,\norm{x-\xt}_2 .
\label{eq:pairing}
\end{equation}
Applied with $a=\pos{r_1(z)}$, $b=\rt_1$ (equal norms by the $\gamma$-shift), it bounds the prediction-step displacement by the pre-iterate movement, which is square-summable (Lemma~\ref{lem:sat-equiv}; Appendix~\ref{sec:global:proofs}, \eqref{eq:pairing-applied}).
\end{lemma}

\subsection{The round map and standing facts}
\label{app:global-setup}

\paragraph{The round map.}
On $\mathcal Z$ define $\Phi:\mathcal Z\to\mathcal Z$ (extragradient
setup, the two players symmetric):
\begin{enumerate}[label=(\arabic*),itemsep=1pt,topsep=2pt]
\item pre-iterates $\tilde X(z)=\rt_1/\norm{\rt_1}_1$,
$\tilde Y(z)=\rt_2/\norm{\rt_2}_1$;
\item predictions $m_1(z)=A\tilde Y(z)$ (note: through the
\emph{opponent's} component $\rt_2$), $m_2(z)=-A^\top\tilde X(z)$;
\item $\gamma$-shift: $v_1=\rt_1+m_1$,
$\gamma_1(z)=\gamma(v_1,N_1(z))$ (the unique solution provided by
Lemma~\ref{lem:gamma-cont} below), $r_1(z)=v_1-\gamma_1\mathbf 1$, and
main iterate $X(z)=\pos{r_1}/\norm{\pos{r_1}}_1$; $r_2$, $Y$ symmetric;
\item residuals $\mathrm{res}_1(z)=A\bigl(Y(z)-\tilde Y(z)\bigr)$ and
$g_1(z)=\mathrm{res}_1-\inner{\mathrm{res}_1}{X(z)}\mathbf 1$;
$g_2$ symmetric;
\item $\Phi(z):=\bigl(\pos{r_1+g_1},\,\pos{r_2+g_2}\bigr)$.
\end{enumerate}
\emph{Well-definedness.} $\norm{\rt_i}_1\ge\norm{\rt_i}_2>0$ on
$\mathcal Z$; $\gamma_i$ exists and is unique (Lemma~\ref{lem:gamma-cont}
with target $c=N_i>0$); $\norm{\pos{r_i}}_1\ge\norm{\pos{r_i}}_2=N_i>0$ by
norm preservation; and $\Phi(\mathcal Z)\subseteq\mathcal Z$ because the
update step does not decrease the norm:
$\norm{\pos{r_i+g_i}}_2\ge\norm{\pos{r_i}}_2=N_i>0$
(Lemma~\ref{lem:normmono}).

\paragraph{Trajectory--map alignment.}
Let $t_+:=\min\{t:N_1(t)>0\text{ and }N_2(t)>0\}=\max(t_1^*,t_2^*)$. The
standing assumption $t_+<\infty$ (Section~\ref{sec:global}) is exactly
$N_{i,\infty}>0$ for both players (Theorem~\ref{thm:saturation} supplies
finiteness), and norm monotonicity gives
$N_i(t)\ge N_i(t_+)>0$ for all $t\ge t_+$. For $t<t_i^*$ player $i$
follows its zero-state branch (Algorithm~\ref{alg:iregprm}: its
pre-iterate is carried over from the previous round, starting from the
default uniform, its prediction is $\mathbf 0$ and its main
iterate equals its pre-iterate), while a player with $t_i^*\le t<t_+$
already runs the ordinary prediction step.  For $t\ge t_+$ we have $z^{(t)}\in\mathcal Z$ and, round by
round, $z^{(t+1)}=\Phi(z^{(t)})$, with
$x^{(t)}=X(z^{(t)})$, $\xt^{(t)}=\tilde X(z^{(t)})$,
$y^{(t)}=Y(z^{(t)})$, $\yt^{(t)}=\tilde Y(z^{(t)})$,
$g_i^{(t)}=g_i(z^{(t)})$. That is, from $t_+$ on the dynamics is a
\emph{memoryless autonomous} discrete system on $\mathcal Z$; we analyze
only the tail $\{z^{(t)}\}_{t\ge t_+}$.

\paragraph{Standing facts.}
We collect the facts quoted in the proofs, with their sources.
\begin{enumerate}[label=(E\arabic*),itemsep=2pt,topsep=2pt]
\item $N_i(t)$ is nondecreasing along the whole trajectory
(Lemma~\ref{lem:normmono}).
\item \emph{(Saturation)} $N_i(t)\uparrow N_{i,\infty}<\infty$
(Theorem~\ref{thm:saturation}) and $N_{i,\infty}>0$ (the standing
assumption $t_+<\infty$), for both players.
\item By Lemma~\ref{lem:sat-equiv} (first inequality of
\eqref{eq:sat-explicit}, established self-containedly in
Appendix~\ref{app:local:satequiv}), saturation of both players
implies
\[
\sum_{t\ge t_+}\norm{x^{(t)}-\xt^{(t)}}_2^2\;\le\;
\frac{2n_1\,(N_{1,\infty}+N_{2,\infty})}{N_1(t_+)}\;<\;\infty,
\]
(the tail sum from $t_+$, which is all the analysis below uses; the
full-sum version carries the norm at the first positive-norm round in
the denominator), and symmetrically for player $2$; hence
$\norm{x^{(t)}-\xt^{(t)}}\to0$ and $\norm{y^{(t)}-\yt^{(t)}}\to0$ along
the \emph{full} sequence (summability implies vanishing terms).
\item \emph{(Scalar core inequality; second inequality of
\eqref{eq:beta-free})}
\begin{equation}
\beta_t:=\inner{m^{(t)}}{x^{(t)}}-\gamma^{(t)}
\;\ge\;\frac{1}{2n_1}\,N_1(t)\,\norm{x^{(t)}-\xt^{(t)}}^2\;\ge\;0,
\qquad
\sum_t\beta_t<\infty ,
\label{eq:beta}
\end{equation}
so $\beta_t\to0$ along the full sequence.
\phantomsection\label{fact:beta}
\item By Proposition~\ref{prop:fixedpoint}: for $z\in\mathcal Z$,
$\Phi(z)=z$ iff $(X(z),Y(z))=(\tilde X(z),\tilde Y(z))$ is a Nash
equilibrium with corresponding supports;
$\Fix(\Phi)=\{(c_1x^*,c_2y^*):c_1,c_2>0,\ (x^*,y^*)\in\mathrm{NE}\}$ (the
positive cone of NE rays), and at fixed points $\gamma_1=v^*$ (the game
value).
\item \emph{(Residual bound)} See Lemma~\ref{lem:gbound} below.
\item By Lemmas~\ref{lem:jacobian},~\ref{lem:centering}
and~\ref{lem:kernel}, all stated under
Assumption~\ref{asm:nondeg-strict} (nondegenerate, strictly complementary
equilibrium; equivalent to Assumption~\ref{asm:unique-strict} by
Lemma~\ref{lem:kernel}): at such an NE ray point
$z^*=(c_1x^*,c_2y^*)$,
\[
\spec\bigl(D\Phi_S(z^*)\bigr)=\{1,1\}\,\cup\,
\bigl\{\,1-\etainf^2\sigma^2\pm i\,\etainf\sigma\,\bigr\},
\qquad \etainf=1/\sqrt{c_1c_2},
\]
with $\sigma$ ranging over the nonzero singular values of
$\Bmat=\Atil-v^*\mathbf 1\mathbf 1^\top$. This is the spectrum statement of
Theorem~\ref{thm:rate}.
\end{enumerate}

\begin{lemma}[Residual bound]
\label{lem:gbound}
For all $t\ge t_+$,
$\norm{g_1^{(t)}}_2\le(1+\sqrt{n_1})\,\norm{A}\,\norm{y^{(t)}-\yt^{(t)}}_2$,
and symmetrically for player $2$.
\end{lemma}

\begin{proof}
$g_1=\mathrm{res}_1-\inner{\mathrm{res}_1}{x}\mathbf 1$ with
$x\in\simplex{n_1}$. Then
$|\inner{\mathrm{res}_1}{x}|\le\norm{\mathrm{res}_1}_\infty\norm{x}_1
\le\norm{\mathrm{res}_1}_2$, $\norm{\mathbf 1}_2=\sqrt{n_1}$, and
$\norm{\mathrm{res}_1}_2=\norm{A(y-\yt)}_2\le\norm{A}\,\norm{y-\yt}_2$.
\end{proof}
\subsection{Proofs}
\label{sec:global:proofs}

This subsection proves Theorems~\ref{thm:global} and~\ref{thm:point}.  Both
rest on continuity properties of the round map $\Phi$:
joint continuity of the $\gamma$-shift across support
changes, continuity of every intermediate map, the residual bound, and the
two elementary facts that a bounded sequence with vanishing steps has a
connected set of limit points.  These are collected, with proofs, in
Appendix~\ref{app:global}, and are used here as stated.  The proofs of the
norm-preserving pairing Lemma~\ref{lem:pairing}, of
Corollary~\ref{cor:capture}, and of Proposition~\ref{prop:selfstab} are also
in Appendix~\ref{app:global}.

\subsubsection{Proof of Theorem~\ref{thm:global}}
\label{app:global-thmA}

We prove the two parts as two propositions, matching the structure of the
theorem.

\begin{proposition}[$\omega$-limit set: existence, structure, invariance]
\label{prop:omega}
Under the hypotheses of Theorem~\ref{thm:global}, the tail trajectory
$\{z^{(t)}\}_{t\ge t_+}$ is contained in the compact set
\[
K:=\{z:\ \rt_i\ge0,\ \norm{\rt_i}_2\le N_{i,\infty},\ i=1,2\} ,
\]
its $\omega$-limit set $\Omega$ is nonempty, compact, connected (the
connectedness is of independent interest and not used in the sequel),
satisfies $\Phi(\Omega)=\Omega$, and
\[
\Omega\;\subseteq\;\mathcal L:=\{z:\ \rt_i\ge0,\
\norm{\rt_i}_2=N_{i,\infty},\ i=1,2\}\;\subset\;\mathcal Z .
\]
Moreover $\dist(z^{(t)},\Omega)\to0$.
\end{proposition}

\begin{proof}
\emph{Boundedness.} $\rt_i^{(t)}\ge0$ always (output of $\pos{\cdot}$),
and $\norm{\rt_i^{(t)}}_2=N_i(t)\le N_{i,\infty}$ by (E1)--(E2); hence
$z^{(t)}\in K$, and $K$ is compact.

\emph{Nonempty, compact, connected.} Nonemptiness, compactness, and
connectedness follow from Lemma~\ref{lem:displacement}
(which uses (E3) and Lemma~\ref{lem:gbound}) and
Lemma~\ref{lem:limitset}.

\emph{$\Omega\subseteq\mathcal L$.} If $\bar z=\lim_kz^{(t_k)}$ then
$\norm{\rt_i(\bar z)}_2=\lim_kN_i(t_k)=N_{i,\infty}$ (norms are
continuous, and a monotone sequence has the same limit along every
subsequence). Since $N_{i,\infty}>0$ we get
$\mathcal L\subset\mathcal Z$, so \rev{$\Phi$ is continuous at every point of $\Omega$ relative to $\mathcal Z$ (Lemma~\ref{lem:phi-cont}); all iterates are feasible, so this suffices for the subsequence arguments below.}

\emph{Positive invariance $\Phi(\Omega)\subseteq\Omega$.} If
$\bar z\in\Omega$ with $z^{(t_k)}\to\bar z$, then
$z^{(t_k+1)}=\Phi(z^{(t_k)})\to\Phi(\bar z)$ by continuity; the left-hand
side has its limit in $\Omega$, so $\Phi(\bar z)\in\Omega$.

\emph{Reverse inclusion $\Omega\subseteq\Phi(\Omega)$.} Let
$\bar z=\lim_kz^{(t_k)}$ with $t_k\ge t_++1$. The sequence
$\{z^{(t_k-1)}\}\subset K$ has an accumulation point $w\in\Omega$; along a
sub-subsequence, $\Phi(w)=\lim\Phi(z^{(t_k-1)})=\lim z^{(t_k)}=\bar z$.
Hence $\Phi(\Omega)=\Omega$.

\emph{$\dist(z^{(t)},\Omega)\to0$.} If not, then
$\dist(z^{(t)},\Omega)\ge\eps$ along a subsequence; by compactness pass to
a convergent sub-subsequence, whose limit lies in $\Omega$, a
contradiction.
\end{proof}

All constants in this proof are explicit and involve no hidden
game-dependent quantities: the displacement constants are
$(1+\sqrt{n_i})\norm A$ and $2\sqrt{n_i}N_{i,\infty}$, and the
summability constant is $2n_i(N_{1,\infty}+N_{2,\infty})/N_i(t_+)$,
\rev{all explicit in the dimensions, $\norm A$, the positive tail norms and the saturation limits; finite-run upper bounds follow by replacing $N_{i,\infty}$ with the explicit bound of Theorem~\ref{thm:saturation}.}

\begin{proposition}[Limit points are fixed points]
\label{prop:omega-fix}
Under the hypotheses of Theorem~\ref{thm:global},
$\Omega\subseteq\Fix(\Phi)$. Combined with (E5): every point of $\Omega$
has the form $(c_1x^*,c_2y^*)$ with $(x^*,y^*)\in\mathrm{NE}$, the scales
determined by the level-set condition
$\norm{c_1x^*}_2=N_{1,\infty}$ and $\norm{c_2y^*}_2=N_{2,\infty}$, and
$\gamma_1=v^*$ at each such point.
\end{proposition}

\begin{proof}
Fix $\bar z\in\Omega$ with $z^{(t_k)}\to\bar z$. By
Proposition~\ref{prop:omega}, $\bar z\in\mathcal L\subset\mathcal Z$, so
all maps of Lemma~\ref{lem:phi-cont} are continuous at $\bar z$.

\emph{Step 1 (pre-iterate $=$ main iterate).} By the trajectory--map
alignment, $x^{(t_k)}=X(z^{(t_k)})$ and
$\xt^{(t_k)}=\tilde X(z^{(t_k)})$; continuity gives
$x^{(t_k)}\to X(\bar z)$ and $\xt^{(t_k)}\to\tilde X(\bar z)$. Both are
limits of \emph{images of the same sequence under continuous maps}, so no
extra uniformity assumption is needed; the dependence of $m_1(\bar z)$ on
the opponent component $\rt_2(\bar z)$ through $\tilde Y(\bar z)$ is
already absorbed in the definition of $X$. By (E3) the full sequence
$\norm{x^{(t)}-\xt^{(t)}}\to0$, so
\[
\norm{X(\bar z)-\tilde X(\bar z)}
=\lim_k\norm{x^{(t_k)}-\xt^{(t_k)}}=0 ,
\]
that is, $X(\bar z)=\tilde X(\bar z)$; symmetrically
$Y(\bar z)=\tilde Y(\bar z)$.

\emph{Step 2 (prediction step stationary).} At $\bar z$, take
$a=\pos{r_1(\bar z)}$ and $b=\rt_1(\bar z)$: both are nonnegative,
nonzero, and $\norm{a}_2=\norm{b}_2=N_1(\bar z)$ (the norm-preservation
equation of the $\gamma$-shift holds at $\bar z$ by definition). Step~1
says their $\ell_1$-normalizations coincide, so the identity of
Lemma~\ref{lem:pairing} gives $a=b$:
\[
\pos{r_1(\bar z)}=\rt_1(\bar z),
\qquad
\pos{r_2(\bar z)}=\rt_2(\bar z) .
\]
(Coordinates with $\rt_i=0$ are handled automatically:
Lemma~\ref{lem:pairing} is a vector identity.)

\emph{Step 3 (update step stationary).}
$\mathrm{res}_1(\bar z)=A\bigl(Y(\bar z)-\tilde Y(\bar z)\bigr)=\mathbf 0$
by Step~1, hence $g_1(\bar z)=\mathbf 0$; symmetrically
$g_2(\bar z)=\mathbf 0$. Therefore
\[
\Phi(\bar z)_1=\pos{r_1(\bar z)+g_1(\bar z)}=\pos{r_1(\bar z)}
=\rt_1(\bar z),
\]
and likewise for player $2$, so $\Phi(\bar z)=\bar z$. The NE
characterization and $\gamma_1=v^*$ follow from (E5).
\end{proof}

\begin{proof}[Proof of Theorem~\ref{thm:global}]
Part (i) is Propositions~\ref{prop:omega} and~\ref{prop:omega-fix}.

Part (ii): the function $\gapf$ is continuous in $(x,y)$ and vanishes on
the NE set. The sequence $\gapf(x^{(t)},y^{(t)})$ is bounded; take any of
its limit points, realized along times $t_k$. By compactness pass to a
sub-subsequence with $z^{(t_k)}\to\bar z\in\Omega$; by continuity
(Lemma~\ref{lem:phi-cont}) and
Proposition~\ref{prop:omega-fix} with (E5),
$(x^{(t_k)},y^{(t_k)})\to(X(\bar z),Y(\bar z))\in\mathrm{NE}$, so the
limit point equals $0$. A bounded sequence all of whose limit points are
$0$ converges to $0$. The distance statement is identical (the NE set is
compact and $\dist(\cdot,\mathrm{NE})$ continuous), and the pre-iterate
statements follow from (E3).
\end{proof}

\begin{lemma}[Equilibrium potential]\label{lem:potential}
Suppose that $t_+<\infty$, and let $(a,b)$ be an arbitrary Nash equilibrium of $A$.  We define
$H_{a,b}(t):=\inner{a}{\rt_1^{(t)}}+\inner{b}{\rt_2^{(t)}}$.  Then, for every $t\ge t_+$,
\begin{equation}\label{eq:potential-step}
H_{a,b}(t+1)-H_{a,b}(t)\;\ge\;\bigl(a^\top Ay^{(t)}-v^*\bigr)+\bigl(v^*-{x^{(t)}}^{\!\top}Ab\bigr)+\beta_t^{1}+\beta_t^{2}\;\ge\;0,
\end{equation}
and $0\le H_{a,b}(t)\le N_{1,\infty}+N_{2,\infty}$.  In particular, $H_{a,b}(t)$ converges as $t\to\infty$.
\end{lemma}
\begin{proof}
We fix $t\ge t_+$. Let $u_1^{(t)}:=Ay^{(t)}$ and $u_2^{(t)}:=-A^\top x^{(t)}$ denote the observed payoff vectors, and write $x_1^{(t)}:=x^{(t)}$ and $x_2^{(t)}:=y^{(t)}$.  The componentwise bound established in Step~2 of the proof of Lemma~\ref{lem:sat-equiv} relies solely on the property $\pos{v}\ge v$ along with the definitions of $g_i^{(t)}$ and $\beta_t^{i}$.  Consequently, it yields
\[
\rt_i^{(t+1)}-\rt_i^{(t)}\;\ge\;u_i^{(t)}-\inner{u_i^{(t)}}{x_i^{(t)}}\ones+\beta_t^{i}\ones,\qquad i=1,2.
\]
Since $a\ge0$ and $\ones^\top a=1$, taking the inner product of the bound for $i=1$ with $a$ yields
$\inner{a}{\rt_1^{(t+1)}-\rt_1^{(t)}}\ge a^\top Ay^{(t)}-{x^{(t)}}^{\!\top}Ay^{(t)}+\beta_t^{1}$.
Similarly, taking the inner product of the bound for $i=2$ with $b$ gives
$\inner{b}{\rt_2^{(t+1)}-\rt_2^{(t)}}\ge -{x^{(t)}}^{\!\top}Ab+{x^{(t)}}^{\!\top}Ay^{(t)}+\beta_t^{2}$.
Summing these two bounds cancels the realized payoff ${x^{(t)}}^{\!\top}Ay^{(t)}$, thereby establishing the first inequality in \eqref{eq:potential-step}.  Next, the second inequality follows from the saddle-point condition $a^\top Ay\ge v^*\ge x^\top Ab$, valid for all $x\in\simplex{n_1}$ and $y\in\simplex{n_2}$ at $(a,b)$, together with the nonnegativity $\beta_t^{i}\ge0$ guaranteed by \eqref{eq:beta}.  Furthermore, since $\rt_i^{(t)}\ge0$, it follows that $H_{a,b}(t)\ge0$. Applying the Cauchy--Schwarz inequality in conjunction with $\norm{a}_2,\norm{b}_2\le1$, we also obtain $H_{a,b}(t)\le N_1(t)+N_2(t)\le N_{1,\infty}+N_{2,\infty}$.  Because any bounded and nondecreasing sequence converges, the proof is complete.
\end{proof}

\subsubsection{Proof of Theorem~\ref{thm:point}}
\label{app:global-thmB}

\begin{proof}
We first treat the case $t_+<\infty$. By Proposition~\ref{prop:omega}, the limit set $\Omega$ is nonempty. We fix $\bar z\in\Omega$ together with a subsequence $z^{(t_k)}\to\bar z$. By Proposition~\ref{prop:omega-fix}, we have $\bar z=(\bar c_1\bar x,\bar c_2\bar y)$ for a Nash equilibrium $(\bar x,\bar y)$, where $\bar c_1=N_{1,\infty}/\norm{\bar x}_2$ and $\bar c_2=N_{2,\infty}/\norm{\bar y}_2$. We keep $(\bar x,\bar y)$ fixed in what follows. For $t\ge t_+$, we define
\[
D_1(t):=N_1(t)\norm{\bar x}_2-\inner{\bar x}{\rt_1^{(t)}},\qquad
D_2(t):=N_2(t)\norm{\bar y}_2-\inner{\bar y}{\rt_2^{(t)}},
\]
which yields $D_1(t)+D_2(t)=N_1(t)\norm{\bar x}_2+N_2(t)\norm{\bar y}_2-H_{\bar x,\bar y}(t)$. By virtue of the Cauchy--Schwarz inequality, both $D_1(t)$ and $D_2(t)$ are nonnegative. Moreover, the norms converge by (E2) and $H_{\bar x,\bar y}(t)$ converges by Lemma~\ref{lem:potential}, from which it follows that $D_1(t)+D_2(t)$ admits a limit. Along the chosen subsequence, we obtain $D_1(t_k)\to N_{1,\infty}\norm{\bar x}_2-\bar c_1\norm{\bar x}_2^2=0$, and analogously $D_2(t_k)\to0$. Consequently, the limit of $D_1(t)+D_2(t)$ is zero. Since both summands are nonnegative, we conclude that $D_1(t)\to0$ and $D_2(t)\to0$ along the full sequence. Next, we observe that
\[
\Bigl\|\frac{\rt_1^{(t)}}{N_1(t)}-\frac{\bar x}{\norm{\bar x}_2}\Bigr\|_2^2
\;=\;2-\frac{2\inner{\bar x}{\rt_1^{(t)}}}{N_1(t)\norm{\bar x}_2}
\;=\;\frac{2\,D_1(t)}{N_1(t)\norm{\bar x}_2}
\;\le\;\frac{2\sqrt{n_1}\,D_1(t)}{N_1(t_+)}\;\longrightarrow\;0,
\]
where the inequality relies on $N_1(t)\ge N_1(t_+)>0$ from (E1) and $\norm{\bar x}_2\ge1/\sqrt{n_1}$. Combined with $N_1(t)\to N_{1,\infty}$, this yields $\rt_1^{(t)}\to N_{1,\infty}\bar x/\norm{\bar x}_2$. The exact same reasoning applies to player~2, thereby establishing $z^{(t)}\to\bar z$. Convergence of the strategies then follows from Lemma~\ref{lem:phi-cont} and (E5): in particular, $\xt^{(t)}=\tilde X(z^{(t)})\to\tilde X(\bar z)=\bar x$ and $x^{(t)}=X(z^{(t)})\to X(\bar z)=\bar x$, and analogously $\yt^{(t)},y^{(t)}\to\bar y$. Under uniqueness, $(\bar x,\bar y)$ coincides with the unique equilibrium. Finally, the case $t_+=\infty$ is established in Lemma~\ref{lem:neverpos}(v).
\end{proof}

\subsubsection{Never-positive players: the case $t_+=\infty$}
\label{app:global-neverpos}

Theorem~\ref{thm:global} and part~(i) of Theorem~\ref{thm:point} operate under the assumption that $t_+<\infty$. We next settle the elementary complementary case via the following lemma. Consequently, both the gap guarantee of Theorem~\ref{thm:global}(ii) and convergence to a single equilibrium hold on every matrix game.

\begin{lemma}[Never-positive players]\label{lem:neverpos}
Suppose $t_1^*=\infty$, i.e.\ $\rt_1^{(t)}=\mathbf 0$ for all $t$
(the case $t_2^*=\infty$ follows by exchanging the two players). Write $\bar x:=\ones/n_1$ and
$\ell:=-A^\top\bar x\in\R^{n_2}$, and let $K_\ell:=\{k:\ell_k=\max_{k'}\ell_{k'}\}$.
\begin{enumerate}[label=(\roman*),itemsep=1pt,topsep=2pt]
\item $x^{(t)}=\xt^{(t)}=\bar x$ for all $t$, the vector $Ay^{(t)}$ is
constant across coordinates for all $t$, and player~1's exploitability
$\max_i(Ay^{(t)})_i-\bar x^\top Ay^{(t)}$ is $0$ at every round.
\item If $t_2^*<\infty$, then for all $t\ge t_2^*$: $N_2(t)=N_2(t_2^*)$,
$\gamma_2^{(t)}\le\max_k\ell_k$, and
$\sum_{t\ge t_2^*}\bigl(\max_k\ell_k-\gamma_2^{(t)}\bigr)\le N_2(t_2^*)$.
\item If $t_2^*<\infty$, there is a finite $T$ such that $y^{(t)}$ is
supported on $K_\ell$ for all $t\ge T$; hence
$\gapf(x^{(t)},y^{(t)})=0$ for all $t\ge T$.
\item If also $t_2^*=\infty$, the uniform profile $(\ones/n_1,\ones/n_2)$
is a Nash equilibrium and $\gapf(x^{(t)},y^{(t)})=0$ for all $t$.
\item In both cases, the strategy pair $(x^{(t)},y^{(t)})$ remains constant for all sufficiently large $t$, and this constant value constitutes a Nash equilibrium.
\end{enumerate}
In particular, the strategy pair is eventually constant at a Nash equilibrium, and $\gapf(x^{(t)},y^{(t)})=0$ from then on.
\end{lemma}

\begin{proof}
(i) With $\rt_1^{(t)}=\mathbf 0$ the zero branch of
Algorithm~\ref{alg:iregprm} keeps $\xt_1^{(t)}=\xt_1^{(1)}=\bar x$, sets
$m_1^{(t)}=r_1^{(t)}=\mathbf 0$, and plays $x^{(t)}=\bar x$. Then
$g_1^{(t)}=u^{(t)}-\inner{u^{(t)}}{\bar x}\ones$ with $u^{(t)}=Ay^{(t)}$,
and $\rt_1^{(t+1)}=\pos{g_1^{(t)}}=\mathbf 0$ forces
$u^{(t)}_i\le\inner{u^{(t)}}{\bar x}$ for every $i$. The mean of the
$u^{(t)}_i$ equals $\inner{u^{(t)}}{\bar x}$, so all $u^{(t)}_i$ are
equal, and $\max_i(Ay^{(t)})_i-\bar x^\top Ay^{(t)}=0$.

(ii) Player~2 faces the fixed opponent $\bar x$, so its prediction and its
observed payoff coincide: $m_2^{(t)}=u_2^{(t)}=\ell$ for all $t\ge t_2^*$
(before $t_2^*$ the zero branch plays $y^{(t)}=\bar y$ and sets
$m_2^{(t)}=0$). Its
update reads $\rt_2^{(t+1)}=\pos{\rt_2^{(t)}+\ell-\gamma_2^{(t)}\ones}$,
where $\gamma_2^{(t)}$ is the norm-preserving shift, so
$N_2(t+1)=N_2(t)$ for $t\ge t_2^*$. If $\gamma_2^{(t)}>\max_k\ell_k$,
every coordinate with $\rt_{2,k}^{(t)}>0$ strictly decreases and the
others stay at $0$; since some coordinate is positive, the norm would
strictly decrease, a contradiction. Hence $\gamma_2^{(t)}\le\max_k\ell_k$.
For $k^*\in K_\ell$ this gives
$\rt_{2,k^*}^{(t+1)}\ge\rt_{2,k^*}^{(t)}+\max_k\ell_k-\gamma_2^{(t)}\ge
\rt_{2,k^*}^{(t)}$, and $\rt_{2,k^*}^{(t)}\le N_2(t_2^*)$ for all
$t\ge t_2^*$, so
the nonnegative increments are summable with the stated bound.

(iii) By (ii), $\gamma_2^{(t)}\to\max_k\ell_k$. Fix $j\notin K_\ell$ and
$\delta_j:=\max_k\ell_k-\ell_j>0$; for all large $t$,
$\ell_j-\gamma_2^{(t)}\le-\delta_j/2$, so
$\rt_{2,j}^{(t+1)}\le\pos{\rt_{2,j}^{(t)}-\delta_j/2}$, which reaches $0$
in finitely many rounds and then stays at $0$. Since
$y^{(t)}\propto\pos{\rt_2^{(t)}+\ell-\gamma_2^{(t)}\ones}=\rt_2^{(t+1)}$
for $t\ge t_2^*$ (line~8 of Algorithm~\ref{alg:iregprm}; before $t_2^*$,
$y^{(t)}=\bar y$ by the zero branch), $y^{(t)}$ is supported on $K_\ell$
from some finite $T$ on. Then $\bar x^\top Ay^{(t)}=-\max_k\ell_k=
\min_j(\bar x^\top A)_j$, and with (i) the Nash gap is $0$.

(iv) If both states stay at zero, both players play uniformly and, by
(i) applied to each, each player's payoff vector is constant across
coordinates, which is the Nash condition for the uniform profile.

(v) In case (iv), the pair is constantly uniform. We next treat case (iii). The proof of (iii) yields $T_0\ge t_2^*$ such that, for all $t\ge T_0$ and all $j\notin K_\ell$, both $\rt_{2,j}^{(t)}=0$ and $\ell_j-\gamma_2^{(t)}<0$ hold. For $t\ge T_0$, the update $\rt_2^{(t+1)}=\pos{\rt_2^{(t)}+\ell-\gamma_2^{(t)}\ones}$ preserves the coordinates outside $K_\ell$ at $0$ and adds the common quantity $c_t:=\max_k\ell_k-\gamma_2^{(t)}$ to each coordinate in $K_\ell$. By virtue of (ii), we have $c_t\ge0$, which ensures that no clipping occurs on $K_\ell$. Consequently, norm preservation $N_2(t+1)=N_2(t)$ yields $2c_t\sum_{k\in K_\ell}\rt_{2,k}^{(t)}+|K_\ell|\,c_t^2=0$. Since $\sum_{k\in K_\ell}\rt_{2,k}^{(t)}=\norm{\rt_2^{(t)}}_1>0$, it follows that $c_t=0$. We thus obtain $\rt_2^{(t+1)}=\rt_2^{(t)}$ for all $t\ge T_0$. In particular, $y^{(t)}=\rt_2^{(t+1)}/\norm{\rt_2^{(t+1)}}_1$ is constant, $x^{(t)}=\bar x$ by (i), and the resulting constant pair attains zero gap by virtue of (iii).
\end{proof}

\subsection{Preliminary lemmas}
\label{app:global-lemmas}

\begin{lemma}[Joint continuity of the $\gamma$-shift]
\label{lem:gamma-cont}
Let $D:=\{(v,c):v\in\R^n,\ c>0\}$. For each $(v,c)\in D$ the equation
$\norm{\pos{v-\gamma\mathbf 1}}_2=c$ has a unique solution $\gamma(v,c)$,
which satisfies $\gamma(v,c)\in[\max_iv_i-c,\ \max_iv_i)$. Moreover
$\gamma$ is continuous on $D$, and for fixed $c$ the map
$\gamma(\cdot,c)$ is $1$-Lipschitz in the $\ell_2$ sense.
\end{lemma}

\begin{proof}
Write $f(\gamma;v):=\norm{\pos{v-\gamma\mathbf 1}}_2$; it is jointly
continuous in $(\gamma,v)$ (composition of $\pos{\cdot}$ and the norm).
For $\gamma\ge\max_iv_i$ we have $f=0<c$; at $\gamma_0:=\max_iv_i-c$,
$f(\gamma_0;v)\ge v_{i^*}-\gamma_0=c$ where $i^*$ is the argmax. On
$\{f>0\}=(-\infty,\max_iv_i)$, $f$ is strictly decreasing: indeed
$f(\gamma)^2=\sum_{i:v_i>\gamma}(v_i-\gamma)^2$, and as long as the active
set is nonempty (which it is whenever $f>0$), each active term strictly
decreases in $\gamma$. Hence the solution exists, is unique, and lies in
the stated interval.

\emph{Continuity.} Let $(v_k,c_k)\to(v,c)$ and
$\gamma_k:=\gamma(v_k,c_k)\in[\max_iv_{k,i}-c_k,\ \max_iv_{k,i}]$, a
bounded sequence. For any accumulation point $\bar\gamma$, joint
continuity along a subsequence gives $f(\bar\gamma;v)=c$, so by uniqueness
$\bar\gamma=\gamma(v,c)$; a bounded sequence all of whose accumulation
points coincide converges, so $\gamma_k\to\gamma(v,c)$.

\emph{$1$-Lipschitz.} By coordinatewise monotonicity of the shift,
$f(\gamma;v+\delta)\le f(\gamma-\norm{\delta}_\infty;v)$ and
$f(\gamma;v+\delta)\ge f(\gamma+\norm{\delta}_\infty;v)$; combined with
the monotonicity of $f$ in $\gamma$ and uniqueness of the solution,
$|\gamma(v+\delta,c)-\gamma(v,c)|\le\norm{\delta}_\infty
\le\norm{\delta}_2$. (This is consistent with the gradient identity
$\nabla_v\gamma=x\in\simplex{n}$ from the anchor analysis.)
\end{proof}

Continuity holds \emph{across support changes}: the
argument uses only monotonicity and uniqueness, with no piecewise case
analysis.

\begin{lemma}[Continuity of $\Phi$ and all intermediate maps]
\label{lem:phi-cont}
The maps $\tilde X,\tilde Y,m_i,\gamma_i,r_i,X,Y,g_i,\Phi$, as well as
$\beta(z):=\inner{m_1(z)}{X(z)}-\gamma_1(z)$, are continuous on
$\mathcal Z$.
\end{lemma}

\begin{proof}
Compose step by step. $\tilde X=\rt_1/\norm{\rt_1}_1$ is continuous (the
denominator is $\ge N_1(z)>0$); $m_1=A\tilde Y$ is continuous (through the
opponent component $\rt_2$); $v_1=\rt_1+m_1$ and $c_1=N_1(z)$ are
continuous with $c_1>0$, so $\gamma_1=\gamma(v_1,c_1)$ is continuous by
Lemma~\ref{lem:gamma-cont}, hence so is $r_1$;
$X=\pos{r_1}/\norm{\pos{r_1}}_1$ is continuous (denominator
$\ge\norm{\pos{r_1}}_2=c_1>0$ by norm preservation); $Y$ symmetrically;
therefore $\mathrm{res}_1$, $g_1$, and finally $\Phi$ are continuous.
\end{proof}

\begin{proof}[Proof of Lemma~\ref{lem:pairing} (norm-preserving pairing)]
Since $a=\norm{a}_1x$ and
$N=\norm{a}_2=\norm{a}_1\norm{x}_2$, we get $\norm{a}_1=N/\norm{x}_2$ and
hence $a=N\,x/\norm{x}_2$; likewise $b=N\,\xt/\norm{\xt}_2$, which is the
identity in~\eqref{eq:pairing}. For the Lipschitz part, for any
$u,w\ne0$ with (w.l.o.g.) $\norm{u}_2\ge\norm{w}_2$,
\[
\left\|\frac{u}{\norm{u}_2}-\frac{w}{\norm{w}_2}\right\|_2
\;\le\;\frac{\norm{u-w}_2}{\norm{u}_2}
+\norm{w}_2\left(\frac{1}{\norm{w}_2}-\frac{1}{\norm{u}_2}\right)
\;\le\;\frac{2\,\norm{u-w}_2}{\max(\norm{u}_2,\norm{w}_2)} ,
\]
and on the simplex $\norm{x}_2\ge1/\sqrt n$, giving the constant
$2\sqrt n$. The last claim follows from the identity: if $x=\xt$ then
$a=b$ directly; the converse is trivial. Coordinates with
$x_i=\xt_i=0$ force $a_i=b_i=0$ and are absorbed by the identity, with no
separate case analysis.
\end{proof}

Applying Lemma~\ref{lem:pairing} with $a=\pos{r_1(z)}$ and $b=\rt_1$
(both nonnegative, nonzero, and of equal $\ell_2$ norm $N_1(z)$ by the
$\gamma$-shift's defining equation) yields the estimate used repeatedly below,
\begin{equation}
\label{eq:pairing-applied}
\bigl\|\pos{r_1(z)}-\rt_1\bigr\|_2 \;\le\; 2\sqrt{n_1}\,N_1(z)\,\norm{X(z)-\tilde X(z)}_2 ,
\end{equation}
where $X,\tilde X$ are the main and pre-iterate strategy maps.

\begin{lemma}[Vanishing one-step displacement]
\label{lem:displacement}
Under saturation (E2), $\norm{z^{(t+1)}-z^{(t)}}\to0$, with the explicit
bound (for $t\ge t_+$)
\[
\norm{\rt_1^{(t+1)}-\rt_1^{(t)}}_2
\;\le\;
\underbrace{(1+\sqrt{n_1})\,\norm{A}\,\norm{y^{(t)}-\yt^{(t)}}_2}_{\text{update step, Lemma~\ref{lem:gbound}}}
\;+\;
\underbrace{2\sqrt{n_1}\,N_{1,\infty}\,\norm{x^{(t)}-\xt^{(t)}}_2}_{\text{prediction step, \eqref{eq:pairing-applied}}} ,
\]
and symmetrically for player $2$.
\end{lemma}

\begin{proof}
Decompose
$\rt_1^{(t+1)}-\rt_1^{(t)}
=\bigl(\pos{r_1+g_1}-\pos{r_1}\bigr)+\bigl(\pos{r_1}-\rt_1^{(t)}\bigr)$.
The first term is bounded by $\norm{g_1^{(t)}}_2$ (the map $\pos{\cdot}$
is coordinatewise $1$-Lipschitz), hence by Lemma~\ref{lem:gbound}; the
second by~\eqref{eq:pairing-applied} together with
$N_1(t)\le N_{1,\infty}$. Both factors tend to $0$ by (E3).
\end{proof}

\begin{lemma}[Connected limit sets for vanishing-step sequences]
\label{lem:limitset}
Let $\{z^{(t)}\}$ be contained in a compact set $K\subset\R^d$ with
$\norm{z^{(t+1)}-z^{(t)}}\to0$. Then its set of limit points $\Omega$ is
nonempty, compact, and \emph{connected}.
\end{lemma}

\begin{proof}
Nonemptiness and compactness are standard
(Bolzano--Weierstrass; $\Omega=\bigcap_T\overline{\{z^{(t)}:t\ge T\}}$ is
an intersection of a nested family of compact sets). For connectedness,
suppose for contradiction that $\Omega=\Omega_A\sqcup\Omega_B$ with both
parts nonempty, compact, and disjoint; let
$\delta:=\dist(\Omega_A,\Omega_B)>0$ and let $U_A,U_B$ be the open
$\delta/3$-neighborhoods of $\Omega_A,\Omega_B$. Since both parts consist
of limit points, the sequence visits $U_A$ and $U_B$ infinitely often, so
there are infinitely many ``exit events'': choose $t_k\uparrow\infty$ with
$z^{(t_k)}\in U_A$ and
$\tau_k:=\min\{\tau>t_k:z^{(\tau)}\notin U_A\}<\infty$ (otherwise the
tail would be trapped in $U_A$ and never revisit $U_B$). For $k$ large the
step size is $<\delta/3$, so
\[
\dist\bigl(z^{(\tau_k)},\Omega_A\bigr)
\le\norm{z^{(\tau_k)}-z^{(\tau_k-1)}}+\delta/3<2\delta/3
\;\Longrightarrow\;
\dist\bigl(z^{(\tau_k)},\Omega_B\bigr)\ge\delta-2\delta/3=\delta/3 ,
\]
hence $z^{(\tau_k)}\in C:=K\setminus(U_A\cup U_B)$. The set $C$ is compact
and contains infinitely many $z^{(\tau_k)}$, so it contains a limit point
of the sequence, i.e.\ $\Omega\cap C\ne\emptyset$, contradicting
$\Omega\subseteq U_A\cup U_B$.
\end{proof}

\begin{revblock}
\subsection{The strategy-level gap from the results of \citet{zhang2025scale}}\label{app:global:zhang}
Theorem~5.5 of \citet{zhang2025scale} gives $\sum_t\norm{x^{(t)}-\xt^{(t)}}^2+\norm{y^{(t)}-\yt^{(t)}}^2<\infty$, so both displacements tend to zero.  The second ingredient is a one-step improvement inequality, for $t\ge t_+$,
\begin{equation}\label{eq:one-step-improve}
\max_{x'\in\Delta}\inner{A\yt^{(t)}}{x'-\xt^{(t)}}\;\le\;C_1\,\norm{x^{(t)}-\xt^{(t)}}_2,\qquad C_1:=2\sqrt{n_1}\,\Nt_{1,\infty}+\norm{A}_{\mathrm{op}},
\end{equation}
and symmetrically for player~2; it is the content of their Lemma~D.2, and it also follows from Lemma~\ref{lem:pairing} as follows.  Write $s:=\rt_1^{(t)}$, $m:=A\yt^{(t)}$, $\gamma:=\gamma_1$, $p:=\pos{s+m-\gamma\ones}$, so $\norm p_2=\norm s_2=\Nt_1(t)$, $\xt^{(t)}=s/\norm s_1$ and $x^{(t)}=p/\norm p_1$.  Componentwise $m-\gamma\ones\le p-s$, with equality where $p_i>0$, so $\inner{m-\gamma\ones}{p}=\inner{p-s}{p}=\tfrac12\norm{p-s}_2^2\ge0$ by equal norms, i.e.\ $\inner{m}{x^{(t)}}\ge\gamma$.  Hence $\max_{x'}\inner{m}{x'-x^{(t)}}=\max_i(m_i-\gamma)-(\inner m{x^{(t)}}-\gamma)\le\max_i(p_i-s_i)\le\norm{p-s}_2\le2\sqrt{n_1}\,\Nt_1(t)\,\norm{x^{(t)}-\xt^{(t)}}_2$ by \eqref{eq:pairing}, and moving the base point from $x^{(t)}$ to $\xt^{(t)}$ costs $\inner{m}{x^{(t)}-\xt^{(t)}}\le\norm A_{\mathrm{op}}\norm{x^{(t)}-\xt^{(t)}}_2$; the bound $\Nt_1(t)\le\Nt_{1,\infty}$ gives \eqref{eq:one-step-improve}.  In zero-sum self-play the two left-hand sides sum to $\gapf(\xt^{(t)},\yt^{(t)})$, which therefore tends to zero; $\gapf(x^{(t)},y^{(t)})\to0$ follows from the vanishing displacements and $|\gapf(x,y)-\gapf(\xt,\yt)|\le\norm A_{\mathrm{op}}(\norm{x-\xt}_2+\norm{y-\yt}_2)$, and $\dist((x^{(t)},y^{(t)}),\mathrm{NE})\to0$ by compactness of the simplices and continuity of the gap.  This is part~(ii) of Theorem~\ref{thm:global} for the strategies when $t_+<\infty$; the case $t_+=\infty$ is Lemma~\ref{lem:neverpos}.  The proof in Section~\ref{sec:global:proofs} yields the statement for the regret state, $\Omega\subseteq\Fix(\Phi)\cap\{\norm{\rt_i}_2=\Nt_{i,\infty}\}$, which Corollary~\ref{cor:capture} and Theorem~\ref{thm:rate} use.
\end{revblock}

\subsection{Drift when the equilibrium is not unique}
\label{app:global-drift}

\begin{remark}[No drift when the NE is not unique]
\label{rem:drift}
In the presence of multiple equilibria, the fixed-point set passing through a limit point is positive-dimensional. In this case, a vanishing gap combined with square-summable displacements does not, by itself, preclude slow drift of the state along the directions of the NE polytope. Lemma~\ref{lem:potential} rules out this possibility. Specifically, Theorem~\ref{thm:point} establishes that $\Omega$ reduces to a singleton for every game satisfying $t_+<\infty$. Our spectral rate theorem still relies on the hypothesis of uniqueness. Indeed, Corollary~\ref{cor:capture} requires a point in $\Omega$ at which Theorem~\ref{thm:rate} applies, and by virtue of Lemma~\ref{lem:kernel}, a strictly complementary non-degenerate equilibrium is necessarily unique. Consequently, in a game admitting multiple equilibria, every equilibrium is either degenerate or violates strict complementarity, and thus Theorem~\ref{thm:rate} does not apply at the limit $\bar z$. By Proposition~\ref{prop:fixedpoint}, the fixed-point set passing through $\bar z$ coincides with the product of the two cones over the optimal-strategy polytopes, restricted to positive scales. This set exhibits positive-dimensional smooth strata together with corners at boundary equilibria. Establishing a convergence rate at such a limit naturally calls for a normally hyperbolic invariant manifold argument along a smooth stratum, as discussed in Appendix~\ref{app:scope-degenerate}.
\end{remark}

\subsection{Proof of Corollary~\ref{cor:capture}}
\label{app:global-capture}

Theorem~\ref{thm:rate} supplies the following capture property: if
$z^\dagger\in\Fix(\Phi)$ corresponds to the strictly complementary
equilibrium of Assumption~\ref{asm:unique-strict} with
$\eta(z^\dagger)\sigmax(\Bmat)<1$ (so that
$\rho^*:=\max\{|\lambda|:\lambda\in\spec(D\Phi_S(z^\dagger))
\setminus\{1\}\}<1$ by (E7)), then there exist an open neighborhood
$U'\ni z^\dagger$, a constant $C>0$, and $\bar\rho\in(\rho^*,1)$ such
that: whenever $z^{(T)}\in U'$ for some $T\ge t_+$, there is a fixed point
$z_\infty\in\Fix(\Phi)$ with $\norm{z^{(t)}-z_\infty}\le
C\bar\rho^{\,t-T}$ for all $t\ge T$ (\emph{entry implies capture}).

\begin{proof}[Proof of Corollary~\ref{cor:capture}]
Since $z^\dagger\in\Omega$, the very definition of an $\omega$-limit
point provides a time $T^*\ge t_+$ with $z^{(T^*)}\in U'$: every
neighborhood of $z^\dagger$ is visited infinitely often.  Capture then gives
$z^{(t)}\to z_\infty\in\Fix(\Phi)$ at the linear rate $\bar\rho$.

For the gap decay: the map $z\mapsto(X(z),Y(z))$ is locally Lipschitz on
a neighborhood of $z_\infty$, by Lemma~\ref{lem:absorb} (the map $\Phi$
and its intermediate maps are smooth on the frozen-support neighborhood)\rev{;} and $\gapf$ is globally Lipschitz in $(x,y)$ (with constant
of order $2\norm A$, up to norm-equivalence factors) and vanishes at
$(X,Y)(z_\infty)\in\mathrm{NE}$. Hence
$\gapf(x^{(t)},y^{(t)})\le L\,\norm{z^{(t)}-z_\infty}
\le C'\bar\rho^{\,t-T^*}$.

Since the captured tail converges and $z^\dagger$ is a subsequential limit, it follows that the limit yielded by capture coincides with $z^\dagger$. Under
Assumptions~\ref{asm:unique-strict} and~\ref{asm:stepsize} the hypothesis
holds automatically: Theorem~\ref{thm:point} makes $\Omega$ the single
point $\bar z$, and $\bar z$ satisfies the hypotheses of
Theorem~\ref{thm:rate} with $m\ge2$ precisely when
$\eta(\bar z)\sigmax(\Bmat)<1$, i.e.\
$\sigmax(\Bmat)<\sqrt{N_{1,\infty}N_{2,\infty}/(\norm{x^*}_2\norm{y^*}_2)}$,
which is Assumption~\ref{asm:stepsize} evaluated at the saturation level;
for $m=1$ the capture is finite-step by Lemma~\ref{lem:absorb}.
\end{proof}

\subsection{Proofs for self-stabilization
(Proposition~\ref{prop:selfstab})}
\label{app:global-selfstab}

Throughout this subsection $z^\dagger=(c_1x^*,c_2y^*)$ is a fixed point
whose equilibrium $(x^*,y^*)$ is strictly complementary (with gap
$\delta>0$ in the notation of Lemma~\ref{lem:absorb}) and nondegenerate,
i.e.\ satisfies Assumption~\ref{asm:nondeg-strict} (by
Lemma~\ref{lem:kernel} it is then the unique equilibrium, but the
arguments below use only the face-local properties);
$S=S_1\times S_2$ is the support pair and
$\Ms=\{z:z_{S^c}=0\}$ the support subspace. We write
$\eta=\eta(z^\dagger)=1/\sqrt{c_1c_2}$ and
$q_k=\eta^2\sigma_k^2$ for the singular values $\sigma_k$ of $\Bmat$.

\begin{lemma}[Analytic ambient extension near the ray; strengthening of
Lemma~\ref{lem:absorb}]
\label{lem:analytic}
The point $z^\dagger$ lies on the boundary of the nonnegative state space
$\mathcal Z$ whenever some coordinate is off-support, so no
$\R^{n_1+n_2}$-open neighborhood of $z^\dagger$ is contained in
$\mathcal Z$.  There exist, however, an open neighborhood
$U\subset\R^{n_1+n_2}$ of $z^\dagger$ and a real-analytic map
$\hat\Phi:U\to\R^{n_1+n_2}$ with $\hat\Phi=\Phi$ on $U\cap\mathcal Z$
(we keep writing $\Phi$ for $\hat\Phi$ below), such that:
\begin{enumerate}[label=(\alph*),itemsep=1pt,topsep=2pt]
\item for $i\notin S_1$: $\Phi(z)_{1,i}\equiv0$ on $U$ (off-support
components vanish exactly in one step), and symmetrically for player $2$;
\item $D\Phi(z^\dagger)$ is block-triangular,
$\begin{pmatrix} J_S & * \\ 0 & 0\end{pmatrix}$ in the $S$/$S^c$
coordinate split, so
$\spec(D\Phi(z^\dagger))$ consists of $\spec(J_S)$ together with the eigenvalue $0$ of multiplicity $n_1+n_2-|S_1|-|S_2|$, the latter being absent when both players are
fully mixed ($S^c=\emptyset$, so $D\Phi(z^\dagger)=J_S$), where
$J_S=D\Phi_S(z^\dagger)=I+L+L^2$ as in Lemma~\ref{lem:jacobian};
\item $J_S$ is \emph{nonsingular}: its eigenvalues are
$\mu=1+\lambda+\lambda^2$ with $\lambda\in\spec(L)$ purely imaginary
($\lambda=\pm i\sqrt q$, $q\ge0$, by Lemma~\ref{lem:centering}), and
$|\mu|^2=(1-q)^2+q=1-q+q^2\ge3/4>0$ for all $q\ge0$ (the minimum
$3/4$ is attained at $q=1/2$).
\end{enumerate}
\end{lemma}

\begin{proof}
\emph{(a) and the extension.} Define $\hat\Phi$ by running one round of
Algorithm~\ref{alg:iregprm} with three substitutions, each of which
agrees with the original operation on feasible states near
$z^\dagger$: (1) the $\ell_1$ normalizations $\rt_i/\norm{\rt_i}_1$ and
$p_i/\norm{p_i}_1$ are replaced by $\rt_i/(\ones^\top\rt_i)$ and
$p_i/(\ones^\top p_i)$; (2) the $\gamma$-shift is replaced by the
solution of the quadratic $\sum_{j\in S_1}(v_j-\gamma)^2=\rt_1^\top\rt_1$
with the active set frozen at $S_1$; (3) every clipping $\pos{\cdot}$ is
replaced by the coordinate projection onto $S_i$ (in-support
coordinates kept, off-support coordinates set to $0$).  Each
substituted operation is real-analytic on a neighborhood of
$z^\dagger$: $\ones^\top\rt_1\to c_1>0$ and $\rt_1^\top\rt_1\to c_1^2\norm{x^*}_2^2>0$,
\rev{the frozen-active-set quadratic $F(z,\gamma):=\sum_{j\in S_1}(v_j(z)-\gamma)^2-\rt_1^\top\rt_1$ satisfies $F(z^\dagger,v^*)=0$ and $\partial_\gamma F(z^\dagger,v^*)=-2\sum_{j\in S_1}p_{1,j}(z^\dagger)=-2c_1\ne0$, so the analytic implicit function theorem selects the unique analytic branch through $\gamma_1(z^\dagger)=v^*$ (the other root of the quadratic, $v^*+2c_1/|S_1|$ at $z^\dagger$, is discarded), and projections are linear.}  Hence $\hat\Phi$ is a
composition of analytic maps on some $\R^{n_1+n_2}$-open $U\ni z^\dagger$.
It remains to check $\hat\Phi=\Phi$ on $U\cap\mathcal Z$, shrinking $U$
if necessary.  The proof of Lemma~\ref{lem:absorb} (under
Assumption~\ref{asm:nondeg-strict}) for the strict version of
Algorithm~\ref{alg:iregprm} gives, on feasible states near $z^\dagger$:
for off-support components, $r_{1,i}(z)\le-\delta/2$ and
$|g_{1,i}(z)|<\delta/4$, hence $r_{1,i}+g_{1,i}<-\delta/4<0$, so both
$(\pos{r_1})_i$ and $(\pos{r_1+g_1})_i$ equal $0$, as does the projection
in (3); for in-support components, $r_{1,j}(z)\to c_1x^*_j>0$ and
$g\to0$, so $r_{1,j}+g_{1,j}>0$ and clipping is the identity, again as
in (3).  Because a feasible $\rt_1$ near $z^\dagger$ is nonnegative,
$\ones^\top\rt_1=\norm{\rt_1}_1$, and likewise for $p_1$, which is
nonnegative and supported on $S_1$; so (1) agrees with the original
normalizations.  Finally the active set of the original $\gamma$-shift
is exactly $S_1$ on feasible states near $z^\dagger$ (in-support:
$v_j-\gamma=r_{1,j}>0$; off-support: $<-\delta/2$), so the original
$\gamma$ solves the same frozen quadratic as (2), and
$\rt_1^\top\rt_1=N_1(z)^2$.  Thus $\hat\Phi=\Phi$ on $U\cap\mathcal Z$,
and (a) holds for $\hat\Phi$ on all of $U$ by construction.  Every
$\Phi$-orbit that stays in $U\cap\mathcal Z$ is a $\hat\Phi$-orbit, which
is the only property the manifold theorem below needs.

\emph{(b).} By (a), the $S^c$-rows of $D\Phi$ vanish identically on $U$;
$\Ms=\{z_{S^c}=0\}$ is a linear subspace, and the $S$-block is the
derivative of the restriction $\Phi_S=\Phi|_{\Ms}$, i.e.\ $J_S$.

\emph{(c).} $\mu=0$ would require $\lambda^2+\lambda+1=0$, i.e.\
$\lambda=e^{\pm2\pi i/3}$, of modulus $1$ and not purely imaginary,
contradicting $\spec(L)\subset i\R$ (Lemma~\ref{lem:centering}). For the
quantitative bound, $\lambda=i\sqrt q$ gives $\mu=1-q+i\sqrt q$ and
$|\mu|^2=(1-q)^2+q=1-q+q^2$, a quadratic in $q$ minimized at $q=1/2$
with value $3/4$.
\end{proof}

\begin{remark}[Applicability of the center-stable manifold theorem to
the non-invertible map $\Phi$]
\label{rem:csm-applicability}
Near the ray, $\det D\Phi=0$ whenever some coordinate is off-support
(Lemma~\ref{lem:analytic}(b)), so $\Phi$ is then
not invertible, whereas textbook statements of the (center-)stable
manifold theorem, including the form invoked by
\citet{lee2019first}, are
typically phrased for diffeomorphisms. The version we need is the local
center-stable manifold theorem for $C^1$
\emph{maps} at a fixed point, with no invertibility hypothesis: a $C^1$
manifold $W^{cs}_{loc}$, tangent to the center-stable eigenspace, exists
and contains \emph{every orbit whose forward iterates remain in a
sufficiently small neighborhood of the fixed point}. This map version is
supplied by the graph-transform construction
\citep{hirsch1977invariant}: both the construction of the
pseudo-stable manifold and the trapping inclusion use forward iteration
only and never invert $\Phi$. The eigenvalue $0$ of $D\Phi$ sits on the
center-stable side of the splitting, which is exactly where the graph
transform places contracted directions.
\end{remark}

\begin{proposition}[Unstable ray points have thin basins;
$\Ms$-local version]
\label{prop:cs-basin}
Let $z^\dagger$ be as above with
$q_{\max}:=\eta^2\sigmax(\Bmat)^2>1$, and set
$d_u:=2\,\#\{k:\eta\sigma_k>1\}\ge2$ (each violating singular value
contributes a conjugate eigenvalue pair with
$|\mu|^2=1-q+q^2>1$ when $q>1$). Then:
\begin{enumerate}[label=(\arabic*),itemsep=2pt,topsep=2pt]
\item \emph{(Full space: convergent orbits lie on a center-stable
manifold.)} There exist a neighborhood $B_\dagger\ni z^\dagger$ and a
$C^1$ local manifold $W^{cs}_{loc}(z^\dagger)$ with
$\dim W^{cs}_{loc}=d-d_u$ such that
\[
\{z\in B_\dagger\cap\mathcal Z:\ \Phi^t(z)\in B_\dagger\ \forall t\ge0\}
\;\subseteq\;W^{cs}_{loc}(z^\dagger)
\]
\rev{(the trapping statement holds for the ambient extension $\hat\Phi$ of Lemma~\ref{lem:analytic} on all of $B_\dagger$; the displayed version is its consequence for feasible states).}
In particular the tail of every trajectory converging to $z^\dagger$ is
contained in $W^{cs}_{loc}(z^\dagger)$.
\item \emph{(Inside $\Ms$: the basin has measure zero.)} $\Ms$ is
$\Phi$-invariant near $z^\dagger$ (Lemma~\ref{lem:analytic}(a)), $J_S$ is
nonsingular (Lemma~\ref{lem:analytic}(c)), and $z\mapsto D\Phi_S(z)$ is
continuous, so there is an $\Ms$-open neighborhood $V\ni z^\dagger$ on
which $\Phi_S|_V$ is a local diffeomorphism. \begin{revblock}Two facts are needed.  \emph{Transversality.}  In (support, off-support)
coordinates $D\hat\Phi(z^\dagger)=\bigl(\begin{smallmatrix}J_S&H\\0&0\end{smallmatrix}\bigr)$
with $J_S$ invertible (Lemma~\ref{lem:analytic}); the linear change of
coordinates $T(u,v)=(u-J_S^{-1}Hv,v)$ conjugates it to
$\operatorname{diag}(J_S,0)$, so the center-stable subspace of
$D\hat\Phi(z^\dagger)$ projects onto the off-support coordinates and
$T_{z^\dagger}W^{cs}_{loc}(z^\dagger)+\Ms=\R^{d}$.  After shrinking the
neighborhood, the transverse-intersection theorem makes
$W^{cs}_{loc}(z^\dagger)\cap\Ms$ a $C^1$ manifold of dimension
$\dim\Ms-d_u<\dim\Ms$, Lebesgue-null in $\Ms$.  \emph{Restricted
preimages.}  For $E\subseteq\Ms$ let
$P_T(E):=\{z\in V:\ \Phi_S^{j}(z)\in V\ (0\le j<T),\ \Phi_S^{T}(z)\in E\}$;
each $P_T(E)$ is the preimage of $E$ under a finite composition of local
diffeomorphisms on open domains, hence null whenever $E$ is null.  The
local basin
\[
\bigl\{z\in V:\ \Phi_S^t(z)\in V\ \forall t,\
\Phi_S^t(z)\to z^\dagger\bigr\}
\;\subseteq\;
\bigcup_{T\ge0}P_T\bigl(W^{cs}_{loc}(z^\dagger)\cap\Ms\bigr)
\]
is therefore contained in a countable union of null sets and has Lebesgue
measure zero in $\Ms$.\end{revblock} The same holds at every individual point of the
step-size-violating ray set
$R_{>1}:=\{(c_1x^*,c_2y^*):c_1,c_2>0,\ \eta(c_1,c_2)\,\sigmax(\Bmat)>1\}
\subset\Ms$. Moreover, for every \emph{compact} subsegment
$R_K\subset R_{>1}$ there are an $\Ms$-open neighborhood
$V_K\supset R_K$ and a Lebesgue-null set $E_K\subset\Ms$ such that any
$\Phi_S$-orbit that remains in $V_K$ forever and converges to a point of
$R_K$ has its tail contained in $E_K$; in particular the local basin
\[
\bigl\{z\in V_K:\ \Phi_S^t(z)\in V_K\ \forall t\ge0,\
\lim_t\Phi_S^t(z)\in R_K\bigr\}
\]
is Lebesgue-null in $\Ms$. \rev{No common quantitative neighborhood radius is asserted along the noncompact set $R_{>1}$ ($c_i\to0$, $c_i\to\infty$, or $\eta\,\sigmax(\Bmat)\to1$); a countable-cover version of the same argument gives the qualitative null-basin statement for convergence to any point of $R_{>1}$ (end of the proof), but not for convergence to the boundary $\eta\sigmax(\Bmat)=1$, for loss of feasibility, or for unbounded trajectories.}  It does \emph{not} bound the set of all $\Ms$-states whose
orbits eventually converge to the segment.  Pulling the null set
back along the portion of an orbit prior to its entry into $V_K$ would
require $\Phi_S$ to be a local diffeomorphism away from the ray, and that
is not established (the clipping pattern is not frozen there).  This is
the obstruction of Remark~\ref{rem:fullspace}, acting inside $\Ms$.
As discussed after Proposition~\ref{prop:selfstab}, the fact that a
convergent trajectory's tail lies in the null set is true by
construction and has no excluding force for that trajectory.
\end{enumerate}
\end{proposition}

\begin{proof}
\emph{Part 1.} By Lemma~\ref{lem:analytic}, the ambient extension
$\hat\Phi$ (written $\Phi$) is $C^\omega\subset
C^1$ on an $\R^{n_1+n_2}$-open neighborhood of $z^\dagger$, agrees with
the algorithm on feasible states there, and
$\spec(D\Phi(z^\dagger))$ consists of $\spec(J_S)$ together with the eigenvalue $0$ of multiplicity $n_1+n_2-|S_1|-|S_2|$ (absent when $S^c=\emptyset$). Split the spectrum into
the center-stable part $\{|\lambda|\le1\}$ (which contains the eigenvalue
$0$ of the $S^c$ directions, if any, and all eigenvalue pairs with $q_k\le1$) and
the unstable part $\{|\lambda|>1\}$ (the $d_u$ eigenvalues coming from
$q_k>1$, since $|\mu|^2=1-q+q^2>1\iff q>1$). The local
center-stable manifold theorem for $C^1$ maps, in the
non-invertible form discussed in Remark~\ref{rem:csm-applicability}
\citep{hirsch1977invariant}, in
the application form used for saddle avoidance of first-order methods
\citep[Thm.~2]{lee2019first}
\citep[cf.][]{daskalakis2018limit}, yields $B_\dagger$ and
$W^{cs}_{loc}(z^\dagger)$ of dimension $d-d_u$ containing all orbits that
remain in $B_\dagger$ forever. A trajectory converging to $z^\dagger$
eventually remains in $B_\dagger$, so its tail lies in
$W^{cs}_{loc}(z^\dagger)$. The dimension count: eigenvalues of $J_S$ with
$|\mu|>1$ arise exactly from the conjugate pairs with $q_k>1$ by (E7),
and the $S^c$ directions (eigenvalue $0$) count as stable; thus
$\dim W^{cs}\le d-2<d$.

\emph{Part 2, single point.} As stated: invariance of $\Ms$ and one-step
exact
vanishing of off-support coordinates are Lemma~\ref{lem:analytic}(a);
nonsingularity of $J_S$ plus continuity of $z\mapsto D\Phi_S(z)$ give,
via the inverse function theorem, an $\Ms$-neighborhood $V$ on which
$\Phi_S$ is a local diffeomorphism. Any point of the local basin reaches
$W^{cs}_{loc}\cap\Ms$ (a $C^1$ patch of dimension
$\dim\Ms-d_u<\dim\Ms$, hence Lebesgue-null in $\Ms$) after some finite
number $T$ of steps while remaining in $V$; \rev{therefore the basin is
covered by $\bigcup_{T\ge0}P_T(W^{cs}_{loc}\cap\Ms)$ with the restricted
preimages $P_T$ defined in the proof of Part~2 below (unrestricted
preimages $\Phi_S^{-T}$ would not do: an orbit segment leaving $V$ may
pass through a clipping region where $\Phi$ is not a local
diffeomorphism), each term is null, and a countable union of null sets
is null.}

\emph{Part 2, version on a compact subsegment $R_K$.} \rev{Two points
need care beyond the single-point statement: (i) the null set must be
built from countably many pointwise constructions, with no uniform lower
bound on their radii being needed; (ii) an orbit may converge to a limit
$z'\in R_K$ that is not a net point, and one must say why its tail still
lands in the null set attached to a \emph{nearby} net point.}

\begin{revblock}(i) \emph{Pointwise data.}  For each $w\in R_K$, Part~1 applied at $w$
together with the transversality step above gives an $\Ms$-open ball
$B_{\Ms}(w,r_w)$ and a $C^1$ patch $W_w:=W^{cs}_{loc}(w)\cap\Ms$ of
dimension $\le\dim\Ms-2$ such that every $\Phi_S$-orbit that remains in
$B_{\Ms}(w,r_w)$ forever lies in $W_w$; the inclusion requires only that
the orbit stay in the ball, not that it converge to $w$.  By compactness
finitely many balls $B_{\Ms}(w_j,r_{w_j}/2)$, $j\le J$, cover $R_K$.
Since $|\det D\Phi_S|\ge(3/4)^{\dim\Ms/2}>0$ on $R_K$
(Lemma~\ref{lem:analytic}(c)) and $\det D\Phi_S$ is continuous, there is
an $\Ms$-open $V_K\supset R_K$ on which $\Phi_S$ is a local
diffeomorphism.  No uniform manifold radius along $R_K$ is needed, so no
parameter-uniform graph-transform estimate enters.

(ii) \emph{Convergence to a non-net point.}  Suppose a $\Phi_S$-orbit
$\{\Phi_S^t(z)\}$ remains in $V_K$ and converges to some $z'\in R_K$, and
pick $j$ with $\norm{z'-w_j}<r_{w_j}/2$.  From some finite $T$ on the
orbit stays in $B_{\Ms}(z',r_{w_j}/2)\subseteq B_{\Ms}(w_j,r_{w_j})$,
hence $\Phi_S^t(z)\in W_{w_j}$ for all $t\ge T$.  With
$P^{V_K}_T(E):=\{z\in V_K:\ \Phi_S^{i}(z)\in V_K\ (0\le i<T),\
\Phi_S^{T}(z)\in E\}$ define
\[
E_K\;:=\;\bigcup_{j\le J}\ \bigcup_{T\ge0}P^{V_K}_T\bigl(W_{w_j}\bigr).
\]
The starting point $z$ lies in $E_K$, since the orbit up to time $T$ stays
in $V_K$, where each $P^{V_K}_T$ is a preimage under a finite composition
of local diffeomorphisms and preserves null sets.  $E_K$ is a countable
union of null sets, hence Lebesgue-null in $\Ms$, and the local basin of
$R_K$ is contained in $E_K$.  The restriction to compact $R_K$ enters only
through the finite cover, and is not essential: by second countability a
countable subcover of $R_{>1}$ by the balls $B_{\Ms}(w,r_w/2)$ exists, and
the same argument with a countable union gives an $\Ms$-open
$V_{>1}\supset R_{>1}$ and a null $E_{>1}$ for orbits that remain in
$V_{>1}$ and converge to a point of $R_{>1}$; what degenerates at the
noncompact ends (as $c_i\to0$ or $\infty$, or as $\eta\sigmax\to1$) is
only the size of the neighborhoods, not the conclusion.\end{revblock}
\end{proof}

\begin{proof}[Proof of Proposition~\ref{prop:selfstab}]
Parts (a) and (b) are Proposition~\ref{prop:cs-basin} parts 1 and 2,
applied at $z^\dagger$ (single-point statement) and at a compact
subsegment through $z^\dagger$ (uniform statement).
\end{proof}

\begin{remark}[Trapped tails lie in the null set by construction]
\label{rem:trapped-tail}
Via Lemma~\ref{lem:absorb}, a trajectory is absorbed into $\Ms$ in one step
near a ray, so the tail of a trajectory that converges to a
step-size-violating point $z_\infty\in\Fix(\Phi)$ remains in the
neighborhood $V$ of Proposition~\ref{prop:cs-basin}(2), hence in the
exceptional Lebesgue-null set.  The measure-zero statement therefore
excludes no individual orbit; its force is the genericity heuristic of
Remark~\ref{rem:init}.
\end{remark}

\begin{remark}[Open gap: \rev{full-space statement for late-entry initial states}]
\label{rem:fullspace}
Near the ray, whenever some coordinate is off-support, $\Phi$ maps the
$S^c$ coordinates to exactly $0$
(Lemma~\ref{lem:analytic}(a)), so $\det D\Phi=0$ and $\Phi$ is
\emph{not invertible}; preimages of Lebesgue-null sets under
non-invertible maps can have positive measure. The
``almost every initial point avoids strict saddles'' argument of
\citet{lee2019first} requires
$\det D\Phi\ne0$ everywhere and therefore does \emph{not} apply to our
map on the full state space. This is why
Proposition~\ref{prop:cs-basin}(2) is asserted only inside $\Ms$ (when
both players are fully mixed, $\Ms$ is the whole state space and the
statement is already full-space). A
\rev{Part~1 of Proposition~\ref{prop:cs-basin} already gives one full-space statement: the trapped local basin $\{z\in B_\dagger\cap\mathcal Z:\Phi^t(z)\in B_\dagger\ \forall t\ge0\}$ lies in $W^{cs}_{loc}$, of dimension $d-d_u<d$, and is ambient Lebesgue-null.  What remains open is the full-space statement for initial states whose orbits enter the neighborhood only later: pulling the local null set backward through earlier clipping regions is not justified, and since the push-forward of Lebesgue measure restricted to the absorbing neighborhood is supported on the face, it is singular with respect to ambient Lebesgue measure, so the argument would have to control preimages of face-null sets directly.  Near the ray this is available: the support-output block of $D\hat\Phi$ has full row rank (it is close to the invertible $J_S$), so $\hat\Phi$ viewed as a map into $\Ms$ is a local submersion and preimages of face-null sets are ambient-null by Fubini; what is open is the pullback through earlier regions of the state space, away from the ray, where neither invertibility nor a submersion structure is available.}
\end{remark}

\begin{remark}[Open gap: fixed initialization]
\label{rem:init}
The algorithm's standard initialization $\rt^{(1)}=\mathbf 0$ is a
single deterministic initial condition, and the genericity statement of
Proposition~\ref{prop:cs-basin}(2) has no direct force for it.  \rev{Its scope is the local statement: for an absolutely continuous perturbation of a state in the local face neighborhood of Proposition~\ref{prop:cs-basin}(2), the event that the perturbed orbit stays in that neighborhood and converges to the violating ray has probability zero, because the exceptional set is Lebesgue-null in $\Ms$ (of codimension $\ge2$).  It says nothing about perturbations supported on lower-dimensional sets, nor about states whose orbits reach the neighborhood only after an excursion outside it.}
Any statement about the $\rt^{(1)}=\mathbf 0$ trajectory itself remains
heuristic in this sense (Remark~\ref{rem:trapped-tail}).
\end{remark}

\begin{remark}[Open gap: the boundary $q=1$ is non-hyperbolic]
\label{rem:boundary}
When $\eta\sigmax(\Bmat)=1$, i.e.\ $q_{\max}=1$, the corresponding
eigenvalues are $\mu=\pm i$ (fourth roots of unity): the fixed point is
non-hyperbolic, higher-order terms on the center manifold decide between
weak stability, weak rotation, and weak instability, and the
linearization is inconclusive. Proposition~\ref{prop:cs-basin} does not
exclude convergence to such boundary points at a sublinear rate. We have
not observed this in the numerical suite (the spectral data allow
checking the distribution of $q_{\max}$ against $1$; see
Figure~\ref{fig:etasigma}), but it remains open.
\end{remark}


\section{Proofs for the Ratio Certificate}\label{app:certificate}
\subsection{Proofs}
\label{sec:cert:proofs}

This subsection proves the increment identity, the two-sided increment
bound and the certificate itself.  It opens with the precise form of the
frozen local regime, described informally in Section~\ref{sec:certificate}: the
list \ref{p:square}--\ref{p:increment} below fixes the constants and the
coordinates that the rest of the section uses.  The gap--distance
sharpness of Lemma~\ref{lem:gap-sharp} rests on a separate first-order
analysis of the gap map, and its proof, together with the Lipschitz upper
bound it builds on, is in Appendix~\ref{app:certificate}.

\subsubsection{Standing conventions: the frozen local regime}
\label{app:cert-regime}

Throughout this subsection and Appendix~\ref{app:certificate},
Assumptions~\ref{asm:unique-strict}
and~\ref{asm:stepsize} are in force and we work in the frozen local
regime constructed in the proof of Theorem~\ref{thm:rate}. We collect
the facts imported from there; none is reproved here.

\begin{enumerate}[label=(P\arabic*),leftmargin=3em]
\item\label{p:square}
Under Assumption~\ref{asm:unique-strict} the equilibrium support is
square, $|S_1|=|S_2|=:m$, and $m\ge2$; the pure saddle $m=1$ is set aside,
since there $\Bmat=0$, the transverse space $W$ introduced in \ref{p:decomp}
below is trivial, and convergence is exact in finitely many steps, so nothing
below is needed. The value-centered support matrix
$\Bmat=\Atil-v^*\mathbf{1}\mathbf{1}^\top$ of
Lemma~\ref{lem:centering} satisfies $\Bmat\psi=0$, $\Bmat^\top\xi=0$,
with $\ker\Bmat=\mathrm{span}\{\psi\}$ and
$\ker\Bmat^\top=\mathrm{span}\{\xi\}$; its nonzero singular values are
$\sigmax=\sigma_1\ge\cdots\ge\sigma_{m-1}=\sigmin^{+}>0$, with
singular vectors $\Bmat v_k=\sigma_k u_k$, $\Bmat^\top u_k=\sigma_k
v_k$.
\item\label{p:decomp}
In the support coordinates of $\Ms$, the state space decomposes
orthogonally as $E\oplus W$, where
$E=\mathrm{span}\{(\xi,0),(0,\psi)\}$ contains the fixed-point cone
$E_+$ (Proposition~\ref{prop:fixedpoint}) and
$W=\bigoplus_k\mathrm{span}\{(u_k,0),(0,v_k)\}$. Orthogonality follows
from $\Bmat^\top\xi=0$ and $\Bmat\psi=0$. We write $P_E,P_W$ for the
orthogonal projections and $w=(w_1,w_2)\in W$, so that
$w_1\in\mathrm{range}(\Bmat)$, $w_2\in\mathrm{range}(\Bmat^\top)$ and
$\norm{w}^2=\norm{w_1}^2+\norm{w_2}^2$.
\item\label{p:regime}
There are radii $r_1,\delta_0>0$, fixed in the proof of
Theorem~\ref{thm:rate}, such that the regime is the set of states
$z=z_f+w$ with $z_f=(a\xi,b\psi)\in E_+$,
$\norm{z_f-z^*}\le r_1$ (hence $a\in[c_1/2,3c_1/2]$,
$b\in[c_2/2,3c_2/2]$), and $\norm{w}\le\delta_0$. On the regime the
effective step size $\eta(a,b)=1/\sqrt{ab}$ ranges in
$[\eta_-,\eta_+]$ with $0<\eta_-\le\eta_+<\infty$.
\item\label{p:frozen}
On the regime the clipping pattern is frozen
(Lemma~\ref{lem:absorb}): for each player, on-support coordinates of
both $r^{(t)}$ and $r^{(t)}+g^{(t)}$ are strictly positive (no
clipping at either the prediction or the update stage), while every
off-support coordinate $j\notin S_i$ satisfies $r^{(t)}_{i,j}\le
-\delta/2$ and $|g^{(t)}_{i,j}|\le\delta/4$, hence
$r^{(t)}_{i,j}+g^{(t)}_{i,j}\le-\delta/4<0$ (clipped exactly to $0$),
where $\delta>0$ is the strict-complementarity margin of
Assumption~\ref{asm:unique-strict}. Consequently the one-round map
restricted to $\Ms$ is a composition of analytic maps on the regime.
\item\label{p:compact}
By analyticity and compactness of the closed regime, the maps
$z\mapsto(x(z),y(z))$ (main-iterate strategies) and
$z\mapsto g_i(z)$ are $C^2$ with uniform first- and second-derivative
bounds; the second-order constants below ($C_q$, $K_g$, $C_3$) are
finite and uniform on the regime. They are \emph{not} exhibited in
closed form. This affects only the size of the thresholds
$\delta_2,\delta_3$, and not the form of any constant that enters
$\kappa$.
\item\label{p:increment}
The exact increment decomposition of \citet{zhang2025scale} reads: for each player $i$,
\[
  \Ndot_i(t)=\bigl\lVert g^{(t)}_i\bigr\rVert^{2}-D_i(t),
  \qquad
  D_i(t)=\bigl\lVert \pos{r^{(t)}_i}+g^{(t)}_i\bigr\rVert^{2}
        -\bigl\lVert \pos{r^{(t)}_i+g^{(t)}_i}\bigr\rVert^{2}\;\ge\;0 .
\]
\end{enumerate}

Since $E$ is a linear subspace, $P_E$ is the (orthogonal) nearest-point
map, and on the regime $P_Ez$ lies in the \emph{interior} of the cone
$E_+$ by \ref{p:regime}; hence
\begin{equation}\label{eq:dist-w}
  \dist(z,E)=\dist(z,\Fix)=\norm{w}.
\end{equation}

\subsubsection{Proof of Lemma~\ref{lem:ndot-identity}}

\begin{lemma}[Frozen-regime increment identity]
\label{lem:ndot-identity}
In the frozen local regime, for each player $i$ and every step $t$,
\[
\Ndot_i(t)\;=\;\sum_{j\in S_i}\bigl(g^{(t)}_{i,j}\bigr)^{2}
\;=\;\bigl\lVert g^{(t)}_i\big|_{S_i}\bigr\rVert^{2},
\]
where $g^{(t)}_i$ is player $i$'s update vector; the increment equals the
on-support update energy, with no error term.
\end{lemma}

\begin{proof}[Proof of Lemma~\ref{lem:ndot-identity}]
Fix a player $i$ and drop the index. Start from the decomposition
\ref{p:increment}. We compute the slack $D(t)$ componentwise under the
frozen pattern \ref{p:frozen}.

\emph{On-support coordinates $j\in S_i$.} Here $r_j>0$ (no clipping in
the prediction step) and $r_j+g_j>0$, so
$(\pos{r}+g)_j=r_j+g_j=(\pos{r+g})_j$: the two squares agree and the
contribution to $D(t)$ cancels exactly.

\emph{Off-support coordinates $j\notin S_i$.} Here $r_j\le-\delta/2<0$
so $(\pos{r})_j=0$ and $(\pos{r}+g)_j=g_j$; and
$r_j+g_j\le-\delta/4<0$ so $(\pos{r+g})_j=0$. The contribution to
$D(t)$ is therefore $g_j^2-0=g_j^2$.

Summing,
$D(t)=\sum_{j\notin S_i}g_j^{2}$, and hence
\[
  \Ndot_i(t)=\norm{g}^{2}-D(t)
  =\sum_{j\in S_i}g_j^{2},
\]
which is the claim.
\end{proof}

\begin{remark}[Numerical verification]\label{rem:n1-verify}
The identity is exact, and this exactness is a hard prediction of the
frozen-pattern analysis.  Our numerical suite checks only its visible
consequence: the exact-zero pattern of $\rt$ freezes in finite time on
$99.5\%$ of instances (Section~\ref{sec:experiments}); the full clipping
pattern behind the identity was not logged, so the on-support energy is
used as an empirical proxy for $\Ndottot$ (Appendix~\ref{app:experiments}).
The pattern-frozen identity is what distinguishes this lemma from the anchor
decomposition \ref{p:increment}, whose slack term is in general
strictly between $0$ and $\norm{g}^2$.
\end{remark}

\subsubsection{First-order expansion of the update vector}

\begin{lemma}[Update expansion]\label{lem:g-expansion}
In the frozen regime, for $z=z_f+w$ with $z_f=(a\xi,b\psi)$ and
$w=(w_1,w_2)\in W$,
\[
  g_1\big|_{S_1}=-\eta^{2}\,\Bmat\Bmat^\top w_1
  +O\bigl(K_g\norm{w}^{2}\bigr),
  \qquad
  g_2\big|_{S_2}=-\eta^{2}\,\Bmat^\top\Bmat\,w_2
  +O\bigl(K_g\norm{w}^{2}\bigr),
\]
with $\eta=\eta(a,b)$ and $K_g$ the uniform constant of
\ref{p:compact}.
\end{lemma}

\begin{proof}
The map $z\mapsto g_1(z)$ is analytic on the regime \ref{p:frozen} and
vanishes identically on $E_+$ (at fixed points $y=\yt$, so the
residual and hence $g$ are zero). Therefore
$g_1(z_f+w)=Dg_1(z_f)[w]+O(K_g\norm{w}^2)$ with a uniform constant by
\ref{p:compact}, and it remains to compute the differential in the
direction $w$. We use the cancellation identities established for
Lemma~\ref{lem:jacobian} (all evaluated at the fixed point $z_f$; the
symbol $m_i$ below denotes player $i$'s payoff vector, not the support
size):
\begin{enumerate}[label=(\roman*),leftmargin=2.5em]
\item The residual is $\mathrm{res}_1=A(y-\yt)$; since both $y$ and
$\yt$ are supported in $S_2$ on the regime,
$\mathrm{res}_1|_{S_1}=\Atil\,(y-\yt)|_{S_2}$.
\item Perturbing $(\rt_1,\rt_2)$ by $(w_1,w_2)$: the pre-iterate
differential is $d\yt=P_\psi w_2/b$ (normalization at level
$\norm{\rt_2^*}_1=b$, with $P_\psi=I-\psi\mathbf{1}^\top$). For the
main iterate $r_2=\rt_2+m_2-\gamma_2\mathbf{1}$ we have
$dm_2=-\Atil^\top P_\xi\,w_1/a=-\Bmat^\top w_1/a$
(Lemma~\ref{lem:centering}) and $d\gamma_2=0$
(cancellation, Lemma~\ref{lem:jacobian}); hence
$dy=P_\psi\bigl(w_2-\Bmat^\top w_1/a\bigr)/b$ and
\[
  d(y-\yt)=-\tfrac{1}{ab}\,P_\psi\Bmat^\top w_1 .
\]
\item Therefore the first-order term of $\mathrm{res}_1|_{S_1}$ is
$-\tfrac{1}{ab}\Atil P_\psi\Bmat^\top w_1
=-\eta^{2}\Bmat\Bmat^\top w_1$, using $\Atil P_\psi=\Bmat$
(Lemma~\ref{lem:centering}).
\item The centering term: to first order
$\inner{\mathrm{res}_1}{x}=\xi^\top(-\eta^2\Bmat\Bmat^\top w_1)
=-\eta^2(\Bmat^\top\xi)^\top\Bmat^\top w_1=0$ by \ref{p:square}, and
$x$ is exactly supported in $S_1$ on the regime; hence
$g_1|_{S_1}=\mathrm{res}_1|_{S_1}
-\inner{\mathrm{res}_1}{x}\mathbf{1}$ has the same first-order term
as $\mathrm{res}_1|_{S_1}$.
\end{enumerate}
Player $2$ is symmetric.
\end{proof}

\subsubsection{Proof of Lemma~\ref{lem:ndot-theta}}

\begin{lemma}[Two-sided increment bound]
\label{lem:ndot-theta}
There exist $0<n_-\le n_+$ and a threshold $\delta_3>0$ such that, in the
frozen regime with $\norm{w^{(t)}}\le\delta_3$, for every single step $t$,
\[
n_-\,\norm{w^{(t)}}^{2}\;\le\;\Ndottot(t)\;\le\;n_+\,\norm{w^{(t)}}^{2},
\qquad
n_-=\tfrac12\,\eta_-^{4}(\sigmin^{+})^{4},\quad
n_+=2\,\eta_+^{4}\,\sigmax^{4}.
\]
\end{lemma}

\begin{proof}[Proof of Lemma~\ref{lem:ndot-theta}]
By Lemma~\ref{lem:ndot-identity} and Lemma~\ref{lem:g-expansion},
\[
  \Ndottot(t)=
  \bigl\lVert g_1|_{S_1}\bigr\rVert^{2}
  +\bigl\lVert g_2|_{S_2}\bigr\rVert^{2}
  =\eta^{4}\Bigl(\bigl\lVert\Bmat\Bmat^\top w_1\bigr\rVert^{2}
  +\bigl\lVert\Bmat^\top\Bmat w_2\bigr\rVert^{2}\Bigr)
  +O\bigl(C_3\norm{w}^{3}\bigr),
\]
with $C_3$ uniform on the regime (cross terms of the expansion;
\ref{p:compact}). The key structural point is that $w\in W$ by
definition of the decomposition \ref{p:decomp}, so
$w_1\in\mathrm{range}(\Bmat)=\mathrm{span}\{u_k\}$ and
$w_2\in\mathrm{range}(\Bmat^\top)=\mathrm{span}\{v_k\}$. Writing
$w_1=\sum_k\alpha_ku_k$,
\[
  \bigl\lVert\Bmat\Bmat^\top w_1\bigr\rVert^{2}
  =\sum_k\sigma_k^{4}\alpha_k^{2}
  \;\in\;
  \Bigl[\bigl(\sigmin^{+}\bigr)^{4}\norm{w_1}^{2},\;
        \sigmax^{4}\norm{w_1}^{2}\Bigr],
\]
and symmetrically for $w_2$; adding and using
$\norm{w}^2=\norm{w_1}^2+\norm{w_2}^2$ gives
\[
  \eta_-^{4}\bigl(\sigmin^{+}\bigr)^{4}\norm{w}^{2}
  -C_3\norm{w}^{3}
  \;\le\;\Ndottot(t)\;\le\;
  \eta_+^{4}\,\sigmax^{4}\norm{w}^{2}
  +C_3\norm{w}^{3}.
\]
For $\norm{w}\le\delta_3:=
\eta_-^{4}(\sigmin^{+})^{4}/(2C_3)$ the cubic error is absorbed: the
lower bound loses at most a factor $\tfrac12$, and since
$\delta_3\le\eta_+^{4}\sigmax^{4}/C_3$, the upper bound gains at most
a factor $2$. This yields the stated
$n_-=\tfrac12\eta_-^{4}(\sigmin^{+})^{4}$ and
$n_+=2\eta_+^{4}\sigmax^{4}$.
\end{proof}

\subsubsection{Proof of Theorem~\ref{thm:certificate} and
Corollary~\ref{cor:theta}}

\begin{proof}[Proof of Theorem~\ref{thm:certificate}]
For every $t\ge t_{\mathrm{reg}}$ the hypotheses put $z^{(t)}$ in the
regime with
$\norm{w^{(t)}}\le\min(\delta_2,\delta_3)$, so
Lemmas~\ref{lem:gap-sharp} and~\ref{lem:ndot-theta} give
\[
  c_-\norm{w^{(t)}}\le\gapf(t)\le L_{\gapf}\norm{w^{(t)}},
  \qquad
  n_-\norm{w^{(t)}}^{2}\le\Ndottot(t)\le n_+\norm{w^{(t)}}^{2}.
\]
The hypothesis $\Ndottot(t_1)>0$ forces $w^{(t_1)}\neq0$ through the
upper increment bound, hence $\gapf(t_1)\ge c_-\norm{w^{(t_1)}}>0$,
and every ratio below is well defined. Dividing the gap bounds at
$t_2$ and $t_1$ and eliminating the amplitudes via the increment
bounds,
\[
  \frac{\gapf(t_2)}{\gapf(t_1)}
  \;\le\;
  \frac{L_{\gapf}\norm{w^{(t_2)}}}{c_-\norm{w^{(t_1)}}}
  \;\le\;
  \frac{L_{\gapf}}{c_-}
  \sqrt{\frac{\Ndottot(t_2)/n_-}{\Ndottot(t_1)/n_+}}
  \;=\;
  \kappa\,\sqrt{\frac{\Ndottot(t_2)}{\Ndottot(t_1)}},
\]
and the lower bound is symmetric. All four constants
$L_{\gapf},c_-,n_\pm$ are fixed attributes of the regime
(\ref{p:regime}--\ref{p:compact}): none depends on $t_1,t_2$ or on
$\norm{w^{(t)}}$, and no factor $1/(1-\rho^*)$ enters: they are
evaluated on the fixed scale rectangle of \ref{p:regime}, and the
amplitudes through which the trajectory contracts have been cancelled.
The factorization of $\kappa$ and the bound
$\sqrt{n_+/n_-}\le2(\eta_+/\eta_-)^2(\sigmax/\sigmin^{+})^2$ are
immediate from the formulas for $n_\pm$.
\end{proof}

\begin{proof}[Proof of Corollary~\ref{cor:theta}]
From the same sandwich,
$C(t):=\log\Ndottot(t)-2\log\gapf(t)$ satisfies
$\log(n_-/L_{\gapf}^2)\le C(t)\le\log(n_+/c_-^2)$, an interval of
width $\log\bigl((n_+/n_-)(L_{\gapf}/c_-)^2\bigr)=2\log\kappa$. For
the least-squares slope over a window, write $x_t=\log\gapf(t)$,
$y_t=2x_t+C(t)$; then
$\hat\theta=\mathrm{Cov}(x,y)/\mathrm{Var}(x)
=2+\mathrm{Cov}(x,C)/\mathrm{Var}(x)$, and Cauchy--Schwarz bounds the
second term by $\mathrm{sd}(C)/\mathrm{sd}(x)$.  A variable confined to
an interval of width $2\log\kappa$ has population standard deviation at
most half the width, $\log\kappa$ (Popoviciu's inequality):
\[
  \bigl|\hat\theta-2\bigr|
  \le\frac{\mathrm{sd}(C)}{\mathrm{sd}(x)}
  \le\frac{\log\kappa}{\mathrm{sd}(\log\gapf)} .
\]
As the window spans more decades of gap decay,
$\mathrm{sd}(\log\gapf)\to\infty$ and $\hat\theta\to2$: the recorded
times are distinct integers and, in the regime, $\norm{w^{(t+s)}}_*\le\theta^s\norm{w^{(t)}}_*$
(Lemma~\ref{lem:bootstrap}) with $\gapf\asymp\norm{w}$
(Lemma~\ref{lem:gap-sharp}), so $\gapf(t+s)\le C\theta^s\gapf(t)$ from every
$t$: at most a bounded number of
records fall in any one decade, and a growing range of $\log\gapf$
forces a growing standard deviation.
\end{proof}

The frozen local regime and the conventions
\ref{p:square}--\ref{p:increment} are fixed in
Section~\ref{sec:cert:proofs}; this appendix uses them without change.
It proves the gap--distance sharpness of Lemma~\ref{lem:gap-sharp},
whose two sides rest on a Lipschitz upper bound and a first-order
analysis of the gap map, and records two remarks on the per-player
bound and on the finite-time behavior of the fitted slope.

\subsection{A Lipschitz bound for the gap}

\begin{lemma}[Gap Lipschitz bound]\label{lem:gap-lip}
There is $L_{\gapf}<\infty$ such that $\gapf$ is
$L_{\gapf}$-Lipschitz on the regime, and $\gapf(z_f)=0$ for every
$z_f\in E_+$ in the regime.
\end{lemma}

\begin{proof}
On the regime the strategy map $z\mapsto(x(z),y(z))$ is a composition
of the frozen-pattern algorithm maps \ref{p:frozen}: an affine
prediction step, the analytic solution $\gamma(\cdot)$ of the norm
constraint, the clipping with fixed pattern (affine), and the
normalizations $r\mapsto r/\norm{r}_1$. The normalization denominators
are bounded below on the regime: writing $\rt_1=a\xi+w_1$ with
$w_1\perp\xi$ (\ref{p:decomp}), $\norm{\rt_1}_1\ge\norm{\rt_1}_2\ge
a\norm{\xi}_2\ge\tfrac{c_1}{2}\norm{\xi}_2\ge c_1/(2\sqrt m)$ by
\ref{p:regime}, and the prediction preserves the $\ell_2$ norm, so
$\norm{p_1}_1\ge\norm{p_1}_2=\norm{\rt_1}_2$ obeys the same bound;
symmetrically for player~2 with $\psi$. The elementary estimate
$\norm{u/\norm{u}_1-v/\norm{v}_1}_1\le
2\norm{u-v}_1/\max(\norm{u}_1,\norm{v}_1)$ makes each normalization
Lipschitz there. Hence $z\mapsto(x(z),y(z))$ is Lipschitz on the
compact closure of the regime. The gap
$\gapf(z)=\max_i(Ay(z))_i-\min_j(x(z)^\top A)_j$ is a difference of a
max and a min of finitely many affine functions of $(x,y)$, each
$\norm{A}_{1\to\infty}$-Lipschitz; composing yields $L_{\gapf}$,
depending only on $\norm{A}$, $c_1,c_2$, and the dimensions. Every
$z_f\in E_+$ in the regime is a fixed point whose strategies are
exactly $(x^*,y^*)$ (Proposition~\ref{prop:fixedpoint}), so
$\gapf(z_f)=0$.
\end{proof}

\subsection{Proof of Lemma~\ref{lem:gap-sharp}}

The upper bound is immediate from Lemma~\ref{lem:gap-lip} and
\eqref{eq:dist-w}:
$\gapf(z)=\gapf(z)-\gapf(P_Ez)\le L_{\gapf}\norm{z-P_Ez}
=L_{\gapf}\norm{w}$. The remainder of this subsection proves the lower
bound. We need two sub-lemmas.

\begin{lemma}[Max lemma]\label{lem:max}
Let $\xi\in\mathrm{int}\,\simplex{m}$ and $u\in\R^m$ with
$\xi^\top u=0$. Then
$\max_i u_i\ge\xi_{\min}\norm{u}_\infty
\ge(\xi_{\min}/\sqrt m)\norm{u}_2$.
\end{lemma}

\begin{proof}
Let $k$ attain $\norm{u}_\infty=|u_k|$. If $u_k>0$ then $\max_iu_i\ge
\norm{u}_\infty\ge\xi_{\min}\norm{u}_\infty$. If $u_k<0$, then from
$0=\sum_i\xi_iu_i\le\xi_ku_k+(1-\xi_k)\max_iu_i$ we get
$\max_iu_i\ge\xi_k\norm{u}_\infty/(1-\xi_k)\ge
\xi_{\min}\norm{u}_\infty$. The $\ell_2$ bound is
$\norm{u}_\infty\ge\norm{u}_2/\sqrt m$.
\end{proof}

\begin{lemma}[First-order gap map and its smallest singular value]
\label{lem:tmap}
Define $T_{a,b}:W\to\R^{2m}$,
\[
  T_{a,b}\,w:=(p,q),\qquad
  p:=\tfrac1b\,\Bmat\bigl(w_2-\tfrac1a\Bmat^\top w_1\bigr),\qquad
  q:=\tfrac1a\,\Bmat^\top\bigl(w_1+\tfrac1b\Bmat w_2\bigr).
\]
Then $T_{a,b}$ is block-diagonal across the orthogonal planes
$W_k=\mathrm{span}\{(u_k,0),(0,v_k)\}$: in the coordinates
$w_1=\alpha u_k$, $w_2=\beta v_k$ one has
$(p,q)=\sigma_k\bigl((M_k(\alpha,\beta)^\top)_1u_k,\;
(M_k(\alpha,\beta)^\top)_2v_k\bigr)$ with
\[
  M_k=\begin{pmatrix}-\sigma_k/(ab)&1/b\\[2pt]1/a&\sigma_k/(ab)
  \end{pmatrix},
  \qquad
  \det M_k=-\eta^{2}\bigl(1+\eta^{2}\sigma_k^{2}\bigr),
  \quad\eta=\eta(a,b),
\]
and consequently
\begin{multline*}
  \sigmin(T_{a,b})
  \;=\;\min_k\,\sigma_k\,\sigmin(M_k)
  \;\ge\;\min_k\frac{\sigma_k\,|\det M_k|}{\norm{M_k}_F}
  \\
  \;\ge\;
  \frac{\sigmin^{+}\,\eta^{2}}
       {\sqrt{\,2\sigmax^{2}\eta^{4}+1/a^{2}+1/b^{2}\,}}
  \;=:\;c_T(a,b)\;>\;0 .
\end{multline*}
\end{lemma}

\begin{proof}
Block structure: by \ref{p:square}, $\Bmat v_k=\sigma_ku_k$ and
$\Bmat^\top u_k=\sigma_kv_k$, so for $w_1=\alpha u_k$,
$w_2=\beta v_k$,
\[
  p=\tfrac1b\bigl(\beta\sigma_ku_k-\tfrac\alpha a\sigma_k^2u_k\bigr)
   =\sigma_k\Bigl(-\tfrac{\sigma_k}{ab}\alpha+\tfrac1b\beta\Bigr)u_k,
  \qquad
  q=\sigma_k\Bigl(\tfrac1a\alpha+\tfrac{\sigma_k}{ab}\beta\Bigr)v_k,
\]
which is the stated $2\times2$ action; the images lie along the
orthonormal directions $u_k$, $v_k$, so the blocks are orthogonal to
one another. The determinant is
$-\sigma_k^2/(ab)^2-1/(ab)=-\eta^2(1+\eta^2\sigma_k^2)$. For a
$2\times2$ matrix, $\sigmin=|\det|/\sigmax\ge|\det|/\norm{\cdot}_F$;
using $|\det M_k|\ge\eta^2$ and
$\norm{M_k}_F^2=2\sigma_k^2\eta^4+1/a^2+1/b^2
\le2\sigmax^2\eta^4+1/a^2+1/b^2$ gives the bound, and
$\min_k\sigma_k=\sigmin^{+}$.
\end{proof}

On the regime, using $a\ge c_1/2$, $b\ge c_2/2$ and
$\eta\in[\eta_-,\eta_+]$ from \ref{p:regime},
\begin{equation}\label{eq:cT-uniform}
  c_T:=\inf_{(a,b)\in\text{regime}}c_T(a,b)
  \;\ge\;
  \frac{\sigmin^{+}\,\eta_-^{2}}
       {\sqrt{\,2\sigmax^{2}\eta_+^{4}+4/c_1^{2}+4/c_2^{2}\,}},
\end{equation}
which is the explicit constant announced in
Lemma~\ref{lem:gap-sharp}.

\begin{lemma}[Gap--distance sharpness]
\label{lem:gap-sharp}
There exist $0<c_-\le L_{\gapf}<\infty$ and a threshold $\delta_2>0$ such
that, in the frozen regime with $\norm{w}\le\delta_2$,
\[
c_-\,\norm{w}\;\le\;\gapf(z)\;\le\;L_{\gapf}\,\norm{w},
\]
where $L_{\gapf}$ is a Lipschitz constant of $\gapf$ in state coordinates
and $c_->0$ is explicit (built from the smallest on-support equilibrium
probabilities and the spectrum of $\Bmat$); the lower bound is where
nondegeneracy enters.
\end{lemma}

\begin{proof}[Proof of Lemma~\ref{lem:gap-sharp} (lower bound)]
Let $z=z_f+w$ be in the regime, $z_f=(a\xi,b\psi)$, and let
$(x,y)=(x(z),y(z))$ be the main-iterate strategies; by \ref{p:frozen}
their supports are contained in $S_1,S_2$. We proceed in four steps.

\emph{Step 1: exact localization to the support.} By strict
complementarity and continuity (shrinking the regime thresholds if
necessary, absorbed into $\delta_2$), every off-support row satisfies
$(Ay)_i\le v^*-\delta/2$ for $i\notin S_1$, while on the support
$\max_{i\in S_1}(\Atil y)_i\ge\xi^\top\Atil y=v^*\mathbf{1}^\top
y=v^*$, using $\xi^\top\Atil=v^*\mathbf{1}^\top$
(Lemma~\ref{lem:centering}) and that $y$ is a probability vector on
$S_2$. Hence the global max is attained on $S_1$; symmetrically the
global min is attained on $S_2$. Using $\Atil\psi=v^*\mathbf{1}$ and
$\Atil^\top\xi=v^*\mathbf{1}$,
\[
  \gapf(z)=
  \underbrace{\max_{i\in S_1}\bigl(\Atil(y-\psi)\bigr)_i}_{=:G_1\ge0}
  \;+\;
  \underbrace{\Bigl(-\min_{j\in S_2}\bigl(\Atil^\top(x-\xi)\bigr)_j
  \Bigr)}_{=:G_2\ge0}
  \qquad\text{exactly.}
\]

\emph{Step 2: sharp lower bound on the max via equilibrium
weighting.} Let $u:=\Atil(y-\psi)$ and $v:=\Atil^\top(x-\xi)$. Then
$\xi^\top u=v^*\mathbf{1}^\top(y-\psi)=0$ and
$\psi^\top v=0$ \emph{exactly}. Lemma~\ref{lem:max} (and its mirror
for the min) gives
\[
  \gapf(z)\;\ge\;
  \frac{\min(\xi_{\min},\psi_{\min})}{\sqrt m}
  \bigl(\norm{u}+\norm{v}\bigr)
  \;\ge\;
  \frac{\min(\xi_{\min},\psi_{\min})}{\sqrt m}\,\norm{(u,v)} .
\]

\emph{Step 3: first-order nondegeneracy of $(u,v)$ in $w$.} The map
$z\mapsto(u,v)$ is analytic on the regime and vanishes identically on
$E_+$ (there $(x,y)=(\xi,\psi)$). Its differential at $z_f$ in the
direction $w$ is computed from the cancellation identities of
Lemmas~\ref{lem:jacobian} and~\ref{lem:centering} ($d\gamma=0$ at
fixed points; $\Atil P_\psi=\Bmat$, $\Atil^\top P_\xi=\Bmat^\top$):
writing the differentials of the normalized iterates as in the proof
of Lemma~\ref{lem:jacobian},
\[
  u=T_{a,b}^{(1)}w+O\bigl(C_q\norm{w}^{2}\bigr),\qquad
  v=T_{a,b}^{(2)}w+O\bigl(C_q\norm{w}^{2}\bigr),
\]
where $(T^{(1)}_{a,b},T^{(2)}_{a,b})w=(p,q)=T_{a,b}w$ is exactly the
map of Lemma~\ref{lem:tmap} and $C_q$ is the uniform second-order
constant of \ref{p:compact}. Hence
\[
  \norm{(u,v)}\;\ge\;\norm{T_{a,b}w}-\sqrt2\,C_q\norm{w}^{2}
  \;\ge\;c_T\norm{w}-\sqrt2\,C_q\norm{w}^{2},
\]
by Lemma~\ref{lem:tmap} and \eqref{eq:cT-uniform}.

\emph{Step 4: composition.} Combining Steps 2 and 3,
\[
  \gapf(z)\;\ge\;
  \frac{\min(\xi_{\min},\psi_{\min})}{\sqrt m}
  \Bigl(c_T\norm{w}-\sqrt2\,C_q\norm{w}^{2}\Bigr)
  \;\ge\;
  \frac{\min(\xi_{\min},\psi_{\min})}{2\sqrt m}\,c_T\,\norm{w}
  \;=\;c_-\norm{w}
\]
whenever $\norm{w}\le\delta_2:=
\min\bigl\{\delta_0,\;c_T/(2\sqrt2\,C_q)\bigr\}$.
\end{proof}

\begin{remark}[Where the lower bound fails: degenerate equilibria]
\label{rem:degenerate}
The entire nontrivial content of the lower bound lies in establishing the injectivity of $T_{a,b}$ on $W$ (Lemma~\ref{lem:tmap}). In particular, on each plane corresponding to a positive singular value, $T_{a,b}$ is injective regardless of the rank of $\Bmat$. By virtue of the kernel condition of Assumption~\ref{asm:unique-strict}, these planes span all of $W$, namely the orthogonal complement of the two scale directions in the face. If the equilibrium is
not unique, every strictly complementary equilibrium is degenerate
(Lemma~\ref{lem:kernel}): for one player, say the column
player, there is a direction
$d$ with $\Atil d=0$,
$\mathbf{1}^\top d=0$ beyond $\mathrm{span}\{\psi\}$; moving along it
leaves $u=\Atil(y-\psi)$ identically zero while $\dist(z,E)$ grows
linearly, so $\gapf\gtrsim\dist(z,E)$ \emph{fails}. The correct
comparison object is then $\dist(z,\Fix)$ with respect to the full
(higher-dimensional) fixed-point set, whose normal-bundle analysis we
leave open.  The failure is localized to the
$E\oplus W$ decomposition: $E$ is two-dimensional and $c_-$ is built
from $\sigmin^{+}$ over $W$.  $\Bmat$ keeps a smallest positive singular
value whenever $\mathrm{rank}\,\Bmat>0$; the obstruction is that the
kernel of $\Bmat$ exceeds $\mathrm{span}\{\psi\}$, \rev{so additional kernel directions appear that $\sigmin^{+}$ does not control, and the decomposition of \ref{p:decomp}, or the injectivity of $T_{a,b}$ on $W$, fails.} \rev{In strategy coordinates $\gapf$ is a max minus a min of finitely many affine functions, so it admits a polyhedral error bound relative to the \emph{full} equilibrium set with or without uniqueness (for $A=uu^\top$, $u=(1,-1,0)^\top$, the equilibrium set is a product of segments and $\sqrt2\,\dist((x,y),\mathrm{NE})\le\gapf(x,y)\le2\,\dist((x,y),\mathrm{NE})$).  Uniqueness is used in Lemma~\ref{lem:gap-sharp} to identify the local fixed-point set with the two-dimensional cone $E_+$, so that $\gapf$ can be compared with $\dist(z,E)$ in state coordinates; the lemma quantifies the constant through $(\xi_{\min},\psi_{\min},\sigmin^{+},\eta)$.  Extending the state-coordinate comparison to non-unique equilibria requires accounting for the additional fixed-point directions, not a new sharpness mechanism.} As $\sigmin^{+}\to0$
(near-degenerate game) or $\xi_{\min}\to0$ (near-boundary support,
near-failure of strict complementarity) the constant $c_-$ degrades
continuously \rev{along the nondegenerate stratum (fixed support and rank)}, consistent with the slow large-support instances
observed in Section~\ref{sec:experiments}.
\end{remark}

\begin{remark}[Measuring $\Ndottot$]
\label{rem:ndot-measure}
Differencing recorded $\Nt^2$ fails in double precision deep in the linear regime, where rounding swamps the increment.  The right-hand side of the increment identity is a sum of small positive terms from $g^{(t)}$; the experiments accumulate it in a separate register $G_S(t)$, which stays small ($O(\Nt^2)$, against the $O(10)$ of the full-energy cumulant), and difference $G_S$ between recorded times under a representability guard (Appendix~\ref{app:exp-protocols}).  The guard is a diagnostic on the differenced increment, not an error bound on the accumulated sum.  \rev{In exact arithmetic the differenced register over consecutive records $t_{k-1}<t_k$, divided by $\Delta:=t_k-t_{k-1}$, is the window average $\bar{\mathcal N}_k$ of $\Ndottot$ over $[t_{k-1},t_k)$, not its value at $t_k$; in floating point it is an estimator of that average with its own accumulation error.  The certificate and Corollary~\ref{cor:theta} are statements about the instantaneous quantity.  The protocol uses a fixed spacing $\Delta=\max(10,T/4000)$ and pairs $\bar{\mathcal N}_k$ with the geometric mean $\bar g_k=\sqrt{\gapf(t_{k-1})\gapf(t_k)}$ of the two endpoint gaps (Appendix~\ref{app:exp-protocols}).  For the pure exponential model $\Ndottot(s)=a\rho^{2s}$, $\gapf(s)=b\rho^{s}$ one computes $\bar{\mathcal N}_k/\bar g_k^{\,2}=(a/b^2)(1-\rho^{2\Delta})/(\Delta(1-\rho^2)\rho^{\Delta})$, a constant (with the geometric mean of all $\Delta$ gaps in the window the last factor is $\rho^{\Delta-1}$): at fixed spacing the exponent $2$ is preserved while the offset is not that of the instantaneous law.  We do not prove a transfer of the slope band to windowed data in general; it requires a bound on the spread of $\gapf$ within each window.  A record schedule whose spacing grows geometrically changes the fitted exponent, because the averaged increment is governed by the start of each window while the paired gap represents a later effective time (for the pure exponential model with spacing ratio $\lambda$ the asymptotic exponent is $2/\lambda$ under endpoint pairing and $4/(1+\lambda)$ under geometric-mean pairing, with a logarithmic correction from the division by the window length).}
\end{remark}
\begin{remark}[Both players by design]
\label{rem:single-player}
The \emph{per-player} analogue of Lemma~\ref{lem:ndot-theta} is false, because the rotation carries $w^{(t)}$ through $\{w_1=0\}$ (Remark~\ref{rem:single-player-window}, which also records a windowed repair).  The certificate therefore uses the \emph{aggregate} $\Ndottot$, for which Lemma~\ref{lem:ndot-theta} holds at every step, with no windowing or genericity condition.
\end{remark}

\subsection{Why the certificate aggregates both players}
\label{app:cert-single-player}

\begin{remark}[Failure and partial repair of the per-player bound]
\label{rem:single-player-window}
The per-player lower bound
$\Ndot_1(t)\gtrsim\norm{w^{(t)}}^2$ is \emph{false} pointwise in $t$:
by Lemma~\ref{lem:g-expansion},
$\Ndot_1(t)=\eta^4\lVert\Bmat\Bmat^\top w^{(t)}_1\rVert^2
+O(\norm{w^{(t)}}^3)$ vanishes to first order whenever
$w^{(t)}_1=0$, i.e.\ whenever the local rotation carries the
transverse state through the opponent's axis of the modal plane $W_k$
(the invariant plane $W_k=\mathrm{span}\{(u_k,0),(0,v_k)\}$ of the $k$th singular pair, Lemma~\ref{lem:tmap}). Two partial
repairs are available.
(i) \emph{Two-step window, single mode (rigorous).} Suppose
$w^{(t)}\in W_k$ for a single $k$. The linearized dynamics on $W_k$
is conjugate to $|\mu_k|R(\varphi_k)$, a rotation--dilation with
$\varphi_k=\arg\mu_k\in(0,\pi)$ and $\sin\varphi_k\neq0$
(Theorem~\ref{thm:rate}). For any nonzero linear functional $\ell$ on
the plane and rotation $R$ by angle $\varphi$ with $\sin\varphi\neq0$,
the $2\times2$ matrix with rows $\ell$ and $\ell\circ R$ has determinant
$-\sin\varphi\,\norm{\ell}^2\neq0$, whence
$\max\bigl(|\ell(\zeta)|,|\ell(R\zeta)|\bigr)\ge
c(\ell,\varphi)\norm{\zeta}$ for all $\zeta$. Applying this with
$\ell$ the $w_1$-coordinate functional gives
$\max\bigl(\Ndot_1(t),\Ndot_1(t{+}1)\bigr)\ge
c_k\norm{w^{(t)}}^2-O(\norm{w^{(t)}}^3)$. This covers the case in
which the trajectory sits in one modal plane; the three-step bound
(ii) below covers the general case.
(ii) \emph{Three-step window, linearized dynamics.} The modal planes
$W_k$ are invariant under the frozen-face \emph{Jacobian}
(Lemma~\ref{lem:blockdiag}), not under the nonlinear map itself, so the
statement below is about the linearization $w\mapsto J_S w$; the
nonlinear map adds the $O(\norm{w}^2)$ remainder of
Lemma~\ref{lem:taylor}, which the cubic error term absorbs over a
window of fixed length.  Writing $w_1^{(t)}=\sum_k\alpha_k(t)u_k$,
$w_2^{(t)}=\sum_k\beta_k(t)v_k$, the
linearized dynamics decouple per mode as
$\bigl(\alpha_k(t{+}1),\beta_k(t{+}1)\bigr)^\top=J_k
\bigl(\alpha_k(t),\beta_k(t)\bigr)^\top$ with
$J_k=\Bigl(\begin{smallmatrix}1-q_k&\sigma_k/b\\-\sigma_k/a&1-q_k\end{smallmatrix}\Bigr)$,
$q_k:=\eta^2\sigma_k^2$ (Theorem~\ref{thm:rate}). Three consecutive
$w_1$-observations of one mode read
$M_k\bigl(\alpha_k(t),\beta_k(t)\bigr)^\top$ with
\[
M_k:=\begin{pmatrix}1&0\\1-q_k&\sigma_k/b\\(1-q_k)^2-q_k&2\sigma_k(1-q_k)/b\end{pmatrix},
\qquad
\det(M_k^\top M_k)=\frac{\sigma_k^2}{b^2}\bigl(|\mu_k|^4+4(1-q_k)^2+1\bigr)>0 .
\]
The two columns are never collinear (their first entries are $1$ and
$0$, and $\sigma_k/b\neq0$), so each mode is observable from three
consecutive $w_1$ observations with \emph{no condition on} $q_k$;
with $\lambda_k:=\lambda_{\min}(M_k^\top M_k)>0$,
$\sum_{s=0}^{2}\alpha_k(t{+}s)^2\ge\lambda_k
(\alpha_k(t)^2+\beta_k(t)^2)$. By
Lemma~\ref{lem:g-expansion},
$\Ndot_1(t)=\eta^4\sum_k\sigma_k^4\alpha_k(t)^2+O(\norm{w^{(t)}}^3)$,
and $\sum_k(\alpha_k^2+\beta_k^2)=\norm{w}^2$ (orthonormal singular
vectors); summing over $s=0,1,2$ and absorbing the cubic error as in
Lemma~\ref{lem:ndot-theta} gives, below a threshold,
\[
\sum_{s=0}^{2}\Ndot_1(t{+}s)\;\ge\;\tfrac12\,\eta_-^4\,c_O
\,\norm{w^{(t)}}^2,\qquad
c_O:=\inf\min_k\sigma_k^4\lambda_k>0 ,
\]
with the infimum over the compact regime in $(\eta,a,b)$; positivity
follows from the displayed determinant and compactness. Three
steps suffice, independent of $\dim W$, because each block is a
rotation--dilation, and with the infimum the bound is uniform in $t$.
Repeated singular values
need no separate treatment: the singular vectors can still be chosen
orthonormal, the blocks $J_k$ coincide for equal $\sigma_k$, and the
modal energies $\alpha_k^2$ add across the merged block, so the same
determinant computation applies.
The certificate itself needs neither repair: it uses $\Ndottot$, for
which Lemma~\ref{lem:ndot-theta} holds at every step.
\end{remark}

\subsection{Finite-time correction to the fitted slope}
\label{app:cert-gap7}

\begin{remark}[Open: finite-time correction to $\theta$]
\label{rem:gap7}
Corollary~\ref{cor:theta} controls the fitted slope through the
\emph{width} of the offset band only, and the offset $C(t)$ need not
converge.  \rev{Each linearized transverse mode is conjugate to a rotation--dilation; in the rotation coordinates both the increment energy $\Ndottot(t)$ and the gap $\gapf(t)$ (whose leading homogeneous term is a maximum of finitely many linear forms in the transverse state) may carry direction-dependent leading factors, and with several modes the leading offset depends on the full normalized transverse direction $\varphi_t$, including the relative modal amplitudes, not on a single phase.  The log-ratio can therefore retain bounded oscillations as the amplitude contracts, and convergence of the offset is not guaranteed.}  A sharper
finite-time statement would therefore expand
$C(t)=C_\infty(\varphi_t)+O(\norm{w^{(t)}})$ with an explicit
\rev{direction-dependent} leading term and an explicit constant in the
remainder; the latter is inherited from the non-closed-form
second-order constants of \ref{p:compact}. Likewise, the practical
choice of the measurement window (discarding the pre-regime segment)
requires an operational regime-detection criterion.  Candidate rules
combine relative stagnation of the norm increments with detection of
support freezing, and we treat them as part of the experimental
protocol (Section~\ref{sec:experiments}); the theorem does not supply
one.  Both items are open.
\end{remark}

\section{Experimental Details}\label{app:experiments}

We document the experimental protocol of Section~\ref{sec:experiments}, covering the instance family, the algorithm implementation, every measurement window and numerical guard, and a case-by-case accounting of the excluded instances. All experiments run in pure NumPy/SciPy, single-threaded per instance, with fixed seeds throughout. No GPU is required.

\subsection{Instance family}\label{app:exp-instances}

Each instance is a square zero-sum matrix game $A \in \R^{n \times n}$, where the row player maximizes. We draw these games from the grid:
\[
\begin{aligned}
&n \in \{5, 20, 50\}\\
&{}\times \text{distribution} \in \{\textsf{gaussian}, \textsf{uniform},
  \textsf{pm1\_sparse}\}\\
&{}\times \text{generator} \in \{\textsf{plain}, \textsf{antisym}\}
\end{aligned}
\]
We run $12$ seeds per cell. This grid yields $3 \times 3 \times 2 \times 12 = 216$ instances. The \textsf{gaussian} distribution samples i.i.d.\ standard Gaussian entries, while \textsf{uniform} samples i.i.d.\ uniform entries on $[-1, 1]$. Under \textsf{pm1\_sparse}, we draw sparse-sign entries that take values $\pm 1$ with equal probability, keeping each independently with probability $0.3$ (zero otherwise). The \textsf{plain} generator uses the raw draw $G$. For \textsf{antisym}, we compute the antisymmetrization $(G - G^\top)/\sqrt{2}$, yielding symmetric (value-zero) games with qualitatively larger equilibrium supports. Seeds are deterministic ($10\,000 + 997\,u$ for instance index $u$).

For every game, we solve both players' linear programs (LP) via SciPy/HiGHS. This yields an equilibrium $(x^*, y^*)$, the value $v^*$, the supports, the strict-complementarity margin, and a duality-gap sanity check. Supports use a threshold of $10^{-8}$, with HiGHS default primal and dual feasibility tolerances of $10^{-7}$. We define the strict-complementarity margin as:
\[
\delta:=\min\bigl(\min_{i \notin S_1}(v^* - (A y^*)_i),\;
\min_{j \notin S_2}((x^{*\top}A)_j - v^*)\bigr).
\]
We use the LP equilibrium only for reference statistics: support-match rates and $|\gamma - v^*|$. All spectral quantities in the protocols for verifications~(3) and~(5) are rebuilt from the \emph{actually converged} state. On $34/216$ instances ($15.7\%$), the realized exact-zero pattern differs from the thresholded LP support. Multiple equilibria provide one possible explanation, which remains unverified as a classification.

\subsection{Algorithm implementation}\label{app:exp-algo}

We implement Algorithm~2 of \citet{zhang2025scale} (\algname{}) in the
extra-gradient setup. The prediction $m^{(t)}$ equals the payoff
against the opponent's pre-iterate. Two lines require care, as a
plausible-looking variant breaks the sharpest theoretical prediction:
\begin{itemize}[leftmargin=2em]
  \item \emph{Prediction step (Algorithm~2, line~9).}
    The vector $r^{(t)} = \rt^{(t)} + m^{(t)} - \gamma^{(t)}\mathbf{1}$ is
    \emph{unclipped}. We set the strategy to $x^{(t)} =
    \pos{r^{(t)}}/\lVert\pos{r^{(t)}}\rVert_1$. The shift $\gamma^{(t)}$
    solves $\lVert\pos{\rt^{(t)} + m^{(t)} - \gamma\mathbf{1}}\rVert_2 =
    \Nt(t)$ via the sorted piecewise closed form, followed by a Newton
    refinement on $f(\gamma) = \lVert\pos{v - \gamma}\rVert_2^2 -
    \Nt(t)^2$. This refinement avoids the catastrophic cancellation of the pure
    closed form when $\Nt(t) \ll \norm{v}$.
  \item \emph{Update step (Algorithm~2, line~13).}
    We compute $\rt^{(t+1)} = \pos{r^{(t)} + g^{(t)}}$ using the \emph{unclipped}
    $r^{(t)}$. Here $g^{(t)}$ denotes the instantaneous regret of the
    correction step.
\end{itemize}
An unclipped prediction step yields a sharp analysis.
Clipping it as well ($\rt^{(t+1)} = \pos{\pos{r^{(t)}} + g^{(t)}}$) still
yields a valid regret minimizer that retains norm monotonicity and scale
invariance. Near a strict-complementarity fixed point, however, this variant emits
$O(\norm{g})$ off-support ``dust.'' Its one-step map loses smoothness at
the fixed point, and exact support freezing fails. All results in this
paper use the strict Algorithm~2. Under this exact rule, the exact-zero pattern freezes
in finite time as predicted by Lemma~\ref{lem:absorb}. We use this freeze as a regression check for our implementation.

We set the horizon to $T = 2 \times 10^5$ for $n \in \{5, 20\}$ and to $T = 10^6$ for $n = 50$. We record the time series of the gap, the norms, the cumulative $\sum\norm{g}^2$, the cumulative truncation loss, and $|\gamma - v^*|$ every $\max(10, T/4000)$ iterations. The snapshot trigger below checks the gap every $10$ iterations.

\subsection{Measurement protocols}\label{app:exp-protocols}

\paragraph{Support freezing (verification~1).}
We track two patterns after every iteration jointly across both players: the relative pattern $\{i : \rt_i > 10^{-6}\,\Nt(t)\}$ and the exact-zero pattern $\{i : \rt_i > 0\}$. We define $t_{\mathrm{freeze}}$ as the iteration of the \emph{last} pattern change over the whole run. The headline numbers refer to the exact-zero pattern, the hard prediction of Lemma~\ref{lem:absorb}. Here $99.5\%$ freeze before $T/2$, with median $t_{\mathrm{freeze}}/T = 2.7 \times 10^{-4}$, p90 $1.5 \times 10^{-3}$, and max $0.52$. The relative pattern gives identical medians.

\paragraph{Snapshot.}
A state snapshot $z = (\rt_1, \rt_2)$ is stored at the first $t \ge 200$
(checked every $10$ iterations) at which the gap falls below $10^{-4}$ of
its value at $t = 1$ \emph{and is still strictly positive}.  On this suite the
second condition is the only one that ever blocks the trigger: exactly $32$
instances have no snapshot, because they jump past the threshold to a
floating-point zero gap before a positive-gap check lands, and report a gap of
exactly $0.0$ from then on ($29$ at $n = 5$, $3$ at $n = 20$).  Those $32$ are
excluded from the spectral and $\theta$ measurements (counted below), since no
positive-gap segment remains to measure.  Every instance with a snapshot ($184$ of $216$) also
converged (final gap below $10^{-10}$ of its initial value), so all $216$
instances converge and the \emph{measurable} pool, on which the spectral,
$\theta$ and step-size statistics are computed, is the $184$ with a snapshot.

\paragraph{The $32$ finite-step convergers.}
A gap of exactly $0.0$ in double precision could be rounding. Two checks rule this out. In $31$ of the $32$ games, a pure saddle point exists (an entry that is both a column maximum and a row minimum). This corresponds to the $m=1$ case, where Section~\ref{sec:local} predicts convergence in one round after absorption. The remaining game is a $\pm1$ sparse instance of value $0$ with a non-unique equilibrium and no pure saddle. We also re-ran all $32$ instances in $60$-digit arithmetic using bisection for the $\gamma$-shift over $200$ rounds. Every instance attains a gap that is \emph{exactly} zero. This occurs at the double-precision round in $30$ instances and one round later in the other two. The latest such round is $29$, attained by the game without a pure saddle. Finite-step convergence on these instances is a property of the iteration, and the excluded group is the exact set on which no exponential segment exists.

\paragraph{Finite-difference (FD) Jacobian (verification~2).}
The one-step map $\Phi$ acts on the state $z = (\rt_1, \rt_2)$. We form its Jacobian at the snapshot column by column with step $\eps = 10^{-7} \max_j z_j$. Coordinates with $z_j > 10 \eps$ use central differences, and the remaining coordinates use forward differences because negative $\rt$ is not a reachable state. A consistency probe recomputes up to $6$ columns with step $10\eps$ and reports the median relative column difference. Across instances, this consistency indicator has median $6.0 \times 10^{-7}$ (p90 $1.7 \times 10^{-6}$). It serves as an accuracy diagnostic for the FD probe and matches the scale of the maximum deviation seen in verification~(3). We classify eigenvalues with $|\lambda| > 1 - 10^{-5}$ as the \emph{unit band}, the numerical stand-in for the two exact unit eigenvalues of Theorem~\ref{thm:rate}(i) (median count $2$, min $2$, max $10$). The quantity $\rho_2(J)$ is the largest modulus below the band and acts as the numerical stand-in for $\rho^*$. A slowly contracting mode with $\eta\sigma<3\times10^{-3}$ would also fall in the band, so a band count above $2$ suggests a degenerate equilibrium without proving one. For the $8$ flagged instances, the kernel dimension of $\Bmat$ on the converged support confirms or refutes degeneracy structurally (verification~6). The predicted per-iteration slope is $\log \rho_2(J)$. This verification tests the mechanism (rate $=$ frozen-face spectrum) with a numerical Jacobian, while verification~3 tests the closed form against the same Jacobian.

\paragraph{Measured rate (verification~2, continued).}
We compute the measured slope via an ordinary least-squares fit of $\log \gapf$ against $t$ on the exponential segment.  We retain recorded points whose gap lies in $[10^3 \cdot \mathrm{floor},\; 0.5 \cdot \mathrm{gap}_{\mathrm{snap}}]$, where $\mathrm{floor}$ denotes the smallest positive gap attained in the run.  If the run never reaches its floor, the lower cut becomes $1.5 \times \mathrm{floor}$.  An instance is \emph{usable} under four conditions: it has a snapshot, the fit uses at least $5$ points, $R^2 \ge 0.9$, and the slope is negative.  Across all $216$ instances, we obtain $155$ usable runs.  Another $32$ runs lack a snapshot.  These correspond to the finite-step convergers of the snapshot paragraph, all having $n \in \{5, 20\}$ and a final gap of exactly $0.0$.  The remaining $29$ instances lack a fittable segment because they hit the double-precision floor almost immediately after the snapshot.  Both excluded groups are visible in the raw traces.  The no-snapshot group reaches a zero gap in double and $60$-digit arithmetic alike before the snapshot trigger fires.  Each of the $29$ non-fittable instances has fewer than $5$ usable points below half the snapshot gap, and one has fewer than $10$ positive gaps overall.  The $R^2$ filter never triggers on them.  One pass over the traces produces the $216 = 155 + 32 + 29$ accounting.  An independent second implementation of the same filter reproduces it, and the two agree on all three counts.  The freeze-time split of the $155$ reported in Section~\ref{sec:experiments} comes from this accounting.

\paragraph{Closed form versus FD (verification~3).}
Our eligible pool consists of the $184$ instances with a snapshot, all of which reach a final gap below $10^{-10}$ of its initial value. From this pool, we sample $32$ instances uniformly without replacement at a fixed, recorded seed. Each instance is re-run for $4\,t_{\mathrm{snap}} + 2000$ iterations to a fully converged state $z$. We take the support blocks from the exact-zero pattern of $\rt$. The normalized strategies are $\xi = \rt_1|_{S_1}/c_1$, $\psi = \rt_2|_{S_2}/c_2$ with $c_i =
\lVert\rt_i\rVert_1$. The closed-form spectrum is $\mu = 1 + \lambda +
\lambda^2$, $\lambda = \pm i \sqrt{q}$, evaluated over the nonzero eigenvalues $q$ of $Q = \Atil P_\psi \,\Atil^\top P_\xi / (c_1 c_2)$. Kernel modes satisfy $|q| \le 10^{-9}$ relative. They correspond to the unit-band directions and are excluded from the rate. We evaluate the comparison FD Jacobian at the \emph{same} state $z$. Relative errors on $\rho_2$ have a median of $2.4 \times 10^{-10}$ and a maximum of $4.1 \times 10^{-7}$. This matches the FD consistency limit quoted above across all $32/32$ instances, including one degenerate instance ($5 \times 8$ support blocks, $9$ unit eigenvalues). Building $\Bmat$ from the thresholded LP support produces relative errors up to $2.66\%$. This discrepancy concentrates on instances whose realized support differs from the thresholded LP support, matching the $84.3\%$ support-match rate of Section~\ref{app:exp-instances}. Rebuilding from the converged support removes this discrepancy.

\paragraph{Certificate exponent $\theta$ (verification~4).}
The window is the linear regime $t \ge t_{\mathrm{snap}}$, clipped to start at $t \ge \max\{t_{\mathrm{snap}}, t_{\mathrm{freeze}}\}$. Here $t_{\mathrm{freeze}}$ marks the last exact-zero support-pattern change. This defines a \emph{candidate} frozen window with a fixed exact-zero pattern of $\rt$. We did not separately check the full clipping pattern behind Lemma~\ref{lem:ndot-identity} (positivity of $r$ and $r+g$ on the support at both the prediction and update stages). The on-support energy serves as an empirical proxy for $\Ndottot$ on this window.

We build the primary estimator $\hat\theta_{\mathrm{supp}}$ from the theorem's quantity. At each step we add the \emph{on-support} update energy $\sum_i \sum_{j \in S_i} (g^{(t)}_{i,j})^2$ into a running cumulative $G_S(t)$ stored at the recorded times. Once the pattern freezes, this increment equals $\Ndottot(t)$ by Lemma~\ref{lem:ndot-identity}. Between consecutive recorded times $t_k<t_{k+1}$, we form the interval average: $\bar\Ndottot_k:=(G_S(t_{k+1})-G_S(t_k))/(t_{k+1}-t_k)$. We measure a per-unit-time average of the on-support energy over the interval. This matches the average of the theorem's pointwise $\Ndottot$ only if the full clipping pattern is frozen on the window. The pointwise sandwich of Theorem~\ref{thm:certificate} transfers to such averages with the same $\kappa$ when reading the interval gap as the root-mean-square gap over the interval (averaging Corollary~\ref{cor:theta}'s two-sided bound over $t\in[t_k,t_{k+1})$). We pair $\bar\Ndottot_k$ with the geometric mean $\sqrt{\gapf(t_k)\gapf(t_{k+1})}$ of the endpoint gaps as an empirical proxy for that RMS. For a geometrically decaying gap with oscillation, the two differ by a bounded factor that the theory leaves unspecified. The interval statistics below test the ratio law empirically. The pointwise inequality itself remains unverified.

The cumulative $G_S$ stays small ($O(\Nt^2)$, compared to the $O(10)$ of the full-energy cumulant), keeping its differences representable deep in the linear regime (Remark~\ref{rem:ndot-measure}). We retain a difference only if it passes the representability guard ($> 64\,\eps_{\mathrm{mach}}$ relative to $G_S(t_{k+1})$) and the interval's geometric-mean gap exceeds the absolute floor $10^{-13}$ (payoffs are $O(1)$). We require at least $6$ surviving increments and fit the OLS slope of $\log\bar\Ndottot_k$ against the log geometric-mean gap.

The cross-check estimator $\hat\theta_{xy}$ applies the same protocol to the \emph{full} update energy $\Delta G^2$ of the cumulative $\sum_t \norm{g^{(t)}}^2$ (adding the off-support energy). A third estimator $\hat\theta_{\mathrm{ratio}}$ takes the ratio of the OLS time-slopes of $\log(\Delta G^2/\Delta t)$ and of $\log \gapf$, averaging the complex-spectrum oscillation out of numerator and denominator before dividing. All three estimators yield $129/216$ usable instances. Exclusions comprise the $32$ no-snapshot instances plus fast floor-hitters with fewer than $6$ guard-surviving increments. Usable windows span a median of $2.8$ decades of gap decay (p10 $2.4$). We obtain these medians: $\hat\theta_{\mathrm{supp}}$ gives $1.986$ (p10 $1.911$, p90 $2.038$, with $96.1\%$ between $1.8$ and $2.2$), $\hat\theta_{xy}$ gives $1.995$ ($98.4\%$), and $\hat\theta_{\mathrm{ratio}}$ gives $2.004$ ($98.4\%$). All three agree with each other and with the predicted $2$ to within $1\%$.

A fourth, deprecated estimator differences the recorded cumulants (cumulative $\sum\norm{g}^2$ minus cumulative truncation loss). It yields a consistent yet degraded median of $1.948$, putting only $92/128$ ($71.9\%$) between $1.8$ and $2.2$. One instance lacks enough guarded increments for this estimator. Its failure mode stems from numerical precision, distinct from \emph{statistical} noise. Once the cumulants reach order $10^{1}$, their accumulated floating-point rounding ($\sim 10^{-13}$ to $10^{-12}$ absolute) exceeds the true late-regime increments and creates a spurious constant plateau. The on-support estimator keeps a small cumulant and recovers every lost instance exactly (Remark~\ref{rem:ndot-measure}).

The protocol above replaces a naive measurement. Differencing $\Nt^2$ across log-spaced checkpoints over the whole run yields a median exponent of $0.375$ on the same data. This estimator suffers from two failure modes: (i) it mixes the transient and linear regimes in one fit, and (ii) once $\gapf \sim 10^{-12}$ the true increment of $\Nt^2$ falls below double-precision resolution, leaving the computed differences dominated by rounding noise. Both are measurement artifacts. Our guarded protocol eliminates them while leaving the dynamics untouched.
\paragraph{Self-stabilization (verification~5).}
We run this verification on the \emph{entire} measurable pool of all $184$ instances with a snapshot, using no sampling. Each instance is re-run to $4\,t_{\mathrm{snap}} + 2000$ iterations. We rebuild $\Bmat$ from the converged support and strategies as above, then evaluate $\etainf \sigmax(\Bmat)$ with $\etainf =
1/\sqrt{c_1 c_2}$ at the realized saturation levels. Out of $0$ of $184$ instances, none violate $\etainf\sigmax(\Bmat) \le 1$. The distribution has median $0.790$, p90 $0.946$, and maximum $0.9985$, with $18$ instances exceeding $0.95$. For $q_{\max} = (\etainf
\sigmax(\Bmat))^2$, the minimum distance to the boundary $1$ is $0.0029$. The constraint holds with room to spare. The median instance sits $0.21$ below the boundary, and the $18$ instances above $0.95$ form the tail that genuinely tests the heuristic. This procedure evaluates the self-stabilization heuristic (Remark~\ref{rem:mechanism}). Proposition~\ref{prop:selfstab} remains untested here because it permits convergence to a violating ray along its center-stable set. A converged instance above $1$ would refute the heuristic and exhibit such an exceptional orbit, while remaining consistent with the theorem.

\paragraph{Ratio certificate (verification~6).}
On the $129$ $\theta$-usable instances we re-run each trajectory to $\min(T,\,4\,t_{\mathrm{snap}} + 2000)$ iterations. We work inside the same linear-regime window as the $\theta$ measurement. The empirical deviation is
\[
\kappa_{\mathrm{emp}}
\;=\;
\max_{t_1 < t_2 \in \text{window}}
\max\!\left\{ R(t_1,t_2),\, R(t_1,t_2)^{-1} \right\},
\qquad
R(t_1,t_2)
=
\frac{\gapf(t_2)/\gapf(t_1)}
     {\sqrt{\Ndottot(t_2)/\Ndottot(t_1)}},
\]
where $t_1,t_2$ range over recorded times. Here $\Ndottot$ and $\gapf$ stand for the interval quantities of the $\theta$ protocol ($\bar\Ndottot_k$ and the geometric-mean gap), computed under three numerators. The primary numerator accumulates the on-support energy of Lemma~\ref{lem:ndot-identity} between recorded times, as detailed in Remark~\ref{rem:ndot-measure}, for which $124$ instances qualify. Full-update-energy increments $\Delta \sum_t \norm{g^{(t)}}^2$ serve as a cross-check across all $124$ instances. A third numerator uses the telescoped increments $\Nt(t{+}1)^2 - \Nt(t)^2$ under the representability guard of the $\theta$ protocol. At least $6$ guarded increments are required, so $123/124$ instances qualify. This third choice represents the naive monitor obtained by differencing the recorded norm.

We assemble the theoretical constant per instance from the definitions in Theorem~\ref{thm:certificate}. We compute $L_{\gapf}$ and $c_-$ from the converged support, strategies, and value (rebuilt as in verification~3), while $\sqrt{n_+/n_-}$ comes from the extreme nonzero singular values of $\Bmat$ and the realized window range of $\eta$. We evaluate the proof formulas over the realized window range of the scales $a_t=\Nt_1(t)/\norm{\xi}_2$ and $b_t=\Nt_2(t)/\norm{\psi}_2$, setting $\eta_{\pm}$ to the extreme values of $1/\sqrt{a_tb_t}$. These scales equal the projection coordinates $a,b$ of the decomposition $\rt_1=a\xi+w_1$ only on the cone, differing by $O(\norm{w}^2)$ inside the regime. The explicit expressions are: $n_-=\tfrac12\eta_-^4(\sigmin^{+})^4$, $n_+=2\eta_+^4\sigmax^4$, $c_-=\tfrac{\min(\xi_{\min},\psi_{\min})}{2\sqrt m}\,c_T$ with $c_T=\sigmin^{+}\eta_-^2\big/\sqrt{2\sigmax^2\eta_+^4+a_-^{-2}+b_-^{-2}}$, and $L_{\gapf}=\sqrt2\max_k\sigma_k\,\sigmax\bigl(\begin{smallmatrix}-\sigma_k/(a_-b_-)&1/b_-\\1/a_-&\sigma_k/(a_-b_-)\end{smallmatrix}\bigr)$. This last expression provides a first-order estimate of the Lipschitz constant of the gap map on $W$. The proof's $L_{\gapf}$ is existential, and its ordering relative to our explicit experimental value is unestablished. The assembled $\kappa$ is therefore an uncertified benchmark. We replace the regime constant with the realized window range of $\eta$ in the same spirit. Because the radii in $\norm{w}\le\min(\delta_2,\delta_3)$ are not explicit, we do not check regime membership. We define the window operationally via the frozen support and the linear-regime start. Failing the inequality against this benchmark would not refute the theorem, and passing does not verify it. Our comparison calibrates the empirical deviation against a first-order scale.

We exclude five instances before assembly. Four exhibit degenerate kernels ($\dim \ker \Bmat > 1$ on the converged frozen support, all in the $n = 20$ sparse antisymmetric family). One has a non-square support, leaving $c_-$ and the decomposition $E\oplus W$ undefined by the theory. Under the on-support energy increment numerator, the inequality $\kappa_{\mathrm{emp}} \le \kappa$ holds on $124/124$ instances, with $\kappa_{\mathrm{emp}}$ having median $1.63$, p90 $2.35$, and max $4.1$. The full-energy cross-check agrees on $124/124$ instances (median $1.63$, p90 $2.18$, max $4.1$). The telescoped monitor satisfies the bound on $119/123$ instances. Its $4$ exceptions stem from accumulated-rounding artifacts (Remark~\ref{rem:ndot-measure}) that conservatively inflate the apparent deviation. All four pass under the on-support numerator. The assembled $\kappa$ has median $1.7 \times 10^{8}$, yielding a median slack of about $10^{8}$.

\paragraph{Calibrated point estimate (verification~7).}
We evaluate the same $129$ $\theta$-usable instances over the same linear-regime window. We fix a calibration time $t_1$ at the first guarded increment. For a sensitivity check, we also evaluate the one-third and one-half points of the window. At every later guarded increment, we compute the estimate $\widehat{\gapf}(t_2) = \gapf(t_1)\sqrt{\Ndottot(t_2)/\Ndottot(t_1)}$ and the multiplicative deviation $D(t_2) = \widehat{\gapf}(t_2)/\gapf(t_2)$. The quantity $D$ is the empirical deviation of the interval proxy. Theorem~\ref{thm:certificate} bounds the pointwise ratio, leaving $D$ without a direct theorem interval. We compare it against the first-order benchmark $\kappa$ of verification~6. For each instance, we record $D_{\max} = \max_{t_2} \max\{D(t_2), D(t_2)^{-1}\}$ and the calibration sensitivity, defined as the ratio of $D_{\max}$ under the two alternate calibration points to that under the first. We also check the internal consistency identity $D_{\max} \le \kappa_{\mathrm{emp}}$ against the $\kappa_{\mathrm{emp}}$ of verification~6. Construction guarantees this bound, so any violation flags an implementation error.

All statistics use the on-support numerator ($129/129$ instances). For $D_{\max}$, the median is $1.50$, the p90 is $1.90$, and the max is $3.9$. We observe $D_{\max} < 2$ on $91.5\%$ of instances and $< 5$ on all of them. Calibration sensitivity has a median of $1.13$. Every instance passes the identity check. The forward gap estimate is accurate to a factor of $2$ on more than nine instances in ten, while the calibration point matters at the $13\%$ level. The telescoped monitor version achieves a $D_{\max}$ median of $1.52$. However, it diverges (exceeding $10^{2}$, the rounding-floor threshold) on $19/129$ instances, reflecting the accumulated-rounding artifact of Remark~\ref{rem:ndot-measure}.

\paragraph{Degenerate-instance sandwich.}
The $129$ $\theta$-usable instances split into three classes by two
criteria fixed before assembly: \emph{degenerate} ($8$) if the snapshot
FD Jacobian has more than $2$ unit-band eigenvalues; \emph{near-degenerate}
($13$, all $n=50$) if the LP equilibrium margin is below $1.2\times10^{-3}$;
\emph{regular} ($108$) otherwise.  The unit-band count is a numerical
flag (see the FD paragraph).  The structural test is the kernel dimension
of $\Bmat$ on the converged frozen support, and on it the $8$ split as
$5$ with $\dim\ker\Bmat>1$ or a non-square support (the verification~6
exclusions, where $c_-$ and $E\oplus W$ are undefined) and $3$ with
$\dim\ker\Bmat=1$ (a third unit-band eigenvalue at the snapshot that
is not a kernel mode of the converged $\Bmat$).  For those $3$, $\kappa$ is
assembled and the inequality holds, so they stay in the count: $124=129-5$.
The \emph{empirical} sandwich is defined on all $8$ because it only
reads the trajectory: over the $\theta$-usable window
$\kappa_{\mathrm{emp}}$ has median $1.80$ on the $8$ (vs.\ $1.63$ on the
$108$ regular and $1.61$ on the $13$ near-degenerate instances), and the
estimate deviation stays below $2$ on $7$ of the $8$ (below $5$ on all).
The ratio law itself does not break at degeneracy; this paper's
constants do.  That supports using the certificate as a monitor.

\paragraph{EFG monitoring (verification~8).}
Our extensive-form solver runs on two small poker games: Kuhn ($3$ cards, $12$ information sets) and Leduc ($3 \times 2$ deck, $1116$ information sets). We instantiate one IR-PRM$^+$ learner per information set (a player's decision point) and feed each the counterfactual utilities of CFR \citep{zinkevich2007regret} in the extra-gradient setup. Our in-house vectorized implementation follows the scale-invariant algorithm of \citet{zhang2025scale}.
Traces are recorded at logarithmic times plus a linear block over the first $2000$ iterations. We monitor the cumulative update energy $G(t)=\sum_{s\le t}\sum_j\norm{g_j^{(s)}}^2$ over information sets $j$, the extensive-form analogue of the full-energy cumulant in the matrix protocol. The solver computes this quantity every iteration. Following the matrix protocol, we evaluate its per-unit-time increment $(G(t_{k+1})-G(t_k))/(t_{k+1}-t_k)$ between recorded times.
Alongside it, we record the aggregate squared regret norm $Q(t)=\sum_j\norm{\rt_j^{(t)}}^2$, which is a level quantity. Its per-unit-time increment is omitted from the headline numbers, coupling to the gap at $0.92$ on Kuhn and $0.998$ on Leduc, compared with $0.998$ and $0.997$ for the update energy.
At the recorded times, we compute the gap via a full-tree best response requiring two passes, one per player. The window and guards mirror the matrix protocol (floor $10^{-13}$, representability guard, at least $6$ guarded points). All runs execute on CPU, reaching $T = 10^{5}$ on Kuhn and $T = 5 \times 10^{4}$ on Leduc.
Wall-clock time measures the best-response cost. One $\mathrm{NashConv}$ pass (one tree traversal per player, computing the extensive-form gap by summing the two players' best-response gains) divided by one solver iteration takes $0.076$ on Kuhn and $0.087$ on Leduc. Moving from an every-iteration NashConv schedule to the sparse record schedule reduces best-response passes by $126\times$ (Kuhn) and $64\times$ (Leduc), matching the ratio $T/\#\{\text{recorded times}\}$. The schedule accounts for this saving.
We report the coupling as $\mathrm{corr}\bigl(\log\tfrac{G(t_{k+1})-G(t_k)}{t_{k+1}-t_k},\,\log\gapf(t_{k+1})\bigr)$ on the guarded window, while the estimate deviation follows verification~7 using $G$ in place of $G_S$.

\paragraph{Matched-time test of the ratio law (verification~9).}
Verifications~4 and~6 pair an interval average of the on-support energy with endpoint gaps and log only the exact-zero pattern of $\rt$. Here we evaluate the quantities of Theorem~\ref{thm:certificate} at the exact same round. Our stratified subset comprises the first two seeds of every distribution and generator at $n=5$ and $n=20$, plus the first seed at $n=50$. These $30$ games run with the horizons of the main protocol. Every $10$ rounds ($50$ at $n=50$), we record the positivity pattern of $r^{(t)}$ and of $\rt^{(t+1)}$ for both players, the update vectors $g_i^{(t)}$, and the exact gap of $(x^{(t)},y^{(t)})$. We also record the increment $\Ndottot(t)$ as: $\sum_i\sum_j(\rt^{(t+1)}_{i,j}-\rt^{(t)}_{i,j})(\rt^{(t+1)}_{i,j}+\rt^{(t)}_{i,j})$. This evaluation uses no cumulative register. The measurement window begins after the last record where either pattern differs from the final support $S_i$, and retains all records with gap above $10^{-12}$. By Lemma~\ref{lem:ndot-identity}, the instantaneous on-support energy $\sum_i\norm{g_i^{(t)}|_{S_i}}^2$ equals $\Ndottot(t)$ throughout this window. We check this identity on records with energy above $10^{-10}$, where the rounding error of the direct increment is negligible. Across the $22$ games containing such records, the relative difference has median $2.7\times10^{-7}$ and maximum $2.1\times10^{-4}$. Out of the $30$ games, $24$ yield at least $20$ records in the window. The remaining $6$ reach an exactly zero gap or the gap floor beforehand. On those $24$ games, we regress $\log\Ndottot$ against $\log\gapf$ via least squares, taking $\Ndottot$ as the on-support energy. The slope attains a median of $2.002$, a minimum of $1.981$, and a maximum of $2.086$. The corresponding windows span $7.0$ to $11.4$ decades of the gap, with a median of $9.8$. The spread of $\gapf/\sqrt{\Ndottot}$ over the window, defined as the ratio of its maximum to its minimum, gives an empirical lower estimate of $\kappa$ in Theorem~\ref{thm:certificate}. It has median $5.5$, 90th percentile $8.7$, and maximum $483$. A single $n=5$ game attains this maximum, as its patterns freeze only at $t=105{,}030$. The 5th-to-95th percentile spread has median $2.4$ and maximum $16.8$. Three of the $24$ games fail the numerical uniqueness flag and exhibit slopes of $2.020$, $2.003$, and $2.014$.

\subsection{Excluded and degenerate instances: summary}\label{app:exp-exclusions}

We collect all exclusions here:
(i)~The spectral comparison yields $155/216$ instances with a fitted slope. We omit $32$ no-snapshot numerical-zero convergers. Their final gap is $0.0$, all have $n \in \{5, 20\}$, and the positive-gap snapshot trigger never fires. Another $29$ instances lack a fittable exponential segment due to fast convergence to the float floor.
(ii)~For the $\theta$ measurement, $129/216$ instances are slope-measurable. The exclusions comprise the same $32$ numerical-zero convergers and fast floor-hitters with fewer than $6$ representable increments.
(iii)~Regarding support match, the frozen pattern differs from the LP support in $15.7\%$ of instances. Multiple equilibria offer one possible explanation, without forming an established classification. For this reason, the protocols for verifications~(3) and~(5) rebuild all spectral quantities from the converged state.
(iv)~To check degenerate equilibria, we flag instances with more than $2$ unit-band eigenvalues (up to $10$) as candidates for a positive-dimensional equilibrium manifold. The kernel dimension of the converged $\Bmat$ confirms five candidates and refutes three. The closed-form $\rho_2$ remains exact there, as the kernel modes land in the unit band.
(v)~Ratio-certificate assembly (verification~6) excludes $5$ of $129$ instances: four with $\dim\ker\Bmat>1$ and one with non-square support. Every exclusion is recorded per instance alongside its source trace. This list is complete.

\subsection{Summary of verifications}\label{app:exp-summary}

Table~\ref{tab:numerics} reports the full statistics behind Table~\ref{tab:main} across all $216$ instances. Reported ranges span the $10$th to the $90$th percentile.

\begin{table}[t]
  \centering
  \caption{Full statistics behind Table~\ref{tab:main}.}
  \label{tab:numerics}
  \small
  \begin{tabular*}{\textwidth}{@{\extracolsep{\fill}}llrp{0.42\textwidth}@{}}
    \toprule
    Prediction & Source & $n_{\mathrm{inst}}$ & Statistic \\
    \midrule
    Exact support freezing before $T/2$ & Lem.~\ref{lem:absorb} & 216 & frozen on $99.5\%$; median $t_{\mathrm{freeze}}/T=2.7\times10^{-4}$ \\
    $\gamma \to v^*$ at the fixed point & Prop.~\ref{prop:fixedpoint} & 216 & median $|\gamma-v^*|=4.47\times10^{-17}$ \\
    Rate equals FD $\rho_2(J_{\mathrm{FD}})$ & Thm.~\ref{thm:rate} & 155 & ratio median $1.000$ (p10--p90 $0.984$--$1.016$); corr.\ $0.9982$ \\
    Closed form $\rho_2 =$ FD $\rho_2$ & Lem.~\ref{lem:jacobian},~\ref{lem:centering} & 32 & relative error median $2.4\times10^{-10}$, max $4.1\times10^{-7}$ \\
    Certificate exponent $\theta = 2$ & Cor.~\ref{cor:theta} & 129 & on-support median $1.986$; $96.1\%$ between $1.8$ and $2.2$; cross-check $1.995$, $98.4\%$ \\
    $\etainf\sigmax(\Bmat) \le 1$ at limits & Rem.~\ref{rem:mechanism}\textsuperscript{$\dagger$} & 184 & $0$ violations; median $0.790$, max $0.9985$; $18$ above $0.95$; closest approach of $(\etainf\sigmax(\Bmat))^2$ to $1$ is $0.0029$ \\
    Ratio certificate (empirical proxy) & Thm.~\ref{thm:certificate} & 124 & holds on $124/124$; $\kappa_{\mathrm{emp}}$ median $1.63$, max $4.1$ \\
    Point-estimate forecast (empirical proxy) & Thm.~\ref{thm:certificate} & 129 & max forward deviation median $1.50$, p90 $1.90$, max $3.9$; below $2$ on $91.5\%$ \\
    Progress monitoring (Kuhn, Leduc) & empirical & 2 & max forward deviation $1.80$, $3.22$; NashConv pass $0.076/0.087$ of an iteration \\
    Degenerate-instance sandwich & empirical & 8 & $\kappa_{\mathrm{emp}}$ median $1.80$ vs.\ $1.63$; below $2$ on $7/8$ \\
    \bottomrule
  \end{tabular*}

\smallskip
{\footnotesize On $216$ random matrix games (three dimensions $\times$ three distributions $\times$ two generators $\times$ $12$ seeds).  \textsuperscript{$\dagger$}Self-stabilization heuristic, empirically consistent with Proposition~\ref{prop:selfstab}.\par}
\end{table}

\section{Additional Experiments}\label{app:add-exp}

All runs use the solver and protocols of Appendix~\ref{app:experiments} unless stated otherwise. Denominators are the full instance sets of each experiment.

\subsection{The family $A_\lambda$: convergence to an unstable ray}\label{app:add-alambda}
The matrix $A_\lambda=uu^\top+\lambda vv^\top$ with $u=(1,-1,-1)^\top$ and $v=(0,1,-1)^\top$ admits a unique fully mixed equilibrium for every $\lambda>0$: $x^*=y^*=(\tfrac12,\tfrac14,\tfrac14)$ and $v^*=0$, with $\sigmax(\Bmat)=\max(3,2\lambda)$. The subspace $\rt_2=\rt_3$ for both players is invariant and contains the zero start. On this subspace, the saturation levels are independent of $\lambda$: $\Nt_{1,\infty}=1.2646$ and $\Nt_{2,\infty}=4.9745$. The ray reached there has $\eta_{\rm sym}\sigmax(\Bmat)\approx0.2442\max(3,2\lambda)$, which exceeds one for $\lambda>\lambda_c:=1/(2\eta_{\rm sym})\approx2.0479$. This ray is a saddle with $d_u=2$. Only the singular value $2\lambda$ exceeds $1/\eta_{\rm sym}$, since $3\eta_{\rm sym}\approx0.73$. The zero start itself lies outside $\mathcal Z$. Our statement regarding Proposition~\ref{prop:selfstab} applies to the tail of the symmetric trajectory. This tail enters a neighborhood of the ray and remains there, placing it within the local center-stable set. Table~\ref{tab:alambda} reports the forced-symmetric ray, the zero-initialized float64 run, and perturbed runs started on the saturated ray.
Over all $240$ perturbed runs, the realized $\etainf\sigmax(\Bmat)$ has median $0.750$ and maximum $0.839$. Random starts on $A_5$ and $A_{20}$ ($9$ runs each) yield $18/18$ runs below one. At the converged state of the zero-initialized runs, the non-unit eigenvalue moduli of the finite-difference Jacobian match the closed form $|\mu_\sigma|=\sqrt{1-q+q^2}$ to within $3\times10^{-10}$ ($\lambda=5,20$).

\subsection{Reduced normal forms of extensive-form games}\label{app:add-rnf}
Blotto-5 has no slope ratio because its gap is non-monotone before the floor ($R^2\approx10^{-5}$). Its exponent $\theta=2.018$ ($R^2=0.999$) is well resolved. On Kuhn, Goofspiel-4, and Blotto-10, $\theta$ equals $1.987$, $1.950$, and $1.947$. We cannot construct the reduced normal forms of Leduc and Liar's dice ($\ge6\times10^{44}$ and $\ge5\times10^{7}$ pure strategies).

\subsection{Extensive-form level spectra}\label{app:add-efg}
We run IR-PCFR$^+$ \citep{zhang2025scale} on the game tree. The state $z$ concatenates every information set's regret vector, and $\Phi(z)$ represents one full iteration. A central-difference Jacobian is evaluated at two points: the rate-fit snapshot (the first record with $0<\gapf<10^{-4}\gapf(1)$ and $t\ge50$) and the first record with $t\ge0.5T$. Here, $\rho_{\rm FD}$ is the largest modulus satisfying $|\lambda|\le1-10^{-5}$. We omit any snapshot at the float64 floor (gap $\le10^{-12}$) because it yields a Jacobian of numerical noise. Liar's dice with three faces reaches this floor at $t=44$, prior to the snapshot. Liar's dice with six faces ($49{,}140$ states) exceeds dense-Jacobian memory.

\begin{table}[h]
\centering\small\setlength{\tabcolsep}{2pt}
\caption{Non-unit spectral radius of the one-iteration Jacobian on nine extensive-form games.}
\label{tab:efgspec}
\begin{tabular*}{\textwidth}{@{\extracolsep{\fill}}lrrcccc@{}}
\toprule
game & $T$ & state dim & $\rho_{\rm FD}$ snapshot & $\rho_{\rm FD}$ late & $\max|\lambda|$ late & slope check \\
\midrule
Kuhn & $10^5$ & 24 & 0.868793 & 0.868787 & -- & rate ratio 0.989 \\
Kuhn-4 & $10^5$ & 32 & 0.943578 & 0.943564 & -- & rate ratio 1.008 \\
Leduc & $2\times10^5$ & 2604 & 0.999989 & 0.999988 & -- & observable mode 0.990 \\
Leduc-4 & $6\times10^5$ & 4592 & 0.999989 & 0.999950 & -- & observable mode 0.957 \\
six-card Kuhn & $10^5$ & 48 & 0.999983 & 0.960291 & 1.00000 & decay-rate ratio 0.996 \\
five-rank Leduc & $2\times10^5$ & 7140 & 0.999990 & 0.999985 & 1.00000 & undetermined \\
$3\times3$ Leduc & $2\times10^5$ & 5796 & 0.999987 & 0.999954 & 1.00000 & observable mode 1.064 \\
Liar's dice, 4 faces & $2\times10^3$ & 2040 & 0.999963 & 0.999968 & 1.00057 & undetermined \\
Liar's dice, 5 faces & $2\times10^4$ & 10230 & 0.999990 & 0.999990 & 1.00005 & undetermined \\
\bottomrule
\end{tabular*}

\smallskip
{\footnotesize ``late'' = probe at $0.5T$; $\max|\lambda|$ is over the full spectrum at the late probe.\par}
\end{table}``Slope check'' compares a spectral prediction with the measured gap slope. Kuhn and Kuhn-4 use the whole-horizon rate ratio. In the Leduc games, we take the ratio of the measured signal-window slope to $\log$ of the gap-carrying mode's modulus. Six-card Kuhn evaluates the decay-rate ratio $(1-\rho_{\rm FD})/(1-e^{\rm slope})$ using the late-probe radius against the fitted factor $0.960142$. ``Undetermined'' marks games whose signal window has fewer than $20$ records or whose gap-carrying mode moves across weight thresholds. We claim no rate there.

\subsection{Scope of Assumption~\ref{asm:unique-strict} and the certificate $\eps(t)$}\label{app:add-asm}
The two explicit $3\times3$ games of Section~\ref{sec:local} are:
\[
G_1=\begin{pmatrix}-0.8&1.6&-1.5\\0.3&0.1&0.2\\-0.3&0.3&2.0\end{pmatrix}\!,\quad
G_3=\begin{pmatrix}-1&3&-2\\-2&1&1\\-3&-3&-1\end{pmatrix}\!;
\]
The game $G_1$ is fully mixed, with $x^*=(0.0595,0.8839,0.0565)$, $y^*=(0.4524,0.4464,0.1012)$, $v^*=0.2006$, and $\sigmax(\Bmat)=2.715$. For $G_3$, we have $x^*=(\tfrac34,\tfrac14,0)$, $y^*=(\tfrac34,0,\tfrac14)$, and $v^*=-\tfrac54$. It exhibits a $2\times2$ support with margin $\delta=\tfrac54$ and $\sigmax(\Bmat)=2.5$. Table~\ref{tab:epscert} reports their trajectories.

We decide uniqueness by maximizing and minimizing every coordinate over the equilibrium polytope with a tolerance of $10^{-9}$. On the $216$-instance suite, $175$ games have a unique equilibrium and all $175$ satisfy strict complementarity. The remaining $41$ games have multiple equilibria. Across $30{,}000$ random integer games of sizes $2$--$4$, $14{,}302$ instances admit a unique equilibrium. Every single one of those unique equilibria is strictly complementary, matching Proposition~\ref{prop:unique-strict}.

\begin{table}[h]
\centering\small
\caption{The certificate $\eps(t)$ on the suite (left) and along explicit trajectories (right).}
\label{tab:epscert}
\begin{tabular}{@{}rccc@{}}
\toprule
$t$ & unique ($175$) & all $216$ & $\sigmax(\Bmat)>0$ ($184$) \\
\midrule
1 & 12 & 35 & 3 \\
250 & 168 & 209 & 177 \\
$10^3$ & 170 & 211 & 179 \\
$10^4$ & 173 & 214 & 182 \\
end of run & 175 & 216 & 184 \\
\bottomrule
\end{tabular}
\qquad
\begin{tabular}{@{}lcccc@{}}
\toprule
game & $\eps(1)$ & first $t$, $\eps<1$ & $\eps_\infty$ & final gap \\
\midrule
$G_1$ & 4.3983 & 200 & 0.9497 & $10^{-15}$ \\
$G_3$ & 1.7046 & 20 & 0.8464 & 0 \\
$A_5$ & 16.3853 & 50 & 0.7082 & $8\times10^{-16}$ \\
\bottomrule
\end{tabular}

\smallskip
{\footnotesize Left: number of suite instances with $\eps(t)<1$, with $\Bmat$, $x^*$, $y^*$ from the LP equilibrium.  Lemma~\ref{lem:epscert} applies to the $175$ instances with a unique equilibrium; the $41$ non-unique instances are diagnostics with the LP-selected equilibrium.  The $32$ instances with $\sigmax(\Bmat)=0$ are pure saddles, where $\Bmat=0$ and the condition is trivial.  Record $t$ is the state after $t$ completed updates.  Right: the explicit games of Section~\ref{sec:local} and $A_5$ along their zero-initialized trajectories, $T=2\times10^5$.\par}
\end{table}For $A_5$, we record $(\Nt_1,\Nt_2,\eps,\gapf)=(0.314,0.444,16.39,0.67)$ at $t=1$ and $(1.174,4.875,2.56,0.57)$ at $t=10$. Subsequent checkpoints yield $(1.302,5.251,2.34,1.60)$ at $t=20$, $(4.164,17.943,0.709,0.25)$ at $t=50$, and $(4.166,17.946,0.708,6\times10^{-6})$ at $t=500$. The orbit leaves $\rt_2=\rt_3$ at $t=10$. Over the suite, the first time with $\eps(t)<1$ has median $50$ and 90th percentile $250$.

\subsection{Larger games, random starts, and more extensive-form games}\label{app:add-more}
\begin{table}[!ht]
\centering\small
\caption{Larger games (left), random initialization (middle), and $\etainf\sigmax(\Bmat)$ by dimension (right).}
\label{tab:more}
\small
\begin{tabular}{@{}rccc@{}}
\toprule
$n$ & $\etainf\sigmax<1$ & max & $\theta$ median \\
\midrule
100 & 18/18 & 0.9594 & 2.036 \\
200 & 18/18 & 0.9595 & 2.089 \\
\bottomrule
\end{tabular}
\quad
\begin{tabular}{@{}lcc@{}}
\toprule
stratum & $\etainf\sigmax<1$ & max \\
\midrule
all runs & 1962/1962 & 0.9988 \\
scale $10^{-3}$ & 648/648 & 0.9985 \\
scale $1$ & 648/648 & 0.9988 \\
scale $10^{3}$ & 648/648 & 0.0025 \\
\bottomrule
\end{tabular}
\quad
\begin{tabular}{@{}rccc@{}}
\toprule
dim & $N$ & median & max \\
\midrule
5 & 72 & 0.580 & 0.9985 \\
20 & 72 & 0.763 & 0.9802 \\
50 & 72 & 0.825 & 0.9978 \\
100 & 18 & 0.866 & 0.9594 \\
200 & 18 & 0.849 & 0.9595 \\
\bottomrule
\end{tabular}

\smallskip
{\footnotesize Left: random matrix games at $n=100,200$ ($18$ each, Gaussian, uniform and sparse $\pm1$, raw and antisymmetrized, $T=10^6$).  Middle: random initialization $\rt_i=s\,U[0,1)^{n_i}$, three seeds per scale, $T=2\times10^5$, on the $216$ games and $A_5$, $A_{20}$.  Right: zero-initialized $\etainf\sigmax(\Bmat)$ by dimension over the suite and the $n=100,200$ games.\par}
\end{table}Of the $1962$ random-start runs, $1334$ reach a final gap $\le10^{-8}$ within the horizon. We take the spectral reading of the remaining runs at the final state. Runs at scale $10^3$ start far from saturation, yielding a small $\etainf\sigmax(\Bmat)$.

\begin{table}[h]
\centering\small
\caption{Seven extensive-form games solved with IR-PCFR$^+$.}
\label{tab:efgmore}
\begin{tabular*}{\textwidth}{@{\extracolsep{\fill}}lrrccc@{}}
\toprule
game & $T$ & states & slope ($R^2$) & $\theta$ ($R^2$) & gap, first $\to$ last \\
\midrule
six-card Kuhn & $10^5$ & 48 & $-0.0407$ (0.999) & 2.104 (0.996) & $0.85\to5.6\times10^{-17}$ \\
five-rank Leduc & $2\times10^5$ & 7140 & $-4\times10^{-5}$ (0.572) & undetermined & $4.86\to1.1\times10^{-6}$ \\
$3\times3$ Leduc & $2\times10^5$ & 5796 & $-1\times10^{-4}$ (0.855) & undetermined & $4.70\to1.6\times10^{-10}$ \\
Liar's dice, 3 faces & $2\times10^3$ & 378 & floor at $t=44$ & undetermined & $1.11\to0$ \\
Liar's dice, 4 faces & $2\times10^3$ & 2040 & $-0.0545$ (1.000) & 1.991 (0.997) & $1.31\to2.8\times10^{-17}$ \\
Liar's dice, 5 faces & $2\times10^4$ & 10230 & undetermined & undetermined & $1.44\to2.5\times10^{-8}$ \\
Liar's dice, 6 faces & $3\times10^4$ & 49140 & $-1\times10^{-4}$ (0.944) & undetermined & $1.56\to1.9\times10^{-5}$ \\
\bottomrule
\end{tabular*}

\smallskip
{\footnotesize ``slope'' is the OLS slope of $\log\gapf$ on $t$ in the linear regime; $\theta$ is reported only where the pre-registered identifiability test passes (gap span $\ge5\times$ and no two-segment break).\par}
\end{table}

\subsection{An equilibrium-free stopping rule}\label{app:add-stop}
The certificate of Theorem~\ref{thm:certificate} is used as a rule: calibrate $\gapf(t_1)$ once at the first guarded record after the snapshot, predict $\hat g(t)=\gapf(t_1)\sqrt{\Ndottot(t)/\Ndottot(t_1)}$ from the increment accumulated within each record window, and stop at the first record with $\hat g(t)\le\eps/\gamma$.  On the suite, $184$ instances reach a calibration point.  With $\gamma=2$ the measured gap at the stop is at most $\eps$ on $89.1$--$91.8\%$ of them for $\eps\in\{10^{-6},\dots,10^{-12}\}$ ($75.9$--$78.2\%$ of all $216$), with 90th percentile of $\gapf/\eps$ over the stopped runs between $0.80$ and $1.01$ and maximum up to $5\times10^{2}$; with $\gamma=4$ on $96.2$--$97.3\%$ ($81.9$--$82.9\%$ of $216$), 90th percentile at most $0.51$.  Stopping saves $98$--$99\%$ of the iterations.  No guarantee is attached to this rule; a certified finite-time version is listed in Appendix~\ref{app:scope}.


\section{Future Work}\label{app:scope}

The saturation \emph{level} remains the open quantitative link. It determines the effective step size and the spectral condition of Theorem~\ref{thm:rate}. A bound on $\Nt_{i,\infty}$ tight enough to decide $\etainf\sigmax(\Bmat)<1$ would turn the measured middle link of Remark~\ref{rem:mechanism} into a theorem. Convergence results for regret matching in potential games \citep{anagnostides2025convergence} and rate refinements for zero-sum problems \citep{anagnostides2024convergence} suggest that the mechanism carries beyond the matrix-game testbed.
\label{app:scope-saturation}
Making the absorbing radius and the constant $C_\eps$ explicit would turn the ratio certificate of Theorem~\ref{thm:certificate} into an explicit post-entry complexity bound: $T(\varepsilon)\le T_{\mathrm{reg}}+\max\{0,\lceil\log(C'/\varepsilon)/\log(1/\bar\rho)\rceil\}$. An a~priori bound on the total iteration count would also require a quantitative bound on the regime-entry time $T_{\mathrm{reg}}$. The global theory does not supply this bound.
\label{app:scope-radii}

Theorem~\ref{thm:point} establishes pointwise convergence to a single equilibrium with no uniqueness hypothesis. Extending the linear convergence rate beyond the uniqueness regime requires a natural geometric step. We replace the two-dimensional equilibrium-ray manifold with the larger fixed-point manifold over a smooth, fixed-support, strictly complementary stratum of the equilibrium set. Next, we transport the frozen-face spectral computation of Theorem~\ref{thm:rate} onto the corresponding normal bundle. This step is subject to a transverse spectral-gap condition of normal attraction.
\label{app:scope-degenerate}
At the boundary $\etainf\sigmax(\Bmat) = 1$, the relevant eigenvalues are $\mu = \pm i$ (a $1{:}4$ resonance). The fixed point is non-hyperbolic. The linearization leaves the boundary behavior undecided, though a center-manifold and resonant normal-form analysis of the $q = 1$ mode could determine it (Remark~\ref{rem:boundary}).
\label{app:scope-boundary}
Support collapse makes the one-round map non-invertible. The push-forward of Lebesgue measure from the absorbing neighborhood is singular with respect to ambient Lebesgue measure. A full-space statement for orbits entering the frozen neighborhood late must control preimages of null subsets of the active face directly (Remark~\ref{rem:fullspace}). The genericity questions for optimistic dynamics in general-sum games \citep{anagnostides2022optimistic,anagnostides2022last} are the closest analogues.
\label{app:scope-fullspace}

A finite-time version of the slope-two law of Corollary~\ref{cor:theta} would make the direction-dependent offset $C(t)$ and its $O(\norm{w^{(t)}})$ remainder explicit. This extension would supply an operational regime-detection rule (Remark~\ref{rem:gap7}). For the linearized dynamics, the per-player analogue admits a three-step-window repair (Remark~\ref{rem:single-player-window}).
\label{app:scope-middle}
Under sampled feedback, the norm increment may become noise-dominated once the deterministic motion has decayed. This behavior depends on the estimator. Action-independent noise cancels exactly in the shift and in the centering. The positive part removes noise on off-support coordinates whose perturbed pre-clipping value stays nonpositive, while larger noise can break support freezing. A theorem therefore requires a specified sampling model with nonvanishing noise energy in the measured on-support directions. Variance reduction through correlated chance sampling \citep{li2026correlated} may reduce that quantity, depending on the estimator and on the coupling between prediction and correction samples. In that setting, restricting responses under a confidence schedule for safe opponent exploitation \citep{li2026agents,li2026av} is the natural analogue of our monitor.
\label{app:scope-feedback}
The extensive-form object to be built is a per-information-set $\Bmat_j$ coupled through counterfactual weights. Natural starting points include the efficient, parallel, and abstracted CFR solvers of \citet{li2026real,li2025efficient,li2026effective,li2024rl} and the black-box equilibrium search of \citet{zhang2021finding}. Discounted regret minimization \citep{brown2019solving} and predictive Blackwell approachability \citep{farina2021faster} also offer natural starting points.
\label{app:scope-efg}

Our analysis treats simultaneous updates, whereas practice defaults to alternating updates. \citet{nguyen2026convergencealternatingregretmatching} establish convergence to a single equilibrium for alternating scale-invariant variants, and \citet{li2026alternation} study why alternation accelerates counterfactual regret minimization. Whether norm saturation yields a closed-form rate under alternation remains unresolved.

The ratio law relates an internal quantity to the Nash gap in zero-sum games. For general-sum games, \citet{li2026we} study how to certify proximity to equilibrium. An increment-based analogue in that setting is open.

In large extensive-form games, algorithms compute equilibria in subgames at decision time \citep{li2026safe,liturborebel} using solvers compiled for parallel hardware \citep{li2026gpu}. A last-iterate rate for the scale-invariant update would remove the averaged strategy from these pipelines.
\label{app:scope-adjacent}
\section{Discussion: Questions, Comparisons, and Open Problems}\label{app:discussion}

This appendix addresses the questions left open by the main text and compares our results with prior last-iterate convergence guarantees. Remaining open problems conclude the discussion.

\paragraph{Is the step-size condition necessary?}
No. Theorem~\ref{thm:point} establishes convergence to a single equilibrium without this hypothesis. Obtaining the linear envelope of Theorem~\ref{thm:rate} requires the condition. The family $A_\lambda$ admits an orbit that converges to a limit violating the condition (Appendix~\ref{app:add-alambda}).

\paragraph{Can the step-size condition be checked?}
Given the equilibrium, yes. The monotone quantity $\eps(t)$ of Lemma~\ref{lem:epscert} certifies the condition at the first recorded time at which it falls below one. Without prior knowledge of the equilibrium, we obtain no such certificate. Every measurable limit on our testbed satisfies the condition (Table~\ref{tab:main}).

\paragraph{Why does the realized limit satisfy the condition?}
We do not establish that it does. Proposition~\ref{prop:selfstab} shows that limits violating the condition are unstable. Obtaining a bound on the saturation level sharp enough to decide the condition remains open (Appendix~\ref{app:scope}).

\paragraph{What is new relative to the closest works?}
Saturation and the vanishing gap follow directly from the path-length bound of \citet{zhang2025scale} (Remark~\ref{rem:attrib}). To prove convergence to a single equilibrium, we adapt the analysis of \citet{nguyen2026convergencealternatingregretmatching} from alternating to simultaneous extragradient updates (Remark~\ref{rem:single-ne}). Our new contributions begin with one-step support identification and the closed-form Jacobian of the normalized and clipped map. We also establish the convergence rate. Finally, we prove the instability of violating limits and derive the ratio law.

\paragraph{Why a ratio law if the gap can be computed?}
In a matrix game, the payoff vectors observed at each round yield the gap directly. The ratio law then explains why the norm increment decays at twice the exponential rate of the gap. Monitoring convergence becomes its primary role in extensive-form games. Evaluating the gap there requires a best-response pass, and our evidence for that regime is empirical (Section~\ref{sec:experiments}).

\paragraph{Is the ratio law verified numerically?}
The interval statistics in Section~\ref{sec:experiments} serve as a proxy. Our matched-time test evaluates the quantities of Theorem~\ref{thm:certificate} at the exact same round (Appendix~\ref{app:exp-protocols}, verification~9).

\paragraph{What happens when the equilibrium is not unique?}
Theorem~\ref{thm:point} still guarantees convergence to a single equilibrium. Our rate theorem does not apply here. Nevertheless, on the games with non-unique equilibria that we evaluated, the closed-form spectral radius matches the finite-difference Jacobian (Table~\ref{tab:rnf}).

\paragraph{Does the behavior depend on the zero initialization?}
Our theorems assume zero initialization, as scale invariance requires (Remark~\ref{rem:init}). Table~\ref{tab:more} evaluates random initializations empirically.

\paragraph{Comparison with the existing theory.}
Our theorems in Sections~\ref{sec:main}--\ref{sec:certificate} depart from the existing last-iterate theory for regret matching. Three aspects set our results apart: the underlying assumptions, the guarantees we establish, and the computability of the quantities involved. Four comparisons structure our discussion, followed by open problems.

\paragraph{The dynamics are not modified.}
Prior linear last-iterate theorems for regret-matching-type dynamics \citep{cai2025lastiterate,meng2025lastiterate,pmlr-v139-perolat21a,liu2023the} restart the iteration or alter the update. The unmodified algorithms were open. Both RM$^+$ and PRM$^+$ can oscillate forever \citep{farina2023regret}. Theorem~\ref{thm:rate} covers \iregprm{} exactly as \citet{zhang2025scale} run it, with no restart and no extra step. Their bounded second-order path length already forces saturation of the regret norm. Theorem~\ref{thm:saturation} restates that bound with an explicit constant, and our local theory builds on the resulting limit. The concurrent parameter-free variant of \citet{meng2026a} also exploits a nondecreasing regret norm to obtain an $O(1/T)$ guarantee for the \emph{average} iterate. In our analysis, the limit of the norm acts as a step size. We obtain a linear guarantee for the \emph{last} iterate.

\paragraph{The rate is a closed form, not an existential constant.}
Linear last-iterate convergence of extragradient and optimistic methods on polytopes \citep{korpelevich1976extragradient,daskalakis2018limit,pmlr-v108-mokhtari20a} rests on saddle-point metric subregularity \citep{wei2021linear} or on a projection-residual error bound \citep{tseng1995linear}. These routes give a linear rate with an existential constant or require vertex enumeration \citep[Appendix~G]{wei2021linear}. Theorem~\ref{thm:rate} governs the rate envelope through the transverse spectral radius $\rho^*$ in \eqref{eq:rhostar}. This quantity is an explicit function of the nonzero singular values of the value-centered support matrix $\Bmat$ and of the saturated step size $\etainf$. Both depend only on the game and the two limiting norms. Game-dependent convergence matches the lower bound of \citet{cai2024fast}. Equation~\eqref{eq:rhostar} characterizes this dependence explicitly.

\paragraph{The hypotheses separate global from local.}
Theorem~\ref{thm:global} requires no spectral, strict-complementarity, or uniqueness condition. Its gap conclusion holds for every matrix game. Under the same absence of hypotheses, Theorem~\ref{thm:point} yields convergence to a single equilibrium. Only the rate and the certificate use Assumptions~\ref{asm:unique-strict} and~\ref{asm:stepsize}. Finite-time identification of the active face followed by local linear convergence has precedents outside games \citep{burke1988identification,hare2004identifying,liang2017activity}. Proximal identification arguments do not transfer to our setting because the state-to-strategy map of \iregprm{} is not a proximal operator. The associated operator satisfies a Minty condition \citep{cai2025lastiterate,anagnostides2025polynomial,zhang2025expected,pethick2025efficient}. Here, we obtain identification from strict complementarity and the continuity of the prediction and correction steps at the equilibrium ray. Lemma~\ref{lem:absorb} yields this property in one step.

\paragraph{The ratio law is two-sided and on the last iterate.}
The gap-to-regret identity \citep{zinkevich2007regret} and utility-increment certificates \citep{burch2019revisiting} yield a one-sided bound on the duality gap of the \emph{averaged} iterate. This bound equals the sum of the two cumulative regrets divided by $T$. The normalized duality gap of \citet{applegate2023faster} tracks the current iterate through a dedicated gap computation. It measures the distance to the entire solution set and remains finite at degenerate equilibria. In Theorem~\ref{thm:certificate}, we establish a two-sided bound on ratios of the last-iterate gap via ratios of the squared-norm increments $\Ndottot$ that the algorithm already computes. The constant $\kappa$ is independent of the time indices and avoids the factor $1/(1-\rho^*)$ incurred by a rate-based argument. Our bound takes the form of an increment ratio, analogous to the step ratio in the certificate of \citet{applegate2023faster}. In a matrix game, the observed payoff vectors directly yield the exact last-iterate gap at no additional cost. The theorem characterizes an internal quantity of the algorithm in that case. In extensive-form games, computing the gap requires a best-response pass. Our theorem does not encompass that setting, and our experimental monitor in Section~\ref{sec:experiments} is strictly empirical. Computing the constant $\kappa$ online is impossible. The ratio law also relies on a two-dimensional fixed-point cone and degenerates alongside it at non-unique equilibria. On that subset of instances, we bracket the deviation empirically. We obtain a median $\kappa_{\mathrm{emp}}$ of $1.80$ compared to $1.63$ on regular instances, with values below $2$ on $7$ out of $8$ instances (Table~\ref{tab:numerics}).

\paragraph{From saturation to convergence, and the constant.}
The bound of Theorem~\ref{thm:saturation} is explicit. On the testbed, it sits three to six orders of magnitude above the realized $\Nt_{i,\infty}$. This bound leaves the \emph{value} of the saturation level undetermined. That value selects the effective step size $\etainf$, determining whether Assumption~\ref{asm:stepsize} holds. Between the global theorem (gap $\to0$ on every game) and the local one (linear rate under the strict spectral condition) sits a quantitative theory of where the norm stops. The boundary $\etainf\sigmax(\Bmat)=1$ and the exceptional set of Proposition~\ref{prop:selfstab} remain open.

\paragraph{Quantitative neighborhood radii.}
The absorbing radius around the frozen face and the constant $C_\eps$ in the rate envelope lack closed forms. They depend on the equilibrium margin and other instance quantities through non-explicit constants. Finite-time control of the $\theta=2$ law of Corollary~\ref{cor:theta} would upgrade our ratio certificate from a two-sided online bound to an explicit post-entry complexity bound. Bounding the total iteration count also requires a quantitative regime-entry time. In Corollary~\ref{cor:theta}, the constants $C(t)$ bracket the ratio by an existential $\kappa$. Meanwhile, the first-order benchmark from Section~\ref{sec:experiments} exceeds the measured deviation by eight orders of magnitude. Appendix~\ref{app:scope-radii} records the technical route to making $L_{\gapf}$, $c_-$, and the spectral range explicit.

\paragraph{Degenerate and non-unique equilibria.}
Theorem~\ref{thm:point} applies to these instances. Namely, the strategies converge to a single equilibrium, and when $t_+<\infty$, the state converges to the corresponding scaled pair at the saturation levels. On $\mathcal Z$, the fixed-point set coincides with the product of the two positive-scale cones over the optimal-strategy polytopes, by virtue of Proposition~\ref{prop:fixedpoint}. Next, when $t_+<\infty$, Theorem~\ref{thm:global} establishes that the $\omega$-limit set is contained in this set at the two limiting norms. Consequently, the equilibrium potential of Lemma~\ref{lem:potential} implies that it is a singleton. Our experimental evaluation corroborates this behavior. In particular, we consider the $8$ instances flagged as degenerate by the unit-band count of the finite-difference Jacobian, of which $5$ are confirmed via the kernel dimension of $\Bmat$ on the converged support (Appendix~\ref{app:experiments}, verification~6). On these instances, the state $z=(\rt_1,\rt_2)$ changes by at most $5\times10^{-10}$ and the strategy pair changes by at most $4\times10^{-12}$ from the first checkpoint with duality gap below $10^{-15}$ until the termination of the run, namely $10^{5}$ to $10^{6}$ rounds later. Characterizing the convergence rate is a separate question. Theorem~\ref{thm:rate} requires uniqueness of the equilibrium, whereas on degenerate instances we only observe empirically that the closed-form $\rho_2$ continues to match the finite-difference Jacobian (Section~\ref{sec:experiments}). Consequently, two regimes remain unresolved, namely establishing the convergence rate under non-unique equilibria, and analyzing the non-hyperbolic boundary $\etainf\sigmax(\Bmat)=1$, near which $18$ of the $184$ measurable instances satisfy $\etainf\sigmax(\Bmat)>0.95$.
For the rate, the natural extension is a normally hyperbolic invariant manifold argument
along a smooth stratum of that set, transporting the frozen-face spectral
computation onto the normal bundle
(Appendix~\ref{app:scope-degenerate}); the boundary case itself needs a
center-manifold analysis of the $q=1$ mode
(Appendix~\ref{app:scope-boundary}).

\paragraph{Extensive-form games and partial feedback.}
The certificate's payoff is largest in extensive-form games. There, $\Nt(t)$ aggregates per-infoset regret norms from quantities the update already produces, while the gap still requires a best-response pass. Extending the frozen-face analysis to counterfactual-value couplings remains open. Establishing the certificate under sampled feedback is open as well, because the norm increment can become noise-dominated once the deterministic motion has decayed. Both problems have concrete starting points: a per-information-set $\Bmat_j$ coupled through counterfactual weights (Appendix~\ref{app:scope-efg}), and variance reduction through correlated chance sampling (Appendix~\ref{app:scope-feedback}).

\subsection{Theorem ledger}
\label{app:ledger}

\textbf{Hypotheses at a glance.}
Theorem~\ref{thm:saturation} establishes saturation. The strategy conclusions of Theorems~\ref{thm:global} and~\ref{thm:point} therefore hold unconditionally. On every matrix game, the gap vanishes and the last strategy pair converges to a single Nash equilibrium. Lemma~\ref{lem:neverpos} treats the case $t_+=\infty$. Characterizing the $\omega$-limit set and the limit of the regret state requires the finite-time hypothesis $t_+<\infty$. The linear rate and the certificate demand strictly stronger conditions. Both assume a unique strictly complementary equilibrium with $m\ge2$ and the post-saturation step-size condition (Assumptions~\ref{asm:unique-strict} and~\ref{asm:stepsize}). Table~\ref{tab:ledger} records each claim with its hypotheses, epistemic status, and verification.

\begin{table}[!ht]
  \centering
  \caption{Ledger of headline claims.}
  \label{tab:ledger}
  \small
  \begin{tabular*}{\textwidth}{@{\extracolsep{\fill}}>{\raggedright\arraybackslash}p{0.26\textwidth}>{\raggedright\arraybackslash}p{0.18\textwidth}>{\raggedright\arraybackslash}p{0.27\textwidth}>{\raggedright\arraybackslash}p{0.19\textwidth}@{}}
    \toprule
    Claim & Hypotheses & Status & Ver. \\
    \midrule
    Norm saturation, explicit bound
      (Thm.~\ref{thm:saturation}) & none &
      proved & n/a (bound loose by $7\times10^{2}$--$9\times10^{5}$, Rem.~\ref{rem:sat-bound-loose}) \\
    Gap $\to 0$, $\omega$-limit in $\Fix$
      (Thm.~\ref{thm:global}) & $t_+<\infty$ (i); none (ii) &
      proved; strategy-level gap also from \citet{zhang2025scale} (Rem.~\ref{rem:attrib}) & n/a (216 inst.\ converge) \\
    Convergence to a single NE
      (Thm.~\ref{thm:point}) & none (state limit for $t_+<\infty$; else Lem.~\ref{lem:neverpos}) &
      proved & n/a \\
    Assumption~\ref{asm:unique-strict} $\Leftrightarrow$ unique NE
      (Prop.~\ref{prop:unique-strict}) & none &
      proved (Goldman--Tucker) & 175/175 unique instances strictly complementary \\
    Monotone certificate $\eps(t)$ for A2
      (Lem.~\ref{lem:epscert}) & A1, $t_+<\infty$, equilibrium known &
      proved (one-sided: a crossing certifies A2) & $168/175$ unique inst.\ by $t=250$, $175/175$ at the end \\
    One-step support freezing
      (Lem.~\ref{lem:absorb}) & strict complementarity, local &
      proved & 1 ($99.5\%$) \\
    Closed-form spectrum $J = I + L + L^2$
      (Lem.~\ref{lem:jacobian},~\ref{lem:centering}) & A3, frozen face &
      proved & 3 ($32$ inst.) \\
    Linear rate envelope $(\rho^* + \eps)^t$
      (Thm.~\ref{thm:rate}) & A1 + A2, local entry &
      proved (upper envelope; realized rate $= \rho^*$ is empirical,
      \S\ref{rem:rate-semantics}) & 2 (FD Jacobian, $155$ inst.); 3 (closed form, $32$) \\
    Instability of limits with $\etainf\sigmax(\Bmat)>1$
      (Prop.~\ref{prop:selfstab}) & A3 (strict compl., nondegenerate), $\etainf\sigmax(\Bmat)>1$ &
      proved (thin trapped local basins; ``observed limits satisfy A2'' is a heuristic, Rem.~\ref{rem:mechanism}) & 5 ($184$, $0$
      violations) \\
    Ratio certificate $\kappa$-sandwich
      (Thm.~\ref{thm:certificate}) & A1 + A2, frozen
      regime & proved (pointwise) & 6 (interval proxy, $124/124$ within a first-order benchmark); 9 (matched-time) \\
    Slope-two law $\theta = 2$
      (Cor.~\ref{cor:theta}) & A1 + A2, frozen regime &
      proved (asymptotic in window width) & 4 ($96.1\%$) \\
    $\Ndottot =$ on-support energy
      (Lem.~\ref{lem:ndot-identity}) & frozen regime &
      proved (identity) & 4, 6 (proxy: pattern not logged) \\
    Self-stabilization narrative (norm grows until A2 holds) &
      n/a & heuristic (instability proved; growth and the zero-initialized
      orbit not; Rem.~\ref{rem:mechanism}) & 5 (consistent) \\
    Degenerate-instance sandwich & outside A1 &
      empirical only ($\kappa$ assembled for $3$ of $8$) & App.~\ref{app:experiments} \\
    EFG monitoring (Kuhn, Leduc) & outside matrix theory &
      empirical only & 8 \\
    \bottomrule
  \end{tabular*}

\smallskip
{\footnotesize A1 is Assumption~\ref{asm:unique-strict} (unique strictly complementary equilibrium), A2 is Assumption~\ref{asm:stepsize} ($\etainf\sigmax(\Bmat)<1$), A3 is Assumption~\ref{asm:nondeg-strict} (strictly complementary, non-degenerate equilibrium ray; equivalent to A1 by Lemma~\ref{lem:kernel}). ``Ver.''\ points to the numbered verification of Appendix~\ref{app:exp-protocols}.\par}
\end{table}

\end{document}